\documentclass[a4paper,USenglish,cleveref,pdfa]{lipics-v2021}

\hideLIPIcs

\usepackage{todonotes}

\usepackage{multicol}

\usepackage{mathtools}

\DeclarePairedDelimiter{\braces}{\{}{\}}		    
\DeclarePairedDelimiter{\parens}{(}{)}		        
\DeclarePairedDelimiter{\ceil}{\lceil}{\rceil}		
\DeclarePairedDelimiter{\floor}{\lfloor}{\rfloor}	
\DeclarePairedDelimiterX{\setdef}[2]{\{}{\}}{#1 \mid #2}		
\DeclarePairedDelimiterXPP{\exclude}[1]{\mathopen{}\setminus}{\{}{\}}{}{#1}

\usepackage{tikz}

\usetikzlibrary{decorations,decorations.pathmorphing,decorations.pathreplacing,fit,positioning,arrows,calc,backgrounds,topaths}

\tikzstyle{vertex}=[circle, draw, inner sep=2pt,  minimum width=1 pt, minimum size=0.1cm]
\tikzstyle{black_vertex}=[circle, draw, inner sep=2pt, fill=black!100, minimum width=1pt, minimum size=0.1cm]
\tikzstyle{gray_vertex}=[circle, draw, inner sep=2pt, fill=lightgray, minimum width=1pt, minimum size=0.1cm]

\usepackage{boxedminipage}
\usepackage[normalem]{ulem}
\usepackage{framed}
\usepackage{xspace}
\usepackage{xifthen}
\usepackage{tabularx}

\newcommand{\bO}{\mathcal{O}}

\newcommand{\sO}{\bO^{\star}}

\newcommand{\paramstyle}[1]{\text{\rm #1}}
\newcommand{\tw}{\paramstyle{tw}}
\newcommand{\pw}{\paramstyle{pw}}
\newcommand{\td}{\paramstyle{td}}
\newcommand{\vc}{\paramstyle{vc}}
\newcommand{\fvs}{\paramstyle{fvs}}
\newcommand{\cw}{\paramstyle{cw}}
\newcommand{\ms}{\paramstyle{ms}}

\newcommand{\chid}{\chi_{\rm d}}

\newcommand{\BDD}{{\sc Bounded Degree Vertex Deletion}}

\newcommand{\SD}{{\sc(3,4)-XSAT}}

\newcommand{\DC}{{\sc Defective Coloring}}

\newcommand{\kMC}{{\sc $k$-Multicolored Clique}}
\newcommand{\IC}{{\sc Improper Coloring}}

\newcommand{\Z}{\mathbb{Z}}
\newcommand{\N}{\mathbb{N}}

\newlength{\RoundedBoxWidth}
\newsavebox{\GrayRoundedBox}
\newenvironment{GrayBox}[1]{
        \setlength{\RoundedBoxWidth}{.93\textwidth}
        \def\boxheading{#1}
        \begin{lrbox}{\GrayRoundedBox}
        \begin{minipage}{\RoundedBoxWidth}
    }{
        \end{minipage}
        \end{lrbox}
        \begin{center}
            \begin{tikzpicture}%
               \node(Text)[draw=black!20,fill=white,rounded corners,%
                     inner sep=2ex,text width=\RoundedBoxWidth]%
                     {\usebox{\GrayRoundedBox}};
                \coordinate(x) at (current bounding box.north west);
                \node [draw=white,rectangle,inner sep=3pt,anchor=north west,fill=white]
                at ($(x)+(6pt,.75em)$) {\boxheading};
            \end{tikzpicture}
        \end{center}
    }

\newenvironment{defproblemx}[2][]
    {
        \noindent\ignorespaces%
        \FrameSep=6pt%
        \parindent=0pt%
        \vspace*{-1.5em}
        \ifthenelse{\isempty{#1}}{%
            \begin{GrayBox}{\textsc{#2}}%
            }{%
              \begin{GrayBox}{\textsc{#2} parameterized by~{#1}}%
            }
        \begin{tabular*}{\textwidth}{@{\hspace{.1em}} >{\itshape} p{1.8cm} p{0.8\textwidth} @{}}%
    }{
        \end{tabular*}%
        \end{GrayBox}%
        \ignorespacesafterend
    }

\newcommand{\defproblema}[3]{%
    \begin{defproblemx}{#1}
    {\bf Instance:}  & #2 \\
    {\bf Goal:} & #3
    \end{defproblemx}
}%

\newcommand{\problemdef}[3]{\defproblema{#1}{#2}{#3}}

\title{Structural Parameterizations for Two Bounded Degree Problems
Revisited\footnote{An extended abstract of this work has been presented at the
31st Annual European Symposium on Algorithms, ESA 2023~\cite{esa/LampisV23}.}}

\titlerunning{Structural Parameterizations for Two Bounded Degree Problems Revisited}

\author{Michael Lampis}{Universit\'{e} Paris-Dauphine, PSL University, CNRS UMR7243, LAMSADE, Paris, France}{michail.lampis@dauphine.fr}{https://orcid.org/0000-0002-5791-0887}{}
\author{Manolis Vasilakis}{Universit\'{e} Paris-Dauphine, PSL University, CNRS UMR7243, LAMSADE, Paris, France}{emmanouil.vasilakis@dauphine.eu}{https://orcid.org/0000-0001-6505-2977}{}

\authorrunning{M. Lampis and M. Vasilakis}

\Copyright{Michael Lampis and Manolis Vasilakis}



\funding{This work is partially supported by ANR projects ANR-21-CE48-0022 (S-EX-AP-PE-AL) and ANR-18-CE40-0025-01 (ASSK).}

\nolinenumbers

\EventEditors{Inge Li G{\o}rtz, Martin Farach-Colton, Simon J. Puglisi, and Grzegorz Herman}
\EventNoEds{4}
\EventLongTitle{31st Annual European Symposium on Algorithms (ESA 2023)}
\EventShortTitle{ESA 2023}
\EventAcronym{ESA}
\EventYear{2023}
\EventDate{September 4--6, 2023}
\EventLocation{Amsterdam, the Netherlands}
\EventLogo{}
\SeriesVolume{274}
\ArticleNo{49}

\begin{document}

\maketitle

\begin{abstract}
We revisit two well-studied problems, \textsc{Bounded Degree Vertex Deletion} and
\textsc{Defective Coloring}, where the input is a graph $G$ and a target degree
$\Delta$ and we are asked either to edit or partition the graph so that the
maximum degree becomes bounded by $\Delta$. Both problems are known to be
parameterized intractable for the most well-known structural parameters, such
as treewidth. 

We revisit the parameterization by treewidth, as well as several related
parameters and present a more fine-grained picture of the complexity of both
problems.
In particular:
\begin{itemize}

\item Both problems admit straightforward DP algorithms with table sizes
$(\Delta+2)^\mathrm{tw}$ and $(\chi_\mathrm{d}(\Delta+1))^{\mathrm{tw}}$ respectively, where tw is
the input graph's treewidth and $\chi_\mathrm{d}$ the number of available colors. We
show that, under the SETH, both algorithms are essentially optimal, for any
non-trivial fixed values of $\Delta, \chi_\mathrm{d}$, even if we replace treewidth by
pathwidth. Along the way, we obtain an algorithm for \textsc{Defective
Coloring} with complexity quasi-linear in the table size, thus settling the
complexity of both problems for treewidth and pathwidth. 

\item Given that the standard DP algorithm is optimal for treewidth and
pathwidth, we then go on to consider the more restricted parameter tree-depth.
Here, previously known lower bounds imply that, under the ETH, \textsc{Bounded
Vertex Degree Deletion} and \textsc{Defective Coloring} cannot be solved in time
$n^{o(\sqrt[4]{\mathrm{td}})}$ and $n^{o(\sqrt{\mathrm{td}})}$ respectively, leaving some hope
that a qualitatively faster algorithm than the one for treewidth may be
possible. We close this gap by showing that neither problem can be solved in
time $n^{o(\mathrm{td})}$, under the ETH, by employing a recursive low tree-depth
construction that may be of independent interest.

\item Finally, we consider a structural parameter that is known to be
restrictive enough to render both problems FPT: vertex cover. For both problems
the best known algorithm in this setting has a super-exponential dependence of
the form $\mathrm{vc}^{\mathcal{O}(\mathrm{vc})}$. We show that this is optimal, as an algorithm with
dependence of the form $\mathrm{vc}^{o(\mathrm{vc})}$ would violate the ETH. Our proof relies
on a new application of the technique of $d$-detecting families introduced by
Bonamy et al.~[ToCT 2019].
    
\end{itemize}

Our results, although mostly negative in nature, paint a clear picture
regarding the complexity of both problems in the landscape of parameterized
complexity, since in all cases we provide essentially matching upper and lower
bounds.

\end{abstract}

\newpage

\section{Introduction}\label{sec:introduction}

Parameterized complexity and in particular the study of structural parameters
such as treewidth is one of the most well-developed approaches for dealing with
NP-hard problems on graphs. Treewidth is of course one of the major success
stories of this field, as a plethora of hard problems become fixed-parameter
tractable when parameterized by this parameter.  Naturally, this success has
motivated the effort to trace the limits of the algorithmic power of treewidth
by attempting to understand what are the problems for which treewidth-based
techniques \emph{cannot} work.

When could the treewidth toolbox fail? One common scenario that seems to be
shared by a multitude of problems which are W[1]-hard\footnote{We assume the
reader is familiar with the basics of parameterized complexity theory, as given
in standard textbooks~\cite{books/CyganFKLMPPS15}.} parameterized by treewidth
is when a natural dynamic programming algorithm does exist, but the DP is
forced to store for each vertex of a bag in the tree decomposition an
arbitrarily large integer -- for example a number related to the degree of the
vertex. Our goal in this paper is to study situations of this type and pose the
natural question of whether one can do better than the ``obvious'' DP, by
obtaining an algorithm with better running time, even at the expense of looking
at a parameter more restrictive than treewidth. 

Given the above, we focus on two problems which are arguably among the most
natural representatives of our scenario: \BDD\ and \DC. In both problems the
input is a graph $G$ and a target degree $\Delta$ and we are asked, in the case
of \BDD\ to delete a minimum number of vertices so that the remaining graph has
degree at most $\Delta$, and in the case of \DC\ to partition $G$ into a
minimum number of color classes such that each class induces a graph of degree
at most $\Delta$. Both problems are well-studied, as they generalize classical
problems (\textsc{Vertex Cover} and \textsc{Coloring} respectively) and we
review some of the previous work below. However, the most relevant aspect of
the two problems for our purposes is the following: (i) both problems admit DP
algorithms with complexity of the form $n^{\bO(\tw)}$ and (ii) both
problems are W[1]-hard parameterized by treewidth; in fact, for \DC,
it is even known that assuming the ETH it cannot be solved in time
$n^{o(\tw)}$~\cite{siamdm/BelmonteLM20}.

Since the $n^{\bO(\tw)}$ algorithms follow from standard DP techniques, it
becomes a natural question whether we can do better. Does a better algorithm
exist? Realistically, one could hope for one of two things: either an algorithm
which still handles the problem parameterized by treewidth and in view of the
aforementioned lower bound only attempts a fine-grained improvement in the
running time; or an algorithm which is qualitatively faster at the expense of
using a more restricted parameter. The results of this paper give strong
negative evidence for both questions: if we parameterize by treewidth (and even
by pathwidth) the running time of the standard DP is optimal under the SETH
even for all fixed values of the other relevant parameters ($\Delta$ and the
number of colors $\chid$); while if we parameterize by more restrictive
parameters, such as tree-depth and vertex cover, we obtain lower bound results
(under the ETH) which indicate that the best algorithm is still essentially to
run a form of the standard treewidth DP, even in these much more restricted
cases. Our results thus paint a complete picture of the structurally
parameterized complexity of these two problems and indicate that the standard
DP is optimal in a multitude of restricted cases.

\subparagraph*{Our contribution in more detail.} Following standard techniques,
the two problems admit DP algorithms with tables of sizes $(\Delta+2)^{\tw}$
and $(\chid (\Delta+1))^{\tw}$ respectively. Our first result is a collection
of reductions proving that, assuming the SETH, no algorithm can improve upon
these dynamic programs, even for pathwidth.  More precisely, we show that no
algorithm can solve \BDD\ and \DC\ in time
$(\Delta+2-\varepsilon)^{\pw}n^{\bO(1)}$ and
$(\chid(\Delta+1)-\varepsilon)^{\pw}n^{\bO(1)}$ respectively, for any
$\varepsilon>0$ and for any combination of fixed values of $\Delta, \chid$
(except the combination $\Delta=0$ and $\chid=2$, which trivially makes \DC\
polynomial-time solvable).  Our reductions follow the general strategy pioneered by
Lokshtanov, Marx, and Saurabh~\cite{talg/LokshtanovMS18} and indeed generalize
their results for \textsc{Vertex Cover} and \textsc{Coloring} (which already
covered the case $\Delta=0$). The main difficulty here is being able to cover
all values of the secondary parameters and for technical reasons we are forced
to give separate versions of our reductions to cover the case $\Delta=1$ for
both problems. Along the way we note that, even though an algorithm with
complexity $(\chid \Delta)^{\bO(\tw)}n^{\bO(1)}$ was given for \DC\
in~\cite{siamdm/BelmonteLM20}, it was not known if an algorithm with complexity
$(\chid(\Delta+1))^{\tw}n^{\bO(1)}$ (that is, with a quasi-linear dependence on
the table size) is possible.  For completeness, we settle this by providing an
algorithm of this running time, using the FFT technique proposed by Cygan and
Pilipczuk~\cite{tcs/CyganP10}.  Taking also into account the \BDD\ algorithm of
running time $(\Delta+2)^{\tw}n^{\bO(1)}$ given by van Rooij~\cite{csr/Rooij21},
we have exactly matching upper and lower bounds for both problems, for both
treewidth and pathwidth.

Given that the results above show rather conclusively that the standard DP is
the best algorithm for parameters treewidth and pathwidth, we then move on to a
more restricted case: tree-depth. We recall that graphs of tree-depth $k$ are a
proper subclass of graphs of pathwidth $k$, therefore one could reasonably hope
to obtain a better algorithm for this parameter. This hope may further be
supported by the fact that known lower bounds do not match the complexity of
the standard algorithm. More precisely, the W[1]-hardness reduction given for
\BDD\ parameterized by tree-depth by Ganian, Klute, and
Ordyniak~\cite{algorithmica/GanianKO21} has a quartic blow-up, thus only
implying that no $n^{o(\sqrt[4]{\td})}$ algorithm is possible; while the
reduction given for \DC\ in~\cite{siamdm/BelmonteLM20} has a quadratic blow-up,
only implying that no $n^{o(\sqrt{\td})}$ algorithm is possible (in both cases
under the ETH). Our contribution is to show that both reductions can be
replaced by more efficient reductions which are \emph{linear} in the parameter;
we thus establish that neither problem can be solved in time $n^{o(\td)}$,
implying that the treewidth-based algorithm remains (qualitatively) optimal
even in this restricted case. One interesting aspect of our reductions is that,
rather than using a modulator to a low tree-depth graph, which is common in
such reductions, we use a recursive construction that leverages the full power
of the parameter and may be of further use in tightening other lower bounds for
the parameter tree-depth.

Finally, we move on to a more special case, parameterizing both problems by
vertex cover. Both problems are FPT for this parameter and, since vertex cover
is very restrictive as a parameter, one would hope that, finally, we should be
able to obtain an algorithm that is more clever than the treewidth-based DP.
Somewhat disappointingly, the known FPT algorithms for both problems have
complexity $\vc^{\bO(\vc)}n^{\bO(1)}$~\cite{siamdm/BelmonteLM20}, and the
super-exponential dependence on the parameter is due to the fact that both
algorithms are simple win/win arguments which, in one case, just execute the
standard treewidth DP.  We show that this is justified, as neither problem can
be solved in time $\vc^{o(\vc)}n^{\bO(1)}$ (under the ETH), meaning that the
algorithm that blindly executes the treewidth-based DP in some cases is still
(qualitatively) best possible. We obtain our result by applying the technique
of $d$-detecting families, introduced by
Bonamy, Kowalik, Pilipczuk, Socala, and Wrochna~\cite{toct/BonamyKPSW19}.
Our results indicate that parameterization by
vertex cover is a domain where this promising, but currently under-used,
technique may find more applications in parameterized complexity.

\subparagraph*{Related work.} Both \BDD\ and \DC\ are well-studied problems with
a rich literature. \BDD\ finds application in a multitude of areas, ranging
from computational biology~\cite{jcss/FellowsGMN11} to some related problems in
voting theory~\cite{algorithmica/BetzlerBBNU14,tcs/BetzlerU09}, and its dual
problem, called $s$-\textsc{Plex Detection}, has numerous applications in
social network
analysis~\cite{ior/BalasundaramBH11,jco/McCloskyH12,jco/MoserNS12}.  Various
approximation algorithms are
known~\cite{dam/Fujito98,ciac/Fujito17,ipl/OkunB03}.
The problem has also been extensively studied under the scope of parameterized
complexity.  It is W[2]-hard for unbounded values of $\Delta$ and parameter
$k$ (the value of the optimal)~\cite{jcss/FellowsGMN11}, while it admits a
linear-size kernel parameterized by $k$~\cite{jcss/FellowsGMN11,jcss/Xiao17},
for any fixed $\Delta \geq 0$; numerous FPT algorithms have been presented in
the latter setting~\cite{jco/MoserNS12,dam/NishimuraRT05,cocoon/Xiao16}. FPT
approximation algorithms were given for \BDD\ in~\cite{icalp/Lampis14}
and~\cite{soda/LokshtanovMRSZ21}. As for \DC, which also appears in the
literature as \IC, it was introduced almost 40 years
ago~\cite{andrews1985generalization,jgt/CowenCW86}.  The main motivation behind
this problem comes from the field of telecommunications, where the colors
correspond to available frequencies and the goal is to assign them to
communication nodes; a small amount of interference between neighboring nodes
may be tolerable, which is modeled by the parameter $\Delta$.  There have been
plenty of works on the problem
(see~\cite{jgaa/AngeliniBLD0KMR17,jgt/Archdeacon87,siamdm/BelmonteLM20,dmtcs/BelmonteLM22,jgt/ChoiE19,networks/HavetKS09}
and the references therein),
especially on unit disk graphs and various classes of grids.

The previous work for both problems that is most relevant to us focuses on
their parameterized complexity for structural parameters, such as treewidth. In
this setting, \BDD\ was one of the first problems to be discovered to be
W[1]-hard parameterized by treewidth~\cite{dam/BetzlerBNU12}, though the
problem does become FPT parameterized by $\tw+\Delta$ or $\tw+k$. This hardness
result was more recently improved by Ganian, Klute, and
Ordyniak~\cite{algorithmica/GanianKO21}, who showed that \BDD\ is W[1]-hard
parameterized by tree-depth and feedback vertex set. \DC\ was shown to be
W[1]-hard parameterized by tree-depth (and hence pathwidth and treewidth)
in~\cite{siamdm/BelmonteLM20}. However,~\cite{siamdm/BelmonteLM20} gave a
hardness reduction for pathwidth that is linear in the parameter, and hence
implies a $n^{o(\pw)}$ lower bound for \DC\ under the ETH, but a hardness
reduction for tree-depth that is quadratic (implying only a $n^{o(\sqrt{\td})}$
lower bound). Similarly, the reduction given by~\cite{algorithmica/GanianKO21}
for \BDD\ parameterized by tree-depth is quartic in the parameter, as it goes
through an intermediate problem (a variant of \textsc{Subset Sum}), implying
only a $n^{o(\sqrt[4]{\td})}$ lower bound. \DC\ is known to be FPT
parameterized by vertex cover using a simple win/win argument which applies the
treewidth-based DP in one case (if $\Delta>\vc$, then the graph is always
$2$-colorable; otherwise the standard DP algorithm runs in FPT time), and the
same is true for \BDD\
(if $\Delta \leq \vc$, we can use the aforementioned FPT algorithm for parameters
$\tw + \Delta$,
else assume that $k<\vc$, as otherwise the problem is trivial,
follow the reduction of~\cite{dam/BetzlerBNU12} to
\textsc{Vector Dominating Set} and notice that at most $\vc$ vertices have degree greater than $\Delta$,
thus \textsc{Vector Dominating Set} can be decided in time $\vc^{\bO(\vc)} n^{\bO(1)}$
due to~\cite{cocoa/RamanSS08}).
Hence, the best algorithms for both problems for this parameter have
complexity $\vc^{\bO(\vc)}n^{\bO(1)}$.

The fine-grained analysis of the complexity of structural parameterizations,
such as by treewidth, is an active field of research. The technique of using
the SETH to establish tight running time lower bounds was pioneered by
Lokshtanov, Marx, and Saurabh~\cite{talg/LokshtanovMS18}. Since then, tight
upper and lower bounds are known for a multitude of problems for
parameterizations by treewidth and related parameters, such as pathwidth and
clique-width~\cite{jacm/CyganKN18,talg/CyganNPPRW22,dmtcs/DubloisLP21,tcs/DubloisLP22,talg/FockeMR22,talg/GanianHKOS22,stacs/GroenlandMNS22,algorithmica/HanakaKLOS20,dam/JaffkeJ23,siamcomp/OkrasaR21,jgaa/GeffenJKM20}. One difficulty of the results we present here is
that we need to present a family of reductions: one for each fixed value of
$\Delta$ and $\chid$. There are a few other problems for which families of
tight lower bounds are known, such as $k$-\textsc{Coloring}, for which the
correct dependence is $k^{\tw}$ for treewidth~\cite{talg/LokshtanovMS18} and
$(2^k-2)^{\cw}$ for clique-width~\cite{siamdm/Lampis20} for all $k\ge 3$;
distance $r$-\textsc{Dominating Set}, for which the correct dependence is
$(2r+1)^{\tw}$~\cite{iwpec/BorradaileL16} and $(3r+1)^{\cw}$~\cite{dam/KatsikarelisLP19}, for all $r\ge 1$; and distance
$d$-\textsc{Independent Set}, for which the correct dependence is $d^{\tw}$~\cite{dam/KatsikarelisLP22}. In all these cases, the optimal algorithm is the ``natural'' DP, and
our results for \BDD\ and \DC\ fit this pattern. 

Even though the previous work mentioned above may make it seem that our
SETH-based lower bounds are not surprising, it is important to stress that it
is not a given that the naïve DP should be optimal for our problems. In
particular, \BDD\ falls into a general category of $(\sigma,\rho)$-domination
problems, which were studied recently in~\cite{soda/FockeMINSSW23} (we refer
the reader there for the definition of $(\sigma,\rho)$-domination).  One of the
main results of that work was to show that significant improvements over the
basic DP are indeed possible in some cases, and in particular when one of
$\sigma,\rho$ is cofinite.  Since \BDD\ is the case where
$\sigma=\{0,\ldots,\Delta\}$ and $\rho = \N$ (that is, $\rho$ is
co-finite), our result falls exactly in the territory left uncharted
by~\cite{soda/FockeMINSSW23}, where more efficient algorithms could still be
found (and where indeed~\cite{soda/FockeMINSSW23} did uncover such algorithms
for some values of $\sigma,\rho$).

\begin{table}[t]
    \caption{Lower bounds established in the current work.
    The results of the first row are under SETH, while all the rest under ETH.}
    \label{table:lb_overview}%
    \begin{tabular}{l|c|c}%
        Parameter & \BDD & \DC\\%
        \hline%
        pathwidth + $\Delta$ & $\sO((\Delta + 2 - \varepsilon)^\pw)$ & $\sO((\chid \cdot (\Delta + 1) - \varepsilon)^\pw)$\\%
        treedepth & $n^{o(\td)}$ & $n^{o(\td)}$\\%
        vertex cover & $\vc^{o(\vc)} n^{\bO(1)}$ & $\vc^{o(\vc)} n^{\bO(1)}$\\%
    \end{tabular}%
\end{table}%


\section{Preliminaries}\label{sec:preliminaries}
Throughout the paper we use standard graph notation~\cite{Diestel17},
and we assume familiarity with the basic notions of parameterized complexity~\cite{books/CyganFKLMPPS15}.
All graphs considered are undirected without loops, unless explicitly stated otherwise.
For a graph $G = (V, E)$ and two integers $\chid \geq 1, \Delta \geq 0$,
we say that $G$ admits a $(\chid, \Delta)$-coloring if one can partition $V$ into $\chid$ sets
such that the graph induced by each set has maximum degree at most $\Delta$.
In that case, \DC{} is the problem of deciding, given $G, \chid, \Delta$, whether $G$ admits a $(\chid, \Delta)$-coloring.
Let $\N$ denote the set of non-negative integers.
For $x, y \in \Z$, let $[x, y] = \setdef{z \in \Z}{x \leq z \leq y}$,
while $[x] = [1,x]$.
Standard $\sO$ notation is used to suppress polynomial factors.
For the pathwidth bounds, we use the notion of \emph{mixed search strategy}~\cite{tcs/TakahashiUK95},
where an edge is cleared by either placing a searcher on both of its endpoints or sliding one along the edge.
We rely on a weaker form of the ETH, which states that 3-SAT on instances with $n$ variables and $m$ clauses
cannot be solved in time $2^{o(n+m)}$.

In \kMC, we are given a graph $G = (V, E)$ and a partition of $V$ into $k$ independent sets $V_1, \ldots, V_k$, each of size $n$,
and we are asked to determine whether $G$ contains a $k$-clique.
It is well-known that this problem does not admit any $f(k) n^{o(k)}$ algorithm,
where $f$ is any computable function,
unless the ETH is false~\cite{books/CyganFKLMPPS15}.

In $q$-CSP-$B$, we are given a \textsc{Constraint Satisfaction} (CSP) instance with $n$ variables and $m$ constraints.
The variables take values in a set $Y$ of size $B$, i.e.~$|Y| = B$.
Each constraint involves at most $q$ variables and is given as a list of satisfying assignments for these variables,
where a satisfying assignment is a $q$-tuple of values from the set $Y$ given to each of the $q$ variables.
The following result was shown by Lampis~\cite{siamdm/Lampis20} to be a natural consequence of the SETH,
and has been used in the past for various hardness results~\cite{dmtcs/DubloisLP21,tcs/DubloisLP22,algorithmica/HanakaKLOS20}.

\begin{theorem}[\cite{siamdm/Lampis20}]\label{thm:q_CSP_B_SETH}
    For any $B \geq 2$ it holds that,
    if the SETH is true,
    then for all $\varepsilon > 0$,
    there exists a $q$ such that $n$-variable {\sc $q$-CSP-$B$}
    cannot be solved in time $\sO \parens*{(B - \varepsilon)^n}$.
\end{theorem}

\section{Treewidth and Maximum Degree}\label{sec:tw_plus_Delta}

In this section we present tight lower bounds on the complexity of solving both
\BDD{} and \DC{} parameterized by the treewidth of the input graph plus the target degree.
In the case of \BDD, this lower bound matches the algorithm of~\cite{soda/FockeMINSSW23},
while for \DC{} we develop an algorithm of matching running time in~\cref{sec:dc_tw_algo}.

Both reductions are similar in nature:
we start from an instance $\phi$ of $q$-CSP-$B$,
and produce an equivalent instance on a graph of pathwidth $\pw = n + \bO(1)$,
where $n$ denotes the number of variables of $\phi$.
An interesting observation however, is that for both problems
we have to distinguish between the case where $\Delta = 1$ and $\Delta \geq 2$;
the whole construction becomes much more complicated in the latter case.

\subsection{Bounded Degree Vertex Deletion}\label{subsec:bdd_tw_lb}

In the following, we will present a reduction from $q$-CSP-$B$ to \BDD,
for any fixed $\Delta \geq 1$, where $\Delta = B-2$.
In that case, if there exists a $\sO ((\Delta + 2 - \varepsilon)^\pw)$ algorithm for \BDD,
where $\varepsilon > 0$,
then there exists a $\sO ((B - \varepsilon)^n)$ algorithm for $q$-CSP-$B$,
for any constant $q$,
which due to~\cref{thm:q_CSP_B_SETH} results in SETH failing.

Our reduction is based on the construction of ``long paths'' of \emph{Block gadgets},
that are serially connected in a path-like manner.
Each such ``path'' corresponds to a variable of the given CSP,
while each column of the construction is associated with one of its constraints.
Intuitively, our aim is to embed the $B^n$ possible variable assignments into the $(\Delta + 2)^\tw$ states of some optimal
dynamic program that would solve the problem on our constructed instance.

Below, we present a sequence of gadgets used in our reduction.
The aforementioned block gadgets,
which allow a solution to choose among $\Delta + 2$ reasonable choices,
are the main ingredient.
Notice that these gadgets will differ significantly depending on whether $\Delta$ is equal to $1$ or not.
We connect these gadgets in a path-like manner that ensures that choices remain consistent throughout the construction,
and connect constraint gadgets in different ``columns'' of the constructed grid in a way that allows us to verify if the choice made represents a satisfying assignment,
without significantly increasing the graph's pathwidth.

\begin{theorem}
    For any constant $\varepsilon > 0$,
    there is no $\sO((3 - \varepsilon)^{\pw})$ algorithm
    deciding \BDD{} for $\Delta = 1$,
    where \pw{} denotes the input graph's pathwidth,
    unless the SETH is false.
\end{theorem}

\begin{proof}
    Fix some positive $\varepsilon > 0$ for which we want to prove the theorem.
    We will reduce $q$-CSP-$3$, for some $q$ that is a constant that only depends on $\varepsilon$,
    to \BDD{} for $\Delta = 1$ in a way that ensures that if the resulting \BDD{} instance could be solved in time $\sO((3-\varepsilon)^{\pw})$,
    then we would obtain an algorithm for $q$-CSP-$3$ that would contradict the SETH due to~\cref{thm:q_CSP_B_SETH}.
    To this end, let $\phi$ be an instance of $q$-CSP-$3$ of $n$ variables $X = \setdef{x_i}{i \in [n]}$ taking values over the set $Y = [3]$
    and $m$ constraints $C = \setdef{c_j}{j \in [m]}$.
    For each constraint we are given a set of at most $q$ variables which are involved in this constraint and a list of satisfying assignments for these variables,
    the size of which is denoted by $s : C \rightarrow [3^q]$,
    i.e.~$s(c_j) \leq 3^q = \bO(1)$ denotes the number of satisfying assignments for constraint $c_j$.
    We will construct in polynomial time an equivalent instance $\mathcal{I} = (G,k)$ of \BDD{} for $\Delta = 1$,
    where $\pw(G) \leq n + \bO (1)$.

    \proofsubparagraph*{Block and Variable Gadgets.}
    For every variable $x_i$ and every constraint $c_j$, construct a path of $3$ vertices $p^1_{i,j}, p^2_{i,j}$ and $p^3_{i,j}$,
    which comprises the \emph{block gadget} $\hat{B}_{i,j}$.
    Intuitively, we will map the deletion of $p^y_{i,j}$ with an assignment where $x_i$ receives value $y$.
    Next, for $j \in [m-1]$, we add an
    edge between $p^3_{i,j}$ and $p^1_{i,j+1}$, thus resulting in $n$ paths $P_1,
    \ldots, P_n$ of length $3m$, called \emph{variable gadgets}.
    
    \begin{figure}[ht]
    \centering
    \begin{tikzpicture}[scale=0.8, transform shape]
    
    \node[] () at (0,5.5) {$P_1$};
    \node[] () at (0,3.5) {$P_2$};
    \node[] () at (0,2.5) {$\vdots$};
    \node[] () at (0,1.5) {$P_n$};
    
    \node[vertex] (vl11) at (0.5,5.5) {};
    \node[] () at (0.7,5.2) {$p^1_{1,1}$};
    \node[vertex] (vl21) at (0.5,3.5) {};
    \node[vertex] (vln1) at (0.5,1.5) {};
    
    \node[vertex] (vl12) at (2.5,5.5) {};
    \node[] () at (2.7,5.2) {$p^2_{1,1}$};
    \node[vertex] (vl22) at (2.5,3.5) {};
    \node[vertex] (vln2) at (2.5,1.5) {};
    
    \node[vertex] (vl13) at (4.5,5.5) {};
    \node[] () at (4.3,5.2) {$p^3_{1,1}$};
    \node[vertex] (vl23) at (4.5,3.5) {};
    \node[vertex] (vln3) at (4.5,1.5) {};
    
    \node[vertex] (vl14) at (6.5,5.5) {};
    \node[] () at (6.7,5.2) {$p^1_{1,2}$};
    \node[vertex] (vl24) at (6.5,3.5) {};
    \node[vertex] (vln4) at (6.5,1.5) {};
    
    \node[vertex] (vl15) at (8.5,5.5) {};
    \node[] () at (8.7,5.2) {$p^2_{1,2}$};
    \node[vertex] (vl25) at (8.5,3.5) {};
    \node[vertex] (vln5) at (8.5,1.5) {};
    
    \node[vertex] (vl16) at (10.5,5.5) {};
    \node[] () at (10.3,5.2) {$p^3_{1,2}$};
    \node[vertex] (vl26) at (10.5,3.5) {};
    \node[vertex] (vln6) at (10.5,1.5) {};
    
    \node[] (vl17) at (12.5,5.5) {$\ldots$};
    \node[] (vl27) at (12.5,3.5) {$\ldots$};
    \node[] (vln7) at (12.5,1.5) {$\ldots$};
    
    
    \draw[dashed] (0.25,1.3) rectangle (4.75,5.8);
    \draw[dashed] (6.25,1.3) rectangle (10.75,5.8);

    
    \draw[] (vl11)--(vl12)--(vl13)--(vl14)--(vl15)--(vl16)--(vl17);
    \draw[] (vl21)--(vl22)--(vl23)--(vl24)--(vl25)--(vl26)--(vl27);
    \draw[] (vln1)--(vln2)--(vln3)--(vln4)--(vln5)--(vln6)--(vln7);
    
    \end{tikzpicture}
    \caption{Sequences of block gadgets comprise the variable gadgets.}
    \label{fig:bdd_tw_lb_D0_block_and_variable_gadgets}
    \end{figure}
    
    \proofsubparagraph*{Constraint Gadget.}
    This gadget is responsible for determining constraint satisfaction,
    based on the choices made in the rest of the graph.
    For constraint $c_j$, construct the \emph{constraint gadget} $\hat{C}_j$ as follows:
    \begin{itemize}
        \item construct a clique of $s(c_j)$ vertices $v^j_1, \ldots, v^j_{s(c_j)}$,
        and fix an arbitrary one-to-one mapping between those vertices and the satisfying assignments of $c_j$,
    
        \item attach to each vertex $v^j_\ell$ a leaf $l^j_\ell$,
    
        \item if variable $x_i$ is involved in the constraint $c_j$ and $v^j_\ell$ corresponds to an assignment where $x_i$ has value $y \in Y$,
        add an edge between $v^j_\ell$ and $p^y_{i,j}$.
    \end{itemize}
    
    Let graph $G_0$ be the graph containing all variable gadgets $P_i$ as well as
    all the constraint gadgets $\hat{C}_j$, for $i \in [n]$ and $j \in [m]$.
    To construct graph $G$,
    introduce $\kappa = 2n + 1$ copies $G_1, \ldots, G_\kappa$ of $G_0$, such that
    they are connected sequentially as follows: for $i \in [n]$ and $j \in [\kappa - 1]$,
    add an edge between $p^3_{i,m}(G_j)$ and $p^1_{i,1}(G_{j+1})$,
    where $p^y_{i,j}(G_z)$ denotes the vertex $p^y_{i,j}$ of graph $G_z$.
    We refer to the block gadget $\hat{B}_{i,j}$, to the variable gadget $P_i$, and to the
    constraint gadget $\hat{C}_j$ of $G_z$ as $\hat{B}^{G_z}_{i,j}$, $P^{G_z}_i$,
    and $\hat{C}^{G_z}_j$ respectively.
    Let $\mathcal{P}^i$ denote the path resulting from $P^{G_1}_i, \ldots, P^{G_\kappa}_i$.
    Set $k = \kappa \cdot k'$,
    where $k' = \sum_{j=1}^m (s(c_j) - 1 + n) = m \cdot n + \sum_{j=1}^m (s(c_j) - 1)$,
    and let $\mathcal{I} = (G,k)$ denote the instance of \BDD{} for $\Delta = 1$.

    \begin{lemma}\label{lem:bdd_tw_lb_D0_cor1}
        If $\phi$ is satisfiable,
        then there exists $S \subseteq V(G)$ such that $G - S$ has maximum degree at most $1$ and $|S| \leq k$.
    \end{lemma}
    
    \begin{proof}
        Let $f : X \rightarrow Y$ denote an assignment which satisfies all the constraints $c_1, \ldots, c_m$.
        In that case, let $p^y_{i,j}(G_z) \in S$, where $j \in [m]$ and $z \in [\kappa]$, if $f(x_i) = y$.
        In other words, from every path $\mathcal{P}^i$, include in $S$ either the first, the second,
        or the third vertex out of every block gadget $\hat{B}_{i,j}$, depending on the value of $f(x_i)$.
        Moreover, since $f$ is a satisfying assignment, for every constraint $c_j$
        there exists a satisfying assignment which is a restriction of $f$ to the at most
        $q$ involved variables of $c_j$ and
        corresponds to a vertex $v^j_\ell(G_z)$ of $\hat{C}^{G_z}_j$.
        Let $v^j_{\ell'}(G_z) \in S$ if $\ell' \neq \ell$.
        It holds that $|S| = \kappa \cdot (m n + \sum_{i=1}^m (s(c_j)-1)) = k$.
    
        It remains to prove that $G - S$ has maximum degree at most $1$.
        Consider the constraint gadget $\hat{C}^{G_z}_j$.
        Any leaf $l^j_i$ has degree either $0$ or $1$ in $G-S$.
        Let $v^j_\ell (G_z) \notin S$ denote the vertex of the clique of the constraint gadget which does not belong to $S$.
        It holds that $\deg_{G-S}(v^j_\ell (G_z)) = 1$, since
        \begin{itemize}
            \item $l^j_\ell (G_z) \notin S$,
            \item $v^j_{\ell'} (G_z) \in S$, for $\ell' \neq \ell$,
            \item $v^j_\ell (G_z)$ has an edge with $p^y_{i,j}(G_z)$, where $y \in Y$,
            only if $x_i$ is involved in $c_j$ and $v^j_\ell (G_z)$ corresponds to an assignment $g$
            where variable $x_i$ has value $y$.
            However, since this assignment is a restriction of $f$, it follows that $g(x_i) = f(x_i) = y$,
            thus $p^y_{i,j}(G_z) \in S$.
        \end{itemize}
        Consequently, there are no edges between vertices of the constraint and the block gadgets in $G-S$.
        Notice that from every three consecutive vertices in $\mathcal{P}^i$, one belongs to $S$.
        Therefore, it holds that if $p^y_{i,j}(G_z) \notin S$, at most one of its neighbors does not belong to $S$,
        thus its degree is at most $1$ in $G-S$.
    \end{proof}
    
    \begin{lemma}\label{lem:bdd_tw_lb_D0_cor2}
        If there exists $S \subseteq V(G)$ such that $G - S$ has maximum degree at most $1$ and $|S| \leq k$,
        then $\phi$ is satisfiable.
    \end{lemma}
    
    \begin{proof}
        Let $S \subseteq V(G)$ such that $G - S$ has maximum degree at most $1$ and $|S| \leq k$.
        First we will prove that $S$ contains a single vertex $p^y_{i,j}$ from each block gadget $\hat{B}_{i,j}$,
        as well as vertices $v^j_l$, where $l \in [s(c_j)] \setminus \braces{\ell}$ for some $\ell \in [s(c_j)]$ from each constraint gadget $\hat{C}_j$.
        
        Notice that $S$ contains at least $1$ vertex from every block gadget $\hat{B}_{i,j}$,
        since otherwise the vertex $p^2_{i,j}$ has degree at least $2$ in $G-S$.
        Additionally, $S$ contains at least $s(c_j) - 1$ vertices from every constraint gadget $\hat{C}_j$,
        since each such gadget has a clique of size $s(c_j)$, every vertex of which has a leaf attached.    
        Therefore, it holds that $|S \cap V(G_z)| \geq k' = m \cdot n + \sum_{j=1}^m (s(c_j) - 1)$,
        for all $z \in [\kappa]$.
        Since $|S| \leq \kappa \cdot k'$, it follows that $|S \cap V(G_z)| = k'$.
        Consequently, $S$ contains exactly one vertex per block gadget and
        exactly $s(c_j) - 1$ vertices per constraint gadget $\hat{C}_j$, which can only be vertices of the clique.
    
        Notice that then it holds that $|S \cap V(G_z)| = k'$, for all $z \in [\kappa]$.
        We will say that an \emph{inconsistency} occurs in a variable gadget $P^{G_z}_i$
        if there exist two consecutive block gadgets $\hat{B}^{G_z}_{i,j}, \hat{B}^{G_z}_{i,j+1}$ such that
        $p^y_{i,j}, p^{y'}_{i,j+1} \in S$, for $y \neq y'$.
        Moreover, we will say that $G_z$ is \emph{consistent} if no inconsistency occurs in any of the variable gadgets $P^{G_z}_i$, for $i \in [n]$.
    
        \begin{claim}
            There exists $\pi \in [\kappa]$ such that $G_\pi$ is consistent.
        \end{claim}
    
        \begin{claimproof}
            Notice that $S$ contains $\kappa \cdot m$ vertices from every path $\mathcal{P}^i$,
            since the latter is comprised of that many block gadgets.
            We will prove that every path $\mathcal{P}^i$ may induce at most $2$ inconsistencies.
            In that case, since there are $n$ such paths and $\kappa = 2n+1$ copies of $G_0$,
            due to the pigeonhole principle there exists some $G_\pi$ with no inconsistencies.
        
            Consider a path $\mathcal{P}^i$ as well as a block gadget $\hat{B}^{G_z}_{i,j}$,
            for some $z \in [\kappa]$ and $j \in [m]$.
            Let $N(\hat{B}^{G_z}_{i,j})$ denote the block gadget right of $\hat{B}^{G_z}_{i,j}$, consisting of vertices $n_1 (\hat{B}^{G_z}_{i,j})$, $n_2 (\hat{B}^{G_z}_{i,j})$ and $n_3 (\hat{B}^{G_z}_{i,j})$,
            where $p^3_{i,j}$ and $n_1 (\hat{B}^{G_z}_{i,j})$ are connected via an edge in $G$.
            Moreover, let $\hat{B}^{G_{z'}}_{i,j'}$, where either a) $z' = z$ and $j' > j$ or
            b) $z' > z$ and $j' \in [m]$,
            denote some block gadget which appears to the right of $\hat{B}^{G_z}_{i,j}$.
            In that case:
            \begin{itemize}
                \item If $p^1_{i,j}(G_z) \in S$, then $p^1_{i,j'} (G_{z'}) \in S$.
                It holds for $N(\hat{B}^{G_z}_{i,j})$, since $p^3_{i,j} (G_z)$ has degree $2$ in $G-S$ otherwise.
                Suppose that it holds for all block gadgets after $\hat{B}^{G_z}_{i,j}$ and up to some block gadget $\hat{B}^{G_{z^*}}_{i,j^*}$.
                Then it holds for $N(\hat{B}^{G_{z^*}}_{i,j^*})$ as well, since $p^3_{i,j^*} (G_{z^*})$ has degree $2$ in $G-S$ otherwise.
            
                \item If $p^2_{i,j}(G_z) \in S$, then, either $p^1_{i,j'} (G_{z'}) \in S$ or $p^2_{i,j'} (G_{z'}) \in S$.
                It holds for $N(\hat{B}^{G_z}_{i,j})$, since $n_3 (\hat{B}^{G_z}_{i,j}) \in S$ otherwise and it follows that $n_1 (\hat{B}^{G_z}_{i,j})$ has degree $2$ in $G-S$.
                Suppose that it holds for all block gadgets after $\hat{B}^{G_z}_{i,j}$ and up to some block gadget $\hat{B}^{G_{z^*}}_{i,j^*}$.
                Then it holds for $N(\hat{B}^{G_{z^*}}_{i,j^*})$ as well, since $n_3 (\hat{B}^{G_{z^*}}_{i,j^*}) \in S$ otherwise and
                it follows that $n_1 (\hat{B}^{G_{z^*}}_{i,j^*})$ has degree $2$ in $G-S$.
            
                \item Lastly, if $p^3_{i,j}(G_z) \in S$, then
                $S$ contains a single vertex from every subsequent block gadget $\hat{B}^{G_{z'}}_{i,j'}$,
                since it contains exactly one vertex per block gadget.
            \end{itemize}
            Thus, it follows that every path can induce at most $2$ inconsistencies, and since there is a total of $n$ paths,
            there exists a copy $G_\pi$ which is consistent.
        \end{claimproof}
    
        Now, consider the assignment $f : X \rightarrow Y$, where $f(x_i) = y$ if $p^y_{i,j}(G_\pi) \in S$.
        This is a valid assignment, since $G_\pi$ is consistent and a single vertex is contained in $S$ from each block gadget.
        We will prove that it satisfies all the constraints.
        Consider a constraint $c_j$.
        Regarding the constraint gadget $\hat{C}^{G_\pi}_j$, it holds that $S$ includes exactly $s(c_j) - 1$ vertices,
        none of which is a leaf vertex.
        Let $v^j_\ell (G_\pi) \notin S$ be the only vertex which is part of the clique and does not belong to $S$.
        Since $l^j_\ell (G_\pi) \notin S$, it follows that every neighbor of $v^j_\ell (G_\pi)$ in the block gadgets $\hat{B}^{G_\pi}_{i,j}$ belongs to $S$.
        In that case, the satisfying assignment corresponding to $v^j_\ell (G_\pi)$ is a restriction of $f$, thus $f$ satisfies $c_j$.
        Since this holds for any $j$, $\phi$ is satisfied.
    \end{proof}
    
    \begin{lemma}\label{lem:bdd_tw_lb_D0_pw}
        It holds that $\pw(G) \leq n + \bO(1)$.
    \end{lemma}
    
    \begin{proof}
        We will prove the statement by providing a mixed search strategy to clean $G$ using at most this many searchers simultaneously.
        Since for the mixed search number \ms{} it holds that $\pw(G) \leq \ms(G) \leq \pw(G) + 1$,
        we will show that $\ms(G) \leq n + 2 \cdot 3^q$ and the statement will follow.
    
        Start with graph $G_1$.
        Place $2s(c_1)$ searchers to all the vertices of $\hat{C}^{G_1}_1$,
        as well as $n$ searchers on vertices $p^1_{i,1} (G_1)$, for $i \in [n]$.
        In this way, all the edges of the constraint gadget are cleared.
        Next we will describe the procedure to clear $\hat{B}^{G_1}_{i,1}$.
        Move the searcher along the edge connecting $p^1_{i,1} (G_1)$ to $p^2_{i,1} (G_1)$ and then along the edge connecting $p^2_{i,1} (G_1)$ and $p^3_{i,1} (G_1)$.
        Repeat the whole process for all $i$, thus clearing all block gadgets $\hat{B}^{G_1}_{i,1}$ as well as the edges between those block gadgets
        and $\hat{C}^{G_1}_1$.
    
        In order to clear the rest of the graph, we first move the searchers from $\hat{C}^{G_z}_j$ to $\hat{C}^{G_z}_{j+1}$ if $j < m$
        or to $\hat{C}^{G_{z+1}}_1$ alternatively (possibly introducing new searchers if required),
        and then proceed by clearing the corresponding block gadgets, by first sliding all searchers from $p^3_{i,j} (G_z)$
        to $p^1_{i,j+1} (G_z)$ if $j < m$ or to $p^1_{i,1} (G_{z+1})$ alternatively.
        By repeating this procedure, in the end we clear all the edges of $G$ by using at most $n + 2 \cdot 3^q = n + \bO (1)$ searchers.
    \end{proof}
    
    Therefore, in polynomial time we can construct a graph $G$,
    of pathwidth $\pw(G) \leq n + \bO(1)$ due to~\cref{lem:bdd_tw_lb_D0_pw},
    such that, due to~\cref{lem:bdd_tw_lb_D0_cor1,lem:bdd_tw_lb_D0_cor2},
    deciding whether there exists $S \subseteq V(G)$ of size $|S| \leq k$ and $G - S$ has maximum degree at most $1$
    is equivalent to deciding whether $\phi$ is satisfiable.
    In that case, assuming there exists a $\sO((3-\varepsilon)^{\pw(G)})$ algorithm for \BDD{} for $\Delta = 1$,
    one could decide $q$-CSP-$3$ in time $\sO((3-\varepsilon)^{\pw(G)}) = \sO((3-\varepsilon)^{n + \bO(1)}) = \sO((3-\varepsilon)^n)$
    for any constant $q$,
    which contradicts the SETH due to~\cref{thm:q_CSP_B_SETH}.
\end{proof}

\begin{theorem}
    For any constant $\varepsilon > 0$,
    there is no $\sO((\Delta + 2 - \varepsilon)^{\pw})$ algorithm
    deciding \BDD{} for $\Delta \geq 2$,
    where \pw{} denotes the input graph's pathwidth,
    unless the SETH is false.
\end{theorem}

\begin{proof}
    Fix some positive $\varepsilon > 0$ for which we want to prove the theorem.
    Let $B \geq 4$.
    We will reduce $q$-CSP-$B$, for some $q$ that is a constant that only depends on $\varepsilon$,
    to \BDD{} for $\Delta = B - 2$ in a way that ensures that if the resulting \BDD{} instance could be solved in time $\sO((\Delta + 2 - \varepsilon)^{\pw})$,
    then we would obtain an algorithm for $q$-CSP-$B$ that would contradict the SETH due to~\cref{thm:q_CSP_B_SETH}.
    To this end, let $\phi$ be an instance of $q$-CSP-$B$ of $n$ variables $X = \setdef{x_i}{i \in [n]}$ taking values over the set $Y = [0, \Delta + 1]$
    and $m$ constraints $C = \setdef{c_j}{j \in [m]}$, where $\Delta = B - 2$.
    For each constraint we are given a set of at most $q$ variables which are involved in this constraint and a list of satisfying assignments for these variables,
    the size of which is denoted by $s : C \rightarrow [B^q]$,
    i.e.~$s(c_j) \leq B^q = \bO(1)$ denotes the number of satisfying assignments for constraint $c_j$.
    We will construct in polynomial time an equivalent instance $\mathcal{I} = (G,k)$ of \BDD{} for $\Delta = B - 2$,
    where $\pw(G) \leq n + \bO (1)$.
    
    \proofsubparagraph*{Block and Variable Gadgets.}
    For every variable $x_i$ and every constraint $c_j$,
    construct a \emph{block gadget} $\hat{B}_{i,j}$, as depicted in~\cref{fig:bdd_tw_lb_D1_block_gadget}.
    In order to do so, we introduce vertices $a, a', b, \chi_i, y_i$, and $q_i$, for $i \in [\Delta]$.
    Then, we attach $\Delta$ leaves on $b$ and $\Delta - 1$ leaves on vertices $q_i$.
    Finally, we add edges $\braces{a,b}, \braces{a',b}, \braces{a,\chi_i}, \braces{a',y_i}, \braces{\chi_i,q_i}$,
    and $\braces{y_i, q_i}$.
    Intuitively, we map the deletion of $b$ as well as of $p$ vertices out of $\setdef{\chi_i}{i \in [\Delta]}$ and $\Delta - p$ vertices out of $\setdef{y_i}{i \in [\Delta]}$ with
    an assignment where $x_i$ receives value $p \in [0, \Delta]$,
    while the deletion of $a$ maps with an assignment where $x_i$ receives value $\Delta+1$.
    Next, for $j \in [m-1]$, we serially connect the block gadgets $\hat{B}_{i,j}$ and $\hat{B}_{i,j+1}$ so that the vertex $a'$ of $\hat{B}_{i,j}$ is the vertex $a$ of $\hat{B}_{i,j+1}$,
    thus resulting in $n$ ``paths'' $P_1, \ldots, P_n$ consisting of $m$ serially connected block gadgets, called \emph{variable gadgets}.
    For an illustration see~\cref{fig:bdd_tw_Delta_block}.
    
    \begin{figure}[ht]
    \centering 
      \begin{subfigure}[b]{0.4\linewidth}
      \centering
        \begin{tikzpicture}[scale=0.75, transform shape]
        
        
        \node[vertex] (a) at (0,5) {};
        \node[] () at (0,5.3) {$a$};
        
        \node[vertex] (x1) at (1.5,6) {};
        \node[] () at (1.5,6.3) {$\chi_1$};
        
        \node[vertex] (xD) at (1.5,4) {};
        \node[] () at (1.5,4.3) {$\chi_\Delta$};
        
        \node[] (dots1) at (1.5,5.1) {$\vdots$};
        
        \node[gray_vertex] (q1) at (3,6) {};
        \node[] () at (3,6.3) {$q_1$};
        \node[vertex] (y1) at (4.5,6) {};
        \node[] () at (4.5,6.3) {$y_1$};
        
        \node[gray_vertex] (qD) at (3,4) {};
        \node[] () at (3,4.3) {$q_\Delta$};
        \node[vertex] (yD) at (4.5,4) {};
        \node[] () at (4.5,4.3) {$y_\Delta$};
        
        \node[] (dots2) at (4.5,5.1) {$\vdots$};
        
        \node[vertex] (a') at (6,5) {};
        \node[] () at (6,5.3) {$a'$};
        
        \node[black_vertex] (b) at (3,2.5) {};
        \node[] () at (3,2.8) {$b$};

        
        \draw[] (a)--(x1)--(q1)--(y1)--(a');
        \draw[] (a)--(xD)--(qD)--(yD)--(a');
        \draw[] (a) edge [bend right] (b);
        \draw[] (a') edge [bend left] (b);
            
        \end{tikzpicture}
        \caption{Block gadget when $\Delta \geq 2$.}
        \label{fig:bdd_tw_lb_D1_block_gadget}
      \end{subfigure}
    \begin{subfigure}[b]{0.4\linewidth}
    \centering
        \begin{tikzpicture}[scale=0.6, transform shape]
        
        
        \node[vertex] (a) at (0,5) {};
        
        \node[vertex] (x1) at (1.5,6) {};
        
        \node[vertex] (xD) at (1.5,4) {};
        
        \node[] (dots1) at (1.5,5.1) {$\vdots$};
        
        \node[gray_vertex] (q1) at (3,6) {};
        \node[vertex] (y1) at (4.5,6) {};
        
        \node[gray_vertex] (qD) at (3,4) {};
        \node[vertex] (yD) at (4.5,4) {};
        
        \node[] (dots2) at (4.5,5.1) {$\vdots$};
            
        \node[vertex] (a') at (6,5) {};
        
        \node[black_vertex] (b) at (3,2.5) {};

        \node[vertex] (x1n) at (7.5,6) {};
        
        \node[vertex] (xDn) at (7.5,4) {};
        
        \node[] (dots1n) at (7.5,5.1) {$\vdots$};
        
        \node[gray_vertex] (q1n) at (9,6) {};
        \node[vertex] (y1n) at (10.5,6) {};
        
        \node[gray_vertex] (qDn) at (9,4) {};
        \node[vertex] (yDn) at (10.5,4) {};
        
        \node[] (dots2n) at (10.5,5.1) {$\vdots$};

        \node[vertex] (a'n) at (12,5) {};
        
        \node[black_vertex] (bn) at (9,2.5) {};

        
        \draw[] (a)--(x1)--(q1)--(y1)--(a');
        \draw[] (a)--(xD)--(qD)--(yD)--(a');
        \draw[] (a) edge [bend right] (b);
        \draw[] (a') edge [bend left] (b);
    
        \draw[] (a')--(x1n)--(q1n)--(y1n)--(a'n);
        \draw[] (a')--(xDn)--(qDn)--(yDn)--(a'n);
        \draw[] (a') edge [bend right] (bn);
        \draw[] (a'n) edge [bend left] (bn);
    
        \end{tikzpicture}
        \caption{Serially connected block gadgets.}
      \end{subfigure}
    \caption{Gray vertices have $\Delta - 1$ and black vertices $\Delta$ leaves attached.}
    \label{fig:bdd_tw_Delta_block}
    \end{figure}  
    
    \proofsubparagraph*{Constraint Gadget.}
    This gadget is responsible for determining constraint satisfaction,
    based on the choices made in the rest of the graph.
    For constraint $c_j$, construct the \emph{constraint gadget} $\hat{C}_j$ as follows:
    \begin{itemize}
        \item construct a clique of $s(c_j)$ vertices $v^j_1, \ldots, v^j_{s(c_j)}$,
        and fix an arbitrary one-to-one mapping between those vertices and the satisfying assignments of $c_j$,
    
        \item attach to each vertex of the clique $\Delta$ unnamed leaves,
    
        \item if variable $x_i$ is involved in the constraint $c_j$ and $v^j_\ell$ corresponds to an assignment where $x_i$ has value $p \in Y$,
        consider two cases:
        \begin{romanenumerate}
            \item If $p \in [0, \Delta]$,
            then add an edge between $v^j_\ell$ and vertex $b$ of $\hat{B}_{i,j}$.
            Moreover, add edges between $v^j_\ell$ and $\setdef{\chi_i}{i \in [p]} \cup \setdef{y_i}{i \in [p+1, \Delta]}$.
            
            \item On the other hand, if $v^j_\ell$ corresponds to an assignment where variable $x_i$ takes value $\Delta + 1$,
            then add an edge between $v^j_\ell$ and the vertex $a$ of $\hat{B}_{i,j}$.
        \end{romanenumerate}
    \end{itemize}

    Let graph $G_0$ correspond to the graph containing all variable gadgets $P_i$ as well as all the constraint gadgets $\hat{C}_j$,
    for $i \in [n]$ and $j \in [m]$.
    We refer to the block gadget $\hat{B}_{i,j}$, to the variable gadget $P_i$,
    and to the constraint gadget $\hat{C}_j$ of $G_z$
    as $\hat{B}^{G_z}_{i,j}$, $P^{G_z}_i$, and $\hat{C}^{G_z}_j$ respectively.
    To construct graph $G$, introduce $\kappa = \kappa_1 \cdot \kappa_2$ copies $G_1, \ldots, G_\kappa$ of $G_0$,
    where $\kappa_1 = n+1$ and $\kappa_2 = (2\Delta + 1) n + 1$,
    such that they are connected sequentially as follows:
    for $i \in [n]$ and $j \in [\kappa - 1]$, 
    the vertex $a'$ of $\hat{B}^{G_j}_{i,m}$ is the vertex $a$ of $\hat{B}^{G_{j+1}}_{i,1}$.
    Let $\mathcal{P}^i$ denote the ``path'' resulting from $P^{G_1}_i, \ldots, P^{G_\kappa}_i$.
    Set $k = n + \kappa \cdot k_c$, where $k_c = m n (\Delta+1) + \sum_{j=1}^m (s(c_j) - 1)$,
    and let $\mathcal{I} = (G,k)$ denote the instance of \BDD{} for $\Delta = B - 2$.
    
    \begin{lemma}\label{lem:bdd_tw_lb_D1_cor1}
        If $\phi$ is satisfiable,
        then there exists $S \subseteq V(G)$ such that $G - S$ has maximum degree at most $\Delta$ and $|S| \leq k$.
    \end{lemma}
    
    \begin{proof}
        Let $f : X \rightarrow Y$ be an assignment that satisfies all constraints $c_j$.
        We will present a set $S \subseteq V(G)$ of size $|S| \leq k$ such that $G - S$ has maximum degree at most $\Delta$.
        First, for every constraint gadget $\hat{C}_j$, since $f$ is a satisfying assignment,
        there exists some vertex $v^j_\ell$ which corresponds to a restriction of $f$ to the at most $q$ variables involved in $c_j$.
        Add in $S$ every vertex $v^j_l$ for $l \in [s(c_j)] \setminus \braces{\ell}$.
        Consequently, $S$ contains exactly $\kappa \cdot \sum_{j=1}^m (s(c_j) - 1)$ such vertices.
        Then, for every $i \in [n]$, consider the following two cases:
        \begin{itemize}
            \item If $f(x_i) = \Delta + 1$, then, for every block gadget $\hat{B}_{i,j}$,
            include in $S$ the vertices $a$, $a'$, and $q_i$, for $i \in [\Delta]$.
            Since the vertex $a'$ of a block gadget is the vertex $a$ of its subsequent block gadget,
            $S$ includes $\kappa \cdot m \cdot (\Delta + 1) + 1$ vertices in this case.
        
            \item If on the other hand $f(x_i) = p \in [0, \Delta]$, then,
            for every block gadget $\hat{B}_{i,j}$,
            include in $S$ the vertices $b$, $\chi_i$, and $y_j$, for $i \in [p]$ and $j \in [p+1, \Delta]$.
            In this case, $S$ includes $\kappa \cdot m \cdot (\Delta + 1)$ vertices.
        \end{itemize}
        Consequently, $S$ contains at most $\kappa \cdot \sum_{j=1}^m (s(c_j) - 1) + n \cdot (\kappa \cdot m \cdot (\Delta + 1) + 1) =
        n + \kappa \cdot k_c = k$ vertices.
    
        It remains to prove that $G - S$ has maximum degree at most $\Delta$.
        Consider any constraint gadget $\hat{C}^{G_z}_j$.
        Any leaf present has degree either $0$ or $1$ in $G-S$.
        Let $v^j_\ell (G_z) \notin S$ denote the only vertex of the clique of the constraint gadget which does not belong to $S$.
        This vertex corresponds to a satisfying assignment for $c_j$, which is a restriction of $f$.
        It holds that $\deg_{G-S}(v^j_\ell (G_z)) = \Delta$, since
        \begin{itemize}
            \item none of its $\Delta$ neighboring leaves belongs to $S$,
            \item $v^j_{\ell'} (G_z) \in S$, for $\ell' \neq \ell$,
            \item $v^j_\ell (G_z)$ has no neighbors in $\hat{B}^{G_z}_{i,j}$, for all $i \in [n]$:
            consider the two cases where either $a$ is a neighbor of $v^j_\ell (G_z)$
            or vertices $b, \chi_i, y_j$, where $i \in [p]$ and $j \in [p+1, \Delta]$, for some $p \in [0,\Delta]$,
            are neighbors of $v^j_\ell (G_z)$.
            In both cases, since $v^j_\ell (G_z)$ corresponds to an assignment $g$ which is a restriction of $f$,
            all of its neighbors in $\hat{B}^{G_z}_{i,j}$ belong to $S$.
        \end{itemize}
        Consequently, there are no edges between vertices of the constraint and the block gadgets in $G-S$.
    
        Lastly, for $i \in [n]$, no vertex in a block gadget $\hat{B}_{i,j}$ of $\mathcal{P}^i$ has degree larger than $\Delta$ in $G-S$:
        \begin{itemize}
            \item if $f(x_i) = \Delta + 1$, then vertex $b$ has degree $\Delta$, while the vertices $\chi_i, y_i$ have degree $0$.
        
            \item on the other hand, if $f(x_i) = p \neq \Delta + 1$, it holds that all vertices $q_i$ have degree $\Delta$, since one of their neighbors belongs to $S$.
            As for the non-deleted vertices $\chi_i$ and $y_i$, they have degree equal to $2 \leq \Delta$.
            Lastly, any vertex $a$ has degree at most $\Delta$, since it has $\Delta - p$ neighbors due to the block gadget it appears as $a$ and $p$ neighbors due to the block gadget it appears as $a'$.
        \end{itemize}
        This concludes the proof.
    \end{proof}
    
    \begin{lemma}\label{lem:bdd_tw_lb_D1_cor2}
        If there exists $S \subseteq V(G)$ such that $G - S$ has maximum degree $\Delta$ and $|S| \leq k$,
        then $\phi$ is satisfiable.
    \end{lemma}
    
    \begin{proof}
        Let $S \subseteq V(G)$ such that $G - S$ has maximum degree $\Delta$ and $|S| \leq k$.
        For $G_z$, consider a mapping between subsets of vertices of $\hat{B}^{G_z}_{i,j}$ that belong to $S$ and the value of $x_i$ for some assignment of the variables of $\phi$.
        In particular, $S$ containing vertex $b$ as well as $p$ vertices out of $\setdef{\chi_i}{i \in [n]}$ and $\Delta - p$ out of $\setdef{y_i}{i \in [n]}$ is mapped with an 
        assignment where $x_i$ receives value $p \in [0, \Delta]$,
        while $S$ containing just vertex $a$ is mapped with an assignment where $x_i$ receives value $\Delta+1$.
        Any other subset will correspond to an \emph{invalid} assignment of $x_i$.
        We will say that an \emph{inconsistency} occurs in a variable gadget $P^{G_z}_i$ if there exist block gadgets $\hat{B}^{G_z}_{i,j}$ and $\hat{B}^{G_z}_{i,j+1}$,
        such that the vertices belonging to $S$ from each block gadget map to either invalid or different assignments of $x_i$.
        We say that $G_z$ is \emph{consistent} if the deletions occurring in its block gadget $\hat{B}^{G_z}_{i,j}$ map to some non-invalid assignment for $x_i$,
        while additionally no inconsistency occurs in its variable gadgets $P^{G_z}_i$, for every $i \in [n]$.
    
        \begin{claim}
            There exists $\pi \in [\kappa]$ such that $G_\pi$ is consistent.
            Moreover, it holds that $S$ contains vertices $v^j_l$, where $l \in [s(c_j)] \setminus \braces{\ell}$,
            for some $\ell \in [s(c_j)]$, from $\hat{C}^{G_\pi}_j$.
        \end{claim}
        
        \begin{claimproof}
            In the following, we will assume that the vertices $a'$ of block gadgets $\hat{B}^{G_z}_{i,j}$,
            which coincide with the vertices $a$ of $\hat{B}^{G_z}_{i,j+1}$ if $j < m$ and of $\hat{B}^{G_{z+1}}_{i,1}$ alternatively,
            will not belong to the vertices of $\hat{B}^{G_z}_{i,j}$, so that we do not double count them
            (with the exception of the block gadget $\hat{B}^{G_\kappa}_{i,m}$).
            As a consequence, for $z \in [\kappa-1]$, vertices $a'$ of $\hat{B}^{G_z}_{i,m}$ will not belong to $V(G_z)$ and will belong to $V(G_{z+1})$ instead.
            We will first prove that $|S \cap V(G_z)| \geq k_c$, for all $z \in [\kappa]$.
            Notice that from each constraint gadget $\hat{C}_j$, at least $s(c_j) - 1$ vertices belong to $S$:
            each such gadget has a clique of size $s(c_j)$, every vertex of which has $\Delta$ leaves attached.
            Furthermore, for each block gadget $\hat{B}_{i,j}$,
            since $b$ as well as vertices $q_i$ have degree larger than $\Delta$, while no two share any neighbors,
            $\Delta + 1$ deletions are required.
            Notice that, since $b$ has degree $\Delta+2$, even if we delete the vertex $a'$, in the following block gadget $b$
            still has degree over $\Delta$.
            Consequently, from every $G_z$, at least $\sum_{j=1}^m (s(c_j) - 1) + nm (\Delta + 1) = k_c$ vertices are deleted,
            and $|S \cap V(G_z)| \geq k_c$ follows.
        
            Since $\kappa = \kappa_1 \cdot \kappa_2 = (n+1) \cdot \kappa_2$,
            while $|S| \leq n + \kappa \cdot k_c$,
            it follows that there exist $\kappa_2$ sequential copies $G_{z'+1}, \ldots, G_{z' + \kappa_2}$, such that $|S \cap V(G_w)| = k_c$, for $w \in [z'+1, z' + \kappa_2]$.
            Consequently, for any such $G_w$, $S$ contains exactly $s(c_j) - 1$ vertices from each $\hat{C}^{G_w}_j$ as well as
            exactly $\Delta + 1$ vertices from $\hat{B}^{G_w}_{i,j}$, where $i \in [n]$ and $j \in [m]$.
            Out of those $\Delta+1$ vertices, $1$ must belong to the set $\braces{\chi_i, y_i, q_i}$, for every $i \in [n]$.
            The last vertex must be either vertex $b$ or one of its neighbors:
            if that were not the case, then $b$ would have $\Delta+1$ neighbors among the vertices of $\hat{B}^{G_w}_{i,j}$.
            Assume that, in the case $b \notin S$, the neighbor of $b$ belonging to $S$ is $a$:
            if that is not the case, one can construct a deletion set $S'$ of the same cardinality for which the statement holds,
            and since the neighborhood of $a$ is a superset of the neighborhood of any leaf attached to $b$,
            $G - S'$ has maximum degree at most $\Delta$.
        
            Let $\hat{B}^{G_w}_{i,j}$ be one of those block gadgets and let $N(\hat{B}^{G_w}_{i,j})$ denote the block gadget whose $a$ vertex is the vertex $a'$ of $\hat{B}^{G_w}_{i,j}$.
            We will consider two cases: either $a \in S$ or $b \in S$.
            
            If $a \in S$, then for any block gadgets $\hat{B}^{G_{w'}}_{i,j'}$, where either a) $w' = w$ and $j' > j$ or
            b) $w' > w$ and $j' \in [m]$, it holds that $a \in S$:
            the statement holds for $N(\hat{B}^{G_w}_{i,j} (G_w))$, otherwise the vertex $b$ of $\hat{B}^{G_w}_{i,j}$ has degree $\Delta+1$ in $G-S$.
            Suppose that it holds for every block gadget up to some $\hat{B}^{G_{w^*}}_{i,j^*}$.
            Then, it holds for $N(\hat{B}^{G_{w^*}}_{i,j^*})$ as well, otherwise the vertex $b$ of $\hat{B}^{G_{w^*}}_{i,j^*}$ has degree $\Delta + 1$ in $G - S$.
            
            On the other hand, assume that $b \in S$.
            We will prove that $\mathcal{P}^i$ may induce at most $2\Delta$ inconsistencies.
            Notice that as soon as we delete vertex $a$ from some block gadget, then we delete also vertex $a$ from all subsequent block gadgets.
            Thus, assume in the following that exactly $\Delta + 1$ vertices are deleted per block gadget, none of which is vertex $a$.
            If for block gadget $\hat{B}^{G_w}_{i,j}$, $\alpha_p$ vertices are included in $S$ from set $\setdef{\chi_i}{i \in [n]}$,
            $\beta_p$ from set $\setdef{q_i}{i \in [n]}$ and $\gamma_p$ from set $\setdef{y_i}{i \in [n]}$,
            where $p$ is an index denoting the block gadget,
            then $a'(\hat{B}^{G_w}_{i,j})$ has $\Delta - \gamma_p$ incoming edges from $\setdef{y_i}{i \in [n]}$ in $G - S$,
            where $a' (\hat{B}^{G_w}_{i,j})$ denotes the vertex $a'$ of $\hat{B}^{G_w}_{i,j}$.
            Notice that vertices $a' (\hat{B}^{G_w}_{i,j})$ and $a (N(\hat{B}^{G_w}_{i,j}))$ coincide, thus the latter may have at most $\gamma_p$ incoming
            edges from the block gadget $N(\hat{B}^{G_w}_{i,j})$, i.e.~$\gamma_p + \alpha_{p+1} \geq \Delta$.
            Since $\alpha_i + \beta_i + \gamma_i = \Delta$,
            $\gamma_p \geq \gamma_{p+1} + \beta_{p+1}$ and $\alpha_{p+1} \geq \alpha_p + \beta_p$ follow.
            Each inconsistency then corresponds to either $\alpha_{p+1} > \alpha_p$ or $\gamma_{p+1} < \gamma_p$.
            Taking into account the fact that $0 \leq \alpha_i, \gamma_i \leq \Delta$, one can infer that the number of possible inconsistencies is at most $2\Delta$;
            half of them occur due to the increase of $\alpha_p$ while the other half due to the decrease of $\gamma_p$ (one may assume that in the worst case these happen independently).
        
            Therefore, the maximum amount of inconsistencies for each of the $n$ paths $\mathcal{P}_i$,
            without deleting more than $\Delta + 1$ vertices per block gadget,
            is $2\Delta + 1$ (the $+1$ being due to the case where $a \in S$).
        
            Now, suppose that $\beta_p > 0$, for some block gadget $\hat{B}_{i,j}$ of index $p$.
            Then we will show that $\alpha_p \neq \alpha_{p+1}$.
            Suppose that this is not the case.
            Then, it holds that $\gamma_p + \alpha_{p + 1} \geq \Delta \iff \gamma_p + \alpha_p \geq \Delta$,
            which implies that $\beta_p \leq 0$, contradiction.
        
            Since we have $\kappa_2 = (2\Delta+1) n + 1$ repetitions of the whole construction,
            due to the pigeonhole principle,
            there is a copy $G_\pi$, for which all the deletions happening in the block gadgets of $P^{G_\pi}_i$ are mapped to the same non-invalid assignment (this is indeed a non-invalid assignment, since $\beta_p = 0$).
            Moreover, since for every constraint gadget $\hat{C}^{G_\pi}_j$ exactly $s(c_j) - 1$ vertices are included in $S$, and each such gadget has a clique
            of size $s(c_j)$, each vertex of which has $\Delta$ leaves attached, $S$ cannot contain any leaves from $\hat{C}^{G_\pi}_j$.
        \end{claimproof}
    
        Consider the following assignment $f : X \rightarrow Y$ on the variables of $\phi$:
        \begin{itemize}
            \item if $a \in S$, then let $f(x_i) = \Delta + 1$,
            \item alternatively, let $f(x_i) = y$, where $y \in [0, \Delta]$ is equal to the number of vertices of $\setdef{\chi_i}{i \in [\Delta]}$
            belonging to $S$,
        \end{itemize}
        where $a$ and $\chi_i$ refer to the vertices of the block gadget $\hat{B}^{G_\pi}_{i,j}$ for some $j \in [m]$
        (since $G_\pi$ is consistent, the choice of $j$ does not matter).
    
        We will prove that $f$ satisfies all the constraints.
        Consider a constraint $c_j$.
        Regarding the constraint gadget $\hat{C}^{G_\pi}_j$, it holds that $S$ includes exactly $s(c_j) - 1$ vertices,
        none of which can be a leaf vertex.
        Let $v^j_\ell \notin S$ be the only non-leaf vertex of $\hat{C}^{G_\pi}_j$ not belonging to $S$.
        Since none of its leaves belongs to $S$, it follows that every neighbor of $v^j_\ell$ in the block gadgets $\hat{B}^{G_\pi}_{i,j}$ belongs to $S$.
        In that case, notice that the deletion of all the neighbors of $v^j_\ell$ in the block gadget $\hat{B}^{G_\pi}_{i,j}$ is mapped to an assignment of $x_i$:
        if the only neighbor was $a$, then $f(x_i) = \Delta + 1$,
        alternatively the neighborhood was comprised of $b$ as well as vertices $\setdef{\chi_i}{i \in [p]} \cup \setdef{y_i}{i \in [p+1, \Delta]}$,
        which implies that exactly $p$ vertices of $\setdef{\chi_i}{i \in [\Delta]}$ belong to $S$.
        Consequently, the satisfying assignment corresponding to $v^j_\ell$ is a restriction of $f$, thus $f$ satisfies $c_j$.
        Since this holds for any $j$, $\phi$ is satisfied.
    \end{proof}
    
    \begin{lemma}\label{lem:bdd_tw_lb_D1_pw}
        It holds that $\pw(G) = n + \bO(1)$.
    \end{lemma}
    
    \begin{proof}
        We will prove the statement by providing a mixed search strategy to clean $G$ using at most this many searchers simultaneously.
        Since for the mixed search number \ms{} it holds that $\pw(G) \leq \ms(G) \leq \pw(G) + 1$,
        we will show that $\ms(G) \leq n + 3 + 2 \cdot B^q$ and the statement will follow.
    
        Start with graph $G_1$.
        Place $2s(c_1)$ searchers to the clique vertices of $\hat{C}^{G_1}_1$ as well as to one leaf per clique vertex,
        and $n$ searchers on vertices $a$ of block gadgets $\hat{B}^{G_1}_{i,1}$, for $i \in [n]$.
        By moving the searchers placed on the leaves among the leaf vertices of the constraint gadget, all the edges of the constraint gadget can be cleaned.
        Next we will describe the procedure to clean $\hat{B}^{G_1}_{i,1}$.
        Move three extra searchers to vertices $b$ and $a'$ of the block gadget, as well as to a leaf of $b$.
        Move the latter among all the different leaves of $b$.
        Finally, the searchers placed in $b$ and its leaf can clean the rest of the edges of the gadget:
        put one on $\chi_i$ and the other on a leaf of $q_i$.
        Slide the first along the edge connecting $\chi_i$ with $q_i$.
        Move the latter searcher among all the leaves of $q_i$, and then place it to $y_i$.
        Repeat the whole process for all $i$, and finally remove all searchers apart from the one placed on $a'$.
        By following the same procedure, eventually all block gadgets $\hat{B}^{G_1}_{i,1}$ are cleaned.
    
        In order to clean the rest of the graph, we first move the searchers from $\hat{C}^{G_z}_j$ to $\hat{C}^{G_z}_{j+1}$ if $j < m$
        or to $\hat{C}^{G_{z+1}}_1$ alternatively (possibly introducing new searchers if required),
        clean the latter, and then proceed by cleaning the corresponding block gadgets.
        By repeating this procedure, in the end we clean all the edges of $G$ by using at most $n + 3 + 2 \cdot B^q = n + \bO (1)$ searchers.
    \end{proof}
    
    Therefore, in polynomial time we can construct a graph $G$,
    of pathwidth $\pw(G) \leq n + \bO(1)$ due to~\cref{lem:bdd_tw_lb_D1_pw},
    such that, due to~\cref{lem:bdd_tw_lb_D1_cor1,lem:bdd_tw_lb_D1_cor2},
    deciding whether there exists $S \subseteq V(G)$ of size $|S| \leq k$ and $G - S$ has maximum degree at most $\Delta$
    is equivalent to deciding whether $\phi$ is satisfiable.
    In that case, assuming there exists a $\sO((\Delta + 2 - \varepsilon)^{\pw(G)})$ algorithm for \BDD,
    one could decide $q$-CSP-$B$ in time $\sO((\Delta + 2 - \varepsilon)^{\pw(G)}) = \sO((B-\varepsilon)^{n + \bO(1)}) = \sO((B-\varepsilon)^n)$,
    which contradicts the SETH due to~\cref{thm:q_CSP_B_SETH}.
\end{proof}

\subsection{Defective Coloring}

In this subsection, we present tight lower bounds for \DC{} parameterized by the treewidth of the input graph plus the target degree.
We start by presenting in~\cref{sec:dc_coloring_gadgets} a variety of useful gadgets employed in our constructions,
followed by the lower bounds in~\cref{sec:dc_tw_lb} and
an algorithm of matching running time in~\cref{sec:dc_tw_algo}.

\subsubsection{Coloring Gadgets}\label{sec:dc_coloring_gadgets}

Here we develop various gadgets that will be mainly used in~\cref{sec:dc_tw_lb}.
We heavily rely on the constructions presented in~\cite{siamdm/BelmonteLM20},
some of which we briefly present for the sake of completeness.

Intuitively, our goal is to extend the toolbox of~\cite{siamdm/BelmonteLM20},
and to construct gadgets that, for any $(\chid,\Delta)$-coloring $c$ of a graph $G$,
where $\chid \geq 2$ and $\Delta \geq 1$,
imply the following relationships for vertices $v_1, v_2, c_1, c_2 \in V(G)$
(where $c_1, c_2$ might coincide):
\begin{itemize}
    \item $D(v_1,v_2)$: Vertices $v_1$ and $v_2$ receive distinct colors,
    i.e.~$c(v_1) \neq c(v_2)$,
    
    \item $E(v_1,v_2,c_1,c_2)$: If vertex $v_1$ receives the color of $c_1$,
    then vertex $v_2$ does not receive the color of $c_2$,
    i.e.~$c(v_1) = c(c_1) \implies c(v_2) \neq c(c_2)$,
    
    \item $I(v_1,v_2,c_1,c_2)$: If vertex $v_1$ receives the color of $c_1$,
    then vertex $v_2$ receives the color of $c_2$,
    i.e.~$c(v_1) = c(c_1) \implies c(v_2) = c(c_2)$.
\end{itemize}

We build upon the results of~\cite{siamdm/BelmonteLM20},
who presented the \emph{equality gadget} $Q(v_1, v_2, \chid, \Delta)$
as well as the \emph{palette gadget} $P(v_1,v_2,v_3,\chid,\Delta)$.

\begin{lemma}[Lemma~3.5 of~\cite{siamdm/BelmonteLM20}]\label{lem:dc_equality_gadget}
    Let $G = (V, E)$ be a graph with $v_1, v_2 \in V$,
    and let $G'$ be the graph obtained from $G$ by adding to it a copy of $Q(u_1, u_2, \chid, \Delta)$ and identifying $u_1$ with $v_1$ and $u_2$ with $v_2$.
    Then, any $(\chid, \Delta)$-coloring of $G'$ must give the same color to $v_1, v_2$.
    Furthermore, if there exists a $(\chid, \Delta)$-coloring of $G$ that gives the same color to $v_1, v_2$,
    this coloring can be extended to a $(\chid, \Delta)$-coloring of $G'$.
\end{lemma}

\begin{lemma}[Lemma~3.8 of~\cite{siamdm/BelmonteLM20}]\label{lem:dc_palette_gadget}
    Let $G = (V, E)$ be a graph with $v_1, v_2, v_3 \in V$,
    and let $G'$ be the graph obtained from $G$ by adding to it a copy of $P(u_1, u_2, u_3, \chid, \Delta)$ and identifying $u_i$ with $v_i$ for $i \in [3]$.
    Then, in any $(\chid, \Delta)$-coloring of $G'$, at least two of the vertices of $\{ v_1, v_2, v_3 \}$ must share a color.
    Furthermore, if there exists a $(\chid, \Delta)$-coloring of $G$ that gives the same color to two of the vertices of $\{ v_1, v_2, v_3 \}$,
    this coloring can be extended to a $(\chid, \Delta)$-coloring of $G'$.
\end{lemma}

\begin{definition}[Difference Gadget]
    For $\chid \geq 2$, $\Delta \geq 0$, $k \geq 0$,
    let $D(u_1, u_2, \chid, \Delta, k)$ be a graph defined as follows:
    $D$ contains vertices $u_1$ and $u_2$, with the latter having $k$ leaves $l_1, \ldots, l_k$,
    as well as $k$ copies of the equality gadget $Q(x, y, \chid, \Delta)$, on each of which we identify $x$ with $u_1$
    and $y$ with $l_i$, for $i \in [k]$.
\end{definition}

\begin{lemma}\label{lem:dc_difference_gadget}
    Let $G = (V, E)$ be a graph with $v_1 , v_2 \in V$,
    and let $G'$ be the graph obtained from $G$ by adding to it a copy of $D(u_1, u_2, \chid, \Delta, \Delta+1)$ and
    identifying $u_1$ with $v_1$ and $u_2$ with $v_2$.
    Then, any $(\chid,\Delta)$-coloring of $G'$ must give different colors to $v_1$ and $v_2$.
    Furthermore, if there exists a $(\chid,\Delta)$-coloring of $G$ that gives different colors
    to $v_1$ and $v_2$, this coloring can be extended to a $(\chid, \Delta)$-coloring of $G'$.
\end{lemma}

\begin{proof}
    For the first statement, consider a $(\chid,\Delta)$-coloring $c' : V(G') \rightarrow [\chid]$ of $G'$.
    Notice that due to the properties of the equality gadget~\cite[Lemma~3.5]{siamdm/BelmonteLM20},
    it follows that $c'(v_1) = c'(l_i)$ for all $i \in [\Delta+1]$,
    where $l_i$ denote the leaves attached to $v_2$ due to the difference gadget.
    In that case, $v_2$ has at least $\Delta + 1$ neighbors of color $c'(v_1)$, thus $c'(v_2) \neq c'(v_1)$.
    
    For the second statement, let $c : V(G) \rightarrow [\chid]$ be a $(\chid,\Delta)$-coloring of $G$,
    where $c(v_1) \neq c(v_2)$.
    In order to extend it to a $(\chid,\Delta)$-coloring of $G'$,
    color the vertices $l_i$ with the color of $v_1$,
    and appropriately color the vertices of the equality gadgets by using~\cite[Lemma~3.5]{siamdm/BelmonteLM20}.
\end{proof}

\begin{definition}[Exclusion Gadget]
    For $\chid \geq 2$, $\Delta \geq 1$, $C = \setdef{c_i}{i \in [\chid]}$,
    and (not necessarily distinct) $i_1, i_2 \in [\chid]$,
    let $E(u_1,u_2,c_{i_1},c_{i_2},C,\chid,\Delta)$ be a graph defined as follows:
    \begin{itemize}
        \item If either a) $i_1 = i_2$, or b) $i_1 \neq i_2$ and $\chid = 2$,
        then let $E$ contain vertices $u_1, u_2, u'_1, u'_2, c_{i_1}$, and $a$.
        Add edges between $a$ and $u'_1, u'_2$,
        as well as gadgets $Q(u_1,u'_1,\chid, \Delta)$,
        $Q(c_{i_1},a,\chid, \Delta)$, and $D(c_{i_1},a,\chid, \Delta,\Delta-1)$.
        If $i_1 = i_2$ add the gadget $Q(u_2,u'_2,\chid, \Delta)$,
        else the gadget $D(u_2,u'_2,\chid, \Delta, \Delta+1)$.
        
        \item Otherwise, i.e.~when $i_1 \neq i_2$ and $\chid > 2$,
        let $E$ contain vertices $u_1, u_2, u'_1, u'_2, c_{i_1}, c_{i_2}, c_{i_3}$ and $a_1, a_2, a_3$,
        where $i_3 \in [\chid] \setminus \braces{i_1, i_2}$.
        Let vertices $u'_1, a_1, a_2, a_3, u'_2$ form a path and add the following gadgets:
        \begin{itemize}
            \item $Q(u_1, u'_1, \chid, \Delta)$,
            \item $Q(u_2, u'_2, \chid, \Delta)$,
            \item $P(c_{i_1}, c_{i_3}, a_1, \chid, \Delta)$, $D(c_{i_1}, a_1, \chid, \Delta, \Delta)$, $D(c_{i_3}, a_1, \chid, \Delta, \Delta)$,
            \item $P(c_{i_1}, c_{i_3}, a_2, \chid, \Delta)$, $D(c_{i_1}, a_2, \chid, \Delta, \Delta)$, $D(c_{i_3}, a_2, \chid, \Delta, \Delta)$,
            \item $P(c_{i_1}, c_{i_2}, a_3, \chid, \Delta)$, $D(c_{i_1}, a_3, \chid, \Delta, \Delta)$, $D(c_{i_2}, a_3, \chid, \Delta, \Delta)$.
        \end{itemize}
    \end{itemize}
\end{definition}

\begin{figure}[ht]
    \centering 
        \begin{subfigure}[b]{0.4\linewidth}
        \centering
        \begin{tikzpicture}[scale=0.75, transform shape]
            \node[vertex] (u1) at (0,1) {};
            \node[above =0.2 of u1] {$u_1$};
            
            \node[vertex, right =1.5 of u1] (u1') {};
            \node[above =0.2 of u1'] {$u'_1$};
        
            \path (u1) -- (u1') node[midway,above] {$=$};
            
            \node[black_vertex, right = of u1'] (a) {};
            \node[above =0.2 of a] {$a$};
        
            \node[vertex, below = of a] (c) {};
            \node[right =0.2 of c] {$c_{i_1}$};
        
            \path (a) -- (c) node[midway,left] {$=$};
        
            \node[vertex, right = of a] (u2') {};
            \node[above =0.2 of u2'] {$u'_2$};
        
            \node[vertex, right =1.5 of u2'] (u2) {};
            \node[above =0.2 of u2] {$u_2$};
        
            \path (u2) -- (u2') node[midway,above] {$=$ or $\neq$};
            
            \draw[] (u1')--(a)--(u2');
            
            \draw[dashed] (u1)--(u1');
            \draw[dashed] (u2)--(u2');
            \draw[dashed] (a)--(c);

        \end{tikzpicture}
        \caption{Either $i_1 = i_2$, or $i_1 \neq i_2$ and $\chid = 2$.}
        \label{fig:dc_exclusion_gadget1}
        \end{subfigure}
    \begin{subfigure}[b]{0.4\linewidth}
    \centering
        \begin{tikzpicture}[scale=0.75, transform shape]
            \node[vertex] (u1) at (0,1) {};
            \node[above =0.2 of u1] {$u_1$};
            
            \node[vertex, right = of u1] (u1') {};
            \node[above =0.2 of u1'] {$u'_1$};
        
            \path (u1) -- (u1') node[midway,above] {$=$};
            
            \node[vertex, right = of u1'] (a1) {};
            \node[above =0.2 of a1] {$a_1$};
            \node[below =0.2 of a1] {$[c_{i_1},c_{i_3}]$};
        
            \node[vertex, right =1.2 of a1] (a2) {};
            \node[above =0.2 of a2] {$a_2$};
            \node[below =0.2 of a2] {$[c_{i_1},c_{i_3}]$};
        
            \node[vertex, right =1.2 of a2] (a3) {};
            \node[above =0.2 of a3] {$a_3$};
            \node[below =0.2 of a3] {$[c_{i_1},c_{i_2}]$};
        
            \node[vertex, right = of a3] (u2') {};
            \node[above =0.2 of u2'] {$u'_2$};
        
            \node[vertex, right = of u2'] (u2) {};
            \node[above =0.2 of u2] {$u_2$};
        
            \path (u2) -- (u2') node[midway,above] {$=$};
            
            \draw[] (u1')--(a1)--(a2)--(a3)--(u2');
            
            \draw[dashed] (u1)--(u1');
            \draw[dashed] (u2)--(u2');        
        \end{tikzpicture}
        \caption{Case where $i_1 \neq i_2$ and $\chid > 2$.}
        \label{fig:dc_exclusion_gadget2}
    \end{subfigure}
\caption{Exclusion gadget $E(u_1,u_2,c_{i_1},c_{i_2},C,\chid,\Delta)$.
Black vertex $a$ has $\Delta - 1$ leaves $l$ attached,
each taking part in an equality gadget $Q(c_{i_1},l,\chid, \Delta)$.
The bracket $[x,y]$ under vertex $v$ denotes that there exists a palette gadget $P(x, y, v, \chid, \Delta)$,
as well as that $v$ has $2\Delta$ leaves $l$ attached, with half taking part in a gadget
$Q(x,l,\chid, \Delta)$, and the other half in $Q(y,l,\chid, \Delta)$.}
\label{fig:dc_exclusion_gadget}
\end{figure}

\begin{lemma}\label{lem:dc_exclusion_gadget}
    Let $G = (V, E)$ be a graph with $v_1 , v_2 \in V$ and $P = \setdef{p^i}{i \in [\chid]} \subseteq V$,
    and let,
    for some (not necessarily distinct) $i_1,i_2 \in [\chid]$,    
    $G'$ be the graph obtained from $G$ by adding to it a copy of
    $E(u_1, u_2, c_{i_1}, c_{i_2}, C, \chid, \Delta)$ and
    identifying $u_1$ with $v_1$, $u_2$ with $v_2$, and $c_i$ with $p^i$.
    Then, in any $(\chid,\Delta)$-coloring $c' : V(G') \rightarrow [\chid]$ of $G'$
    where a) $c'(p^i) \neq c'(p^j)$ for distinct $i,j \in [\chid]$ and b) $c'(v_1) = c'(p^{i_1})$,
    it holds that $c'(v_2) \neq c'(p^{i_2})$.
    Furthermore, if there exists a $(\chid,\Delta)$-coloring $c : V(G) \rightarrow [\chid]$ of $G$ where
    a) $c(p^i) \neq c(p^j)$ for distinct $i,j \in [\chid]$ and
    b) either $c(v_1) \neq c(p^{i_1})$ or $c(v_2) \neq c(p^{i_2})$,
    this can be extended to a $(\chid,\Delta)$-coloring of $G'$.
\end{lemma}

\begin{proof}
    For the first statement, consider a $(\chid,\Delta)$-coloring $c' : V(G') \rightarrow [\chid]$ of $G'$,
    where every vertex $p^i$ receives a distinct color.
    Assume that $c'(v_1) = c'(p^{i_1})$ and consider the following cases:
    \begin{itemize}
        \item If either a) $i_1 = i_2$, or b) $i_1 \neq i_2$ and $\chid = 2$,
        then vertex $a$ of $E$ has $\Delta - 1$ leaves of color $c'(p^{i_1})$.
        Moreover, it holds that $c'(u'_1) = c'(p^{i_1})$, due to the gadget $Q(u_1,u'_1,\chid, \Delta)$.
        Consequently, $c'(v_2) \neq c'(p^{i_2})$, since otherwise it holds that $c'(u'_2) = c'(p^{i_1})$
        due to either gadget $Q(u_2,u'_2,\chid, \Delta)$ or $D(u_2,u'_2,\chid, \Delta, \Delta+1)$,
        thus $a$ has $\Delta+1$ same colored neighbors, which is a contradiction.

        \item Otherwise, i.e.~if $i_1 \neq i_2$ and $\chid > 2$,
        it follows that $a_1$ has $\Delta$ leaves of color $c'(p^{i_1})$
        as well as a neighbor which is connected via an equality gadget with $v_1$.
        Consequently, $a_1$ has $\Delta+1$ neighbors of color $c'(p^{i_1})$, and due to the palette gadget,
        it follows that $c'(a_1) = c'(p^{i_3})$.
        In an analogous way, it follows that $c'(a_2) = c'(p^{i_1})$ and that $c'(a_3) = c'(p^{i_2})$.
        Assume that $c'(v_2) = c'(p^{i_2})$.
        Then, $a_3$ has $\Delta+1$ neighbors of color $c(p^{i_2})$,
        due to its $\Delta$ leaves as well as vertex $u'_2$,
        which is connected via an equality gadget with $v_2$,
        which leads to a contradiction.
    \end{itemize}

    For the second statement, consider a $(\chid,\Delta)$-coloring $c : V(G) \rightarrow [\chid]$ of $G$,
    where every vertex $p^i$ receives a distinct color and additionally
    either $c(v_1) \neq c(p^{i_1})$ or $c(v_2) \neq c(p^{i_2})$.
    Consider the following cases:
    \begin{itemize}
        \item If either a) $i_1 = i_2$, or b) $i_1 \neq i_2$ and $\chid = 2$,
        then let vertex $a$ of $E$, as well as the latter's $\Delta - 1$ leaves receive color $c(p^{i_1})$.
        Next, color $u'_1$ with $c(u_1)$.
        In case $i_1 = i_2$, then let $u'_2$ receive color $c(v_2)$,
        otherwise it receives color different than $c(v_2)$.
        In both cases $a$ has at most $\Delta$ same colored neighbors,
        and by using~\cref{lem:dc_equality_gadget,lem:dc_difference_gadget},
        we can color all the internal vertices of the equality and difference gadgets.

        \item Otherwise, i.e.~if $i_1 \neq i_2$ and $\chid > 2$,
        consider the following two cases:
        \begin{itemize}
            \item If $c(v_1) \neq c(p^{i_1})$, then let $a_1$ receive color $c(p^{i_1})$, $a_2$ color $c(p^{i_3})$ and $a_3$ color $c(p^{i_1})$.
            \item Alternatively, it holds that $c(v_2) \neq c(p^{i_2})$, and let $a_1$ receive color $c(p^{i_3})$, 
            $a_2$ color $c(p^{i_1})$ and $a_3$ color $c(p^{i_2})$.
        \end{itemize}
        In both cases, all vertices $a_1,a_2,a_3$ have exactly $\Delta$ same colored neighbors,
        and it remains to color
        the internal vertices of the equality and palette gadgets 
        using~\cref{lem:dc_equality_gadget,lem:dc_palette_gadget}.
    \end{itemize}
    This concludes the proof.
\end{proof}

\begin{definition}[Implication Gadget]
    For $\chid \geq 2$, $\Delta \geq 1$, $C = \setdef{c_i}{i \in [\chid]}$,
    and (not necessarily distinct) $i_1, i_2 \in [\chid]$,
    let $I(u_1, u_2, c_{i_1}, c_{i_2}, C, \chid, \Delta)$ be a graph defined as follows:
    $I$ contains vertices $u_1$ and $u_2$ and exclusion gadgets $E(u_1,u_2, c_{i_1}, c_k, C, \chid, \Delta)$, for all $k \in [\chid] \setminus \braces{i_2}$.
\end{definition}

\begin{lemma}\label{lem:dc_implication_gadget}
    Let $G = (V, E)$ be a graph with $v_1 , v_2 \in V$ and $P = \setdef{p^i}{i \in [\chid]} \subseteq V$,
    and let,
    for some (not necessarily distinct) $i_1,i_2 \in [\chid]$,    
    $G'$ be the graph obtained from $G$ by adding to it a copy of
    $I(u_1, u_2, c_{i_1}, c_{i_2}, C, \chid, \Delta)$ and
    identifying $u_1$ with $v_1$, $u_2$ with $v_2$, and $c_i$ with $p^i$.
    Then, in any $(\chid,\Delta)$-coloring $c' : V(G') \rightarrow [\chid]$ of $G'$ where
    a) $c'(p^i) \neq c'(p^j)$ for distinct $i,j \in [\chid]$ and
    b) $c'(v_1) = c'(p^{i_1})$,
    it holds that $c'(v_2) = c'(p^{i_2})$.
    Furthermore, if there exists a $(\chid,\Delta)$-coloring $c : V(G) \rightarrow [\chid]$ of $G$ where
    a) $c(p^i) \neq c(p^j)$ for distinct $i,j \in [\chid]$ and
    b) either $c(v_1) \neq c(p^{i_1})$ or $c(v_2) = c(p^{i_2})$,
    this can be extended to a $(\chid,\Delta)$-coloring of $G'$.
\end{lemma}

\begin{proof}
    For the first statement, consider a $(\chid,\Delta)$-coloring $c' : V(G') \rightarrow [\chid]$ of $G'$,
    where every vertex $p^i$ receives a distinct color.
    Assume that $c'(v_1) = c'(p^{i_1})$.
    Then, due to~\cref{lem:dc_exclusion_gadget}
    it follows that $c'(v_2) \neq c'(p^k)$ for all $k \in [\chid] \setminus \braces{i_2}$,
    thus $c'(v_2) = c'(p^{i_2})$ follows.

    For the second statement, consider a $(\chid,\Delta)$-coloring $c : V(G) \rightarrow [\chid]$ of $G$,
    where every vertex $p^i$ receives a distinct color and additionally
    either $c(v_1) \neq c(p^{i_1})$ or $c(v_2) = c(p^{i_2})$.
    Consequently, it holds that either $c(v_1) \neq c(p^{i_1})$ or $c(v_2) \neq c(p^k)$, for all $k \in [\chid] \setminus \braces{i_2}$,
    thus from~\cref{lem:dc_exclusion_gadget} the statement follows.
\end{proof}

The following lemma proves that the use of the previously described gadgets does not increase the pathwidth of the graph by much.
Notice that part of it has been proven in~\cite{siamdm/BelmonteLM20}.

\begin{lemma}\label{dc_gadgets_pw}
    Let $G = (V, E)$ be a graph and let $G'$ be the graph obtained from $G$
    by repeating the following operation:
    find a copy of one of the following gadgets
    \begin{itemize}
        \item $Q(u_1, u_2, \chid, \Delta)$,
        \item $P(u_1, u_2, u_3, \chid, \Delta)$,
        \item $D(u_1, u_2, \chid, \Delta,k)$,
        \item $E(u_1, u_2, c_{i_1}, c_{i_2}, C, \chid, \Delta)$,
        \item $I(u_1, u_2, c_{i_1}, c_{i_2}, C, \chid, \Delta)$,
    \end{itemize}
    remove all its internal vertices from the graph,
    and add all edges between its endpoints which are not already connected.
    Then $\tw(G) \leq  \max \{ \tw(G'), \bO(\chid) \}$ and $\pw(G) \leq \pw(G') + \bO(\chid)$.
\end{lemma}

\begin{proof}
    The result for the equality and palette gadgets has been shown in~\cite[Lemma~4.2]{siamdm/BelmonteLM20}.
    In particular, it is shown that one can obtain a path decomposition of an equality or a palette gadget,
    named $T_Q$ and $T_P$ respectively,
    of width $\chid$, every bag of which contains vertices $u_1,u_2$ (and $u_3$ in the case of the palette gadget).
    
    \proofsubparagraph*{Difference Gadget.}
    Remember that the difference gadget contains $k$ leaves $l_i$, attached to vertex $u_2$.
    First, as observed in the proof of~\cite[Lemma~4.2]{siamdm/BelmonteLM20},
    there is a path decomposition of $Q(u_1, l_i, \chid, \Delta)$ with width $\chid$,
    where every bag contains the vertices $u_1$ and $l_i$.
    In that case, there exists a path decomposition of $D(u_1, u_2, \chid, \Delta)$ of width $\chid+1$:
    serially connect the path decompositions of $Q(u_1, l_i, \chid, \Delta)$, for $i \in [k]$,
    and add to all the bags the vertex $u_2$.
    Call this path decomposition $T_D$, and notice that all of its bags contain both vertices $u_1$ and $u_2$.

    \proofsubparagraph*{Exclusion Gadget.}
    First consider the case where either $i_1 = i_2$ or $\chid = 2$.
    Then, the gadget consists of vertices $u_1, u_2, u'_1, u'_2, c_{i_1}$, and $a$.
    Vertex $a$ has an edge with both $u'_1, u'_2$,
    while there exist gadgets $Q(u_1,u'_1,\chid, \Delta)$, $Q(c_{i_1},a,\chid, \Delta)$,
    and $D(c_{i_1},a,\chid, \Delta,\Delta-1)$.
    If $i_1 = i_2$, there also exists the equality gadget $Q(u_2,u'_2,\chid, \Delta)$,
    else the difference gadget $D(u_2,u'_2,\chid, \Delta, \Delta+1)$.
    Construct a path decomposition $T_E$ of width $\bO(\chid)$ comprised of the following path decompositions:
    \begin{itemize}
        \item a path decomposition of $Q(u_1,u'_1,\chid, \Delta)$ of width $\chid$, every bag of which contains $u_1,u'_1$,

        \item a path decomposition of $Q(c_{i_1},a,\chid, \Delta)$ of width $\chid$, every bag of which contains $c_{i_1},a$,

        \item a path decomposition of $D(c_{i_1},a,\chid, \Delta,\Delta-1)$ of width $\chid+1$, every bag of which contains $c_{i_1},a$,

        \item depending on which case we are, a path decomposition of either $Q(u_2,u'_2,\chid, \Delta)$
        of width $\chid$,
        or of $D(u_2,u'_2,\chid, \Delta, \Delta+1)$ of width $\chid+1$,
        every bag of which contains vertices $u_2,u'_2$.
    \end{itemize}
    To conclude the construction of $T_E$, add to every bag of this path decomposition the necessary vertices such that $u_1,u_2,u'_1,u'_2,c_i$ are contained in every bag, for all $i \in [\chid]$.

    For the remaining case, the exclusion gadget consists of equality, palette, and difference gadgets.
    By connecting path decompositions of each respective gadget in an analogous way,
    while subsequently adding vertices in order to ensure that every bag contains
    $u_1, u_2, u'_1, u'_2, a_1, a_2, a_3, c_i$ for all $i \in [\chid]$,
    it follows that we get a path decomposition $T_E$ of width $\bO(\chid)$.

    \proofsubparagraph*{Implication Gadget.}
    Lastly, for the case of the implication gadget,
    remember that this consists of exclusion gadgets
    $E(u_1,u_2,c_{i_1},c_k,C,\chid,\Delta)$, for all $k \in [\chid] \setminus \braces{i_2}$.
    Again, consider a path decomposition of each of those $|C| - 1 = \chid - 1$ exclusion gadgets
    resulting in a path decomposition $T_I$ of width $\bO(\chid)$
    (notice that each bag of the decomposition of an exclusion gadget contains vertices $C \cup \{ u_1, u_2 \}$).

    \medskip
    
    We now take an optimal tree or path decomposition of $G'$, call it $T'$,
    and construct from it a decomposition of $G$.
    Consider a gadget $H \in \braces{Q,P,D,E,I}$ that appears in $G$ with endpoints $u_1, u_2$ (and possibly $u_3$ and $C = \braces{c_1, \ldots, c_{\chid}}$).
    Since in $G'$ these endpoints form a clique, there is a bag in $T'$ that contains all of them.
    Let $B$ be the smallest such bag, that is, the bag that contains the smallest number of vertices.
    Now, if $T'$ is a tree decomposition, we take $T_H$ and attach it to $B$.
    If $T'$ is a path decomposition, we insert in the decomposition immediately after $B$ the decomposition $T_H$
    where we have added all vertices of $B$ in all bags of $T_H$.
    It is not hard to see that in both cases the decompositions remain valid,
    and we can repeat this process for every $H$ until we have a decomposition of $G$.
\end{proof}
\subsubsection{Lower Bound}\label{sec:dc_tw_lb}

In the following, we will present a reduction from $q$-CSP-$B$ to \DC,
for any fixed $\Delta \geq 1$ and $\chid \geq 2$, where $B = \chid \cdot (\Delta + 1)$.
In that case, if there exists a $\sO ((\chid \cdot (\Delta + 1) - \varepsilon)^\pw)$ algorithm for \DC,
where $\varepsilon > 0$,
then there exists a $\sO ((B - \varepsilon)^n)$ algorithm for $q$-CSP-$B$,
for any constant $q$,
which due to~\cref{thm:q_CSP_B_SETH} results in SETH failing.

The reduction is similar in nature to the one presented in~\cref{subsec:bdd_tw_lb},
consisting of ``long paths'' of serially connected block gadgets,
each of which corresponds to a variable of the given CSP,
while each column of this construction is associated with one of its constraints.

In the whole section we will use the coloring gadgets presented in~\cref{sec:dc_coloring_gadgets}.
As was the case for \BDD, we first start with the case where $\Delta = 1$ and
then with the case where $\Delta \geq 2$.

\begin{theorem}
    For any constant $\varepsilon > 0$ and for any fixed $\chid \geq 2$,
    there is no $\sO((2\chid - \varepsilon)^{\pw})$ algorithm
    deciding whether $G$ admits a $(\chid,1)$-coloring,
    where \pw{} denotes the graph's pathwidth,
    unless the SETH is false.
\end{theorem}

\begin{proof}
    Fix some positive $\varepsilon > 0$ for which we want to prove the theorem.
    Let $B = 2 \chid$.
    We will reduce $q$-CSP-$B$, for some $q$ that is a constant that only depends on $\varepsilon$,
    to \DC{} for $\Delta = 1$ and $\chid$ colors in a way that ensures that if the resulting \DC{} instance could be solved in time $\sO((2 \chid - \varepsilon)^{\pw})$,
    then we would obtain an algorithm for $q$-CSP-$B$ that would contradict the SETH due to~\cref{thm:q_CSP_B_SETH}.
    To this end, let $\phi$ be an instance of $q$-CSP-$B$ of $n$ variables $X = \setdef{x_i}{i \in [n]}$ taking values over the set $Y = [B]$
    and $m$ constraints $C = \setdef{c_j}{j \in [m]}$.
    For each constraint we are given a set of at most $q$ variables which are involved in this constraint and a list of satisfying assignments for these variables,
    the size of which is denoted by $s : C \rightarrow [B^q]$,
    i.e.~$s(c_j) \leq B^q = \bO(1)$ denotes the number of satisfying assignments for constraint $c_j$.
    We will construct in polynomial time an equivalent instance $G$ of \DC{} for $\Delta = 1$ and $\chid$ colors,
    where $\pw(G) \leq n + \bO (1)$.
    
    Since we will repeatedly use the equality, difference, and palette gadgets (see~\cref{sec:dc_coloring_gadgets}),
    we will use the following convention:
    whenever $v_1, v_2, v_3$ are vertices we have already introduced to $G$,
    when we say that we add an equality gadget $Q(v_1, v_2)$, a difference gadget $D(v_1, v_2)$,
    or a palette gadget $P(v_1, v_2, v_3)$,
    this means that we add to $G$ a copy of $Q(u_1, u_2, \chid, \Delta)$, of $D(u_1, u_2, \chid, \Delta, \Delta+1)$,
    or of $P(u_1, u_2, u_3, \chid, \Delta)$ respectively,
    and then identify $u_1, u_2(,u_3)$ with $v_1, v_2(,v_3)$ respectively.
    
    \proofsubparagraph*{Palette Vertices.}
    Construct a clique of $\chid$ vertices $P = \setdef{p^i}{i \in [\chid]}$.
    Attach to vertex $p^i$ a leaf $p^i_l$, and add equality gadgets $Q(p^i, p^i_l)$, where $i \in [\chid]$.
    
    Whenever $v_1, v_2$ are vertices we have already introduced to $G$,
    when we say that we add an exclusion gadget $E(v_1, v_2, v'_1, v'_2)$ or an implication gadget $I(v_1, v_2,v'_1, v'_2)$,
    this means that we add to $G$ a copy of $E(u_1, u_2, c_1, c_2, C, \chid, \Delta)$ or of $I(u_1, u_2, c_1, c_2, C, \chid, \Delta)$ respectively
    and then identify $u_1, u_2, c_i$ with $v_1, v_2, p^i$ respectively, for all $i \in [\chid]$.
    
    \proofsubparagraph*{Block and Variable Gadgets.}
    For every variable $x_i$ and every constraint $c_j$, construct a \emph{block gadget} $\hat{B}_{i,j}$ as depicted in~\cref{fig:dc_tw_lb_t1_block}.
    Dashed lines between $u_1$ and $u_2$ imply an equality $Q(u_1, u_2)$ or a difference gadget $D(u_1, u_2)$, and not an edge.
    If $\chid \geq 3$, add palette gadgets $P(a, b, x)$ and $P(a', b, y)$.
    Next, for $j \in [m-1]$, we serially connect the block gadgets $\hat{B}_{i,j}$ and $\hat{B}_{i,j+1}$ so that the vertex $a'$ of $\hat{B}_{i,j}$ is the vertex $a$ of $\hat{B}_{i,j+1}$,
    thus resulting in $n$ ``paths'' $P_1, \ldots, P_n$ consisting of $m$ serially connected block gadgets, called \emph{variable gadgets}.
    Intuitively, the variable gadget is meant to represent a variable $x_i$ and hence needs to have $2\chid$
    different viable configurations.
    These are made up by deciding on a color for $a$ ($\chid$ choices) and then deciding
    which of $x,y$ will receive the same color as $a$ (two choices).
    We will show that the gadget is set up so that exactly one of $x,y$ receives
    the same color as $a$ (and $a'$).

    \begin{figure}[ht]
    \centering 
      \begin{subfigure}[b]{0.4\linewidth}
      \centering
        \begin{tikzpicture}[scale=0.75, transform shape]
            \node[vertex] (a) at (0, 1) {};
            \node[above =0.05 of a] {$a$};
            
            \node[vertex, right=1.5 of a] (x) {};
            \node[above =0.05 of x] {$x$};
            
            \node[vertex, right=1.5 of x] (y) {};
            \node[above =0.05 of y] {$y$};
            
            \node[vertex, right=1.5 of y] (a') {};
            \node[above =0.05 of a'] {$a'$};
            
            \path (x) -- (y) node[vertex,midway,below=1] (b) {};
            \node[right =0.05 of b] {$b$};
            
            \path (a) -- (b) node[midway,below=0.01] {$\neq$};
            \path (a) -- (a') node[midway,above=0.75] {$=$};
            
            \draw[] (a)--(x)--(y)--(a');
            \draw[] (x)--(b);
            \draw[] (y)--(b);
            
            \draw[dashed] (a)--(b);
            \draw[dashed] (a) edge [bend left] (a');
            
            \end{tikzpicture}
            \caption{Block gadget when $\Delta = 1$.}
            \label{fig:dc_tw_lb_t1_block}
        \end{subfigure}
    \begin{subfigure}[b]{0.4\linewidth}
    \centering
        \begin{tikzpicture}[scale=0.6, transform shape]
        
    
            \node[vertex] (a) at (0, 1) {};
            
            \node[vertex, right=1.5 of a] (x) {};
            
            \node[vertex, right=1.5 of x] (y) {};
            
            \node[vertex, right=1.5 of y] (a') {};
            
            \path (x) -- (y) node[vertex,midway,below=1] (b) {};
            
            \path (a) -- (b) node[midway,below=0.01] {$\neq$};
            \path (a) -- (a') node[midway,above=0.75] {$=$};
    
            \node[vertex, right=1.5 of a'] (x2) {};
            
            \node[vertex, right=1.5 of x2] (y2) {};
            
            \node[vertex, right=1.5 of y2] (a2) {};
    
            \path (x2) -- (y2) node[vertex,midway,below=1] (b2) {};
            
            \path (a') -- (b2) node[midway,below=0.01] {$\neq$};
            \path (a') -- (a2) node[midway,above=0.75] {$=$};

            \draw[] (a)--(x)--(y)--(a');
            \draw[] (x)--(b);
            \draw[] (y)--(b);
            
            \draw[dashed] (a)--(b);
            \draw[dashed] (a) edge [bend left] (a');
    
            \draw[] (a')--(x2)--(y2)--(a2);
            \draw[] (x2)--(b2);
            \draw[] (y2)--(b2);
            
            \draw[dashed] (a')--(b2);
            \draw[dashed] (a') edge [bend left] (a2);
        \end{tikzpicture}
        \caption{Serially connected block gadgets.}
        \label{fig:dc_tw_lb_t1_variable}
      \end{subfigure}
    \caption{Variable gadgets are comprised of serially connected block gadgets.}
    \label{fig:dc_tw_lb_t1_block_and_variable}
    \end{figure}  
    
    \proofsubparagraph*{Constraint Gadget.}
    This gadget is responsible for determining constraint satisfaction,
    based on the choices made in the rest of the graph.
    For constraint $c_j$, construct the \emph{constraint gadget} $\hat{C}_j$ as depicted in~\cref{fig:dc_tw_lb_c2_constraint}:
    \begin{itemize}
        \item introduce vertices $r_w$, where $w \in [s(c_j)]$,
        and add $Q(p^2, r_1)$, as well as $P(p^1,p^2,r_w)$ when $\chid \geq 3$, for $w \in [2, s(j)]$,
        
        \item for $w \in [s(j)]$, introduce vertices $v^j_w$, as well as
        palette gadgets $P(p^1,p^2,v^j_w)$ when $\chid \geq 3$,
        and fix an arbitrary one-to-one mapping between those vertices and the satisfying assignments of $c_j$,
    
        \item introduce vertices $k_w$, where $w \in [s(c_j)]$, and add $Q(p^2, k_{s(c_j)})$, as well as $P(p^1, p^2, k_w)$ when $\chid \geq 3$ for $w \in [s(j)-1]$,
    
        \item add edges between vertex $k_w$ and vertices $r_w, v^j_w$, as well as $D(k_w, r_{w+1})$,
    
        \item if variable $x_i$ is involved in the constraint $c_j$ and $v^j_\ell$ corresponds to an assignment where $x_i$ has value $s \in Y$,
        then:
        \begin{romanenumerate}
            \item if $s \leq \chid$, then add implication gadgets $I(v^j_\ell, a, p^1, p^s)$ and $I(v^j_\ell, x, p^1, p^s)$,
    
            \item if on the other hand $\chid < s \leq 2 \chid$, then add implication gadget $I(v^j_\ell, a, p^1, p^{s'})$ and exclusion gadget $E(v^j_\ell, x, p^1, p^{s'})$,
            for $s' = s - \chid$,
        \end{romanenumerate}
        where vertices $a$ and $x$ belong to $\hat{B}_{i,j}$.
    \end{itemize}
    Intuitively, the constraint gadget is set up in a way that forces, for some $\ell \in [s(c_j)]$,
    vertex $v^j_\ell$ to receive color $1$,
    which in turn ``activates'' the implication and exclusion gadgets we have added to this vertex.
    This ensures that the assignment encoded by the variable gadgets agrees with the satisfying assignment
    of $c_j$ represented by $v^j_\ell$.
    
    \begin{figure}[ht]
    \centering
    \begin{tikzpicture}[transform shape]
    
    
    \node[vertex] (r) at (0,1) {};
    \node[above =0.2 of r] {$r_1$};
    
    \node[vertex, right = of r] (k1) {};
    \node[below right =0.01 of k1] {$k_1$};
    
    \node[vertex, below = of k1] (c1) {};
    \node[right =0.01 of c1] {$v^j_1$};
    
    \node[vertex, right = of k1] (r1) {};
    \node[above =0.2 of r1] {$r_2$};
    \path (k1) -- (r1) node[midway,above] {$\neq$};
    
    \node[vertex, right = of r1] (k2) {};
    \node[below right =0.01 of k2] {$k_2$};
    
    \node[vertex, below = of k2] (c2) {};
    \node[right =0.01 of c2] {$v^j_2$};
    
    \node[vertex, right =2 of k2] (rm) {};
    \node[above =0.2 of rm] {$r_{s(c_j)}$};
    
    \node[vertex, right = of rm] (km) {};
    \node[below right =0.01 of km] {$k_{s(c_j)}$};
    
    \node[vertex, below = of km] (cm) {};
    \node[right =0.01 of cm] {$v^j_{s(c_j)}$};
    
    \draw[] (r)--(k1);
    \draw[] (c1)--(k1);
    
    \draw[dashed] (k1)--(r1);
    
    \draw[] (r1)--(k2);
    \draw[] (c2)--(k2);
    
    \draw[dotted] (k2)--(rm);
    
    \draw[] (rm)--(km);
    \draw[] (cm)--(km);
    
    \end{tikzpicture}
    \caption{Constraint gadget $\hat{C}_j$ when $\Delta = 1$.}
    \label{fig:dc_tw_lb_c2_constraint}
    \end{figure}
    
    Let graph $G_0$ correspond to the graph containing all variable gadgets $P_i$ as well as all the constraint gadgets $\hat{C}_j$,
    for $i \in [n]$ and $j \in [m]$.
    We refer to the block gadget $\hat{B}_{i,j}$, to the variable gadget $P_i$,
    and to the constraint gadget $\hat{C}_j$ of $G_z$
    as $\hat{B}^{G_z}_{i,j}$, $P^{G_z}_i$ and $\hat{C}^{G_z}_j$ respectively.
    To construct graph $G$, introduce $\kappa = n+1$ copies $G_1, \ldots, G_\kappa$ of $G_0$,
    such that they are connected sequentially as follows:
    for $i \in [n]$ and $j \in [\kappa - 1]$, 
    the vertex $a'$ of $\hat{B}^{G_j}_{i,m}$ is the vertex $a$ of $\hat{B}^{G_{j+1}}_{i,1}$.
    Let $\mathcal{P}^i$ denote the ``path'' resulting from $P^{G_1}_i, \ldots, P^{G_\kappa}_i$.
    
    \begin{lemma}\label{lem:dc_tw_lb_D1_cor1}
        For any $\chid \geq 2$, if $\phi$ is satisfiable, then $G$ admits a $(\chid,1)$-coloring.
    \end{lemma}
    
    \begin{proof}
        Let $f : X \rightarrow Y$ denote an assignment which satisfies all the constraints $c_1, \ldots, c_m$.
        We will describe a $(\chid, 1)$-coloring $c : V(G) \rightarrow [\chid]$ of $G$.
    
        Let $c(p^i) = c(p^i_l) = i$, for $i \in [\chid]$.
        Next, for the vertices of block gadget $\hat{B}^{G_z}_{i,j}$,
        where $z \in [\kappa]$, $i \in [n]$ and $j \in [m]$,
        consider the following cases:
        \begin{enumerate}
            \item if $f(x_i) = k$, for $k \in [\chid]$, then let $k' \in [\chid] \setminus \{ k \}$ be an arbitrary color and set $c(a) = c(x) = k$, while $c(b) = c(y) = k'$,
            \item if $f(x_i) = \chid + k$, for $k \in [\chid]$, then let $k' \in [\chid] \setminus \{ k \}$ be an arbitrary color and set $c(a) = c(y) = k$, while $c(b) = c(x) = k'$.
        \end{enumerate}
        Regarding the constraint gadgets, let $c_j$ be one of the constraints of $\phi$.
        Since $f$ is a satisfying assignment, there exists at least one vertex among $v^j_1, \ldots, v^j_{s(j)}$ in $\hat{C}^{G_z}_j$, for some $z \in [\kappa]$, mapping to the restriction of $f$ to the variables appearing in $c_j$.
        Let $v^j_\ell$ be one such vertex of minimum index.
        Then, set $c(v^j_\ell) = 1$, while any other vertex $v^j_{\ell'}$, with $\ell' \neq \ell$, receives color $2$.
        Moreover, let $k_w$ receive color $1$ for $w < \ell$ and color $2$ for $\ell \leq w \leq s(c_j)$.
        On the other hand, let $r_w$ receive color $2$ for $w \leq \ell$ and color $1$ for $\ell < w \leq s(c_j)$.   
        
        Lastly, properly color the internal vertices of the equality/palette/difference/exclusion/implication gadgets 
        using~\cref{lem:dc_equality_gadget,lem:dc_palette_gadget,lem:dc_difference_gadget,lem:dc_exclusion_gadget,lem:dc_implication_gadget}.
        To see that all gadgets are properly colored using these lemmas,
        observe that any vertex colored so far has at most $1$ same-colored neighbor, while the following hold:
        \begin{itemize}
            \item $P = \setdef{p^i}{i \in [\chid]}$ consists of $\chid$ vertices,
            each receiving a distinct color,
            and for all $i \in [\chid]$,
            $c(p^i) = c(p^i_l)$,
            
            \item in all block gadgets, $c(a) = c(a')$,
            $c(a) \neq c(b)$,
            and vertices $x,y$ are colored either $c(a)$ or $c(b)$,
            
            \item for $z \in [\kappa]$ and $j \in [m]$, in constraint gadget $\hat{C}^{G_z}_j$ it holds that
            $c(r_1) = c(p^2)$,
            $c(k_{s(c_j)}) = c(p^2)$,
            $c(k_w) \neq c(r_{w+1})$ for $w \in [s(c_j)-1]$,
            vertices $k_w, r_w, v^j_w$ for $w \in [s(c_j)]$ are colored either $c(p^1)$ or $c(p^2)$,
            and lastly if $c(v^j_\ell) = c(p^1)$ and $x_i$ denotes a variable appearing in $c_j$ such that
            vertex $v^j_\ell$ corresponds to an assignment where $x_i$ receives value $s \in [2 \chid]$
            (which is a restriction of assignment $f$),
            then
            \begin{itemize}
                \item if $s \leq \chid$, $c(a) = s $ and $c(x) = s$,
                \item else $c(a) = s'$ and $c(x) \neq s'$, for $s' = s - \chid$,
            \end{itemize}
            where vertices $a,x$ belong to $\hat{B}^{G_z}_{i,j}$.
        \end{itemize}
        This concludes the proof.
    \end{proof}
    
    \begin{lemma}\label{lem:dc_tw_lb_D1_cor2}
        For any $\chid \geq 2$, if $G$ admits a $(\chid,1)$-coloring,
        then $\phi$ is satisfiable.
    \end{lemma}
    
    \begin{proof}
        Let $c : V(G) \rightarrow [\chid]$ be a $(\chid,1)$-coloring of $G$.
        Due to the properties of the equality gadgets, it holds that $c(p^i) = c(p^i_l)$,
        for all $i \in [\chid]$.
        Since $c$ is a $(\chid,1)$-coloring, it follows that $c(p^i) \neq c(p^j)$, for distinct $i,j \in [\chid]$.
        Assume without loss of generality that $c(p^i) = i$, for $i \in [\chid]$.
    
        For $G_z$, consider a mapping between the coloring of vertices of $\hat{B}^{G_z}_{i,j}$ and the value of $x_i$ for some assignment of the variables of $\phi$.
        In particular, the coloring of both vertices $a$ and $x$ with color $k \in [\chid]$ is mapped with an 
        assignment where $x_i$ receives value $k$,
        while if only $a$ receives color $k$, with an assignment where $x_i$ receives value $\chid + k$.
        We will say that an \emph{inconsistency} occurs in a variable gadget $P^{G_z}_i$ if there exist block gadgets $\hat{B}^{G_z}_{i,j}$ and $\hat{B}^{G_z}_{i,j+1}$,
        such that the coloring of the vertices of each block gadget maps to different assignments of $x_i$.
        We say that $G_z$ is \emph{consistent} if no inconsistency occurs in its variable gadgets $P^{G_z}_i$, for every $i \in [n]$.
    
        \begin{claim}
            There exists $\pi \in [\kappa]$ such that $G_\pi$ is consistent.
        \end{claim}
    
        \begin{claimproof}
            We will prove that every path $\mathcal{P}^i$ may induce at most 1 inconsistency.
            In that case, since there are $n$ such paths and $\kappa = n + 1$ copies of $G_0$,
            due to the pigeonhole principle there exists some $G_\pi$ without any inconsistencies.
    
            Consider a path $\mathcal{P}^i$ as well as a block gadget $\hat{B}^{G_z}_{i,j}$, for some $z \in [\kappa]$ and $j \in [m]$.
            Let $N(\hat{B}^{G_z}_{i,j})$ denote the block gadget right of $\hat{B}^{G_z}_{i,j}$,
            i.e.~vertex $a'$ of $\hat{B}^{G_z}_{i,j}$ coincides with vertex $a$ of $N(\hat{B}^{G_z}_{i,j})$.
            Moreover, let $\hat{B}^{G_{z'}}_{i,j'}$, where either a) $z' = z$ and $j' > j$ or b) $z' > z$ and $j' \in [m]$,
            denote some block gadget which appears to the right of $\hat{B}^{G_z}_{i,j}$.
            For every block gadget, due to the properties of the equality gadget,
            it holds that $c(a) = c(a')$,
            therefore the color of vertex $a$ is the same for all block gadgets belonging to the same path $\mathcal{P}^i$.
            Consider the following two cases regarding the vertices of $\hat{B}^{G_z}_{i,j}$:
            \begin{itemize}
                \item If $c(a) \neq c(x)$, it follows that $c(a) = c(y)$, since alternatively $b$ would have $2$ same colored neighbors.
                Then, it holds that $a'$, which is the vertex $a$ of $N(\hat{B}^{G_z}_{i,j})$, has a same colored neighbor in $\hat{B}^{G_z}_{i,j}$,
                thus for the vertex $x$ of $N(\hat{B}^{G_z}_{i,j})$ it follows that $c(x) \neq c(a)$, and inductively,
                it follows that $c(a) \neq c(x)$ for all block gadgets $\hat{B}^{G_{z'}}_{i,j'}$.
    
                \item If $c(a) = c(x)$, it follows that $c(a) \neq c(y)$, since $c(x)$ cannot have two same colored neighbors.
                Then, it holds that $a'$, which is the vertex $a$ of $N(\hat{B}^{G_z}_{i,j})$, has no same colored neighbor in $\hat{B}^{G_z}_{i,j}$.
                Consequently, for the vertices $a$ and $x$ of $N(\hat{B}^{G_z}_{i,j})$ it follows that either $c(a) = c(x)$ or $c(a) \neq c(x)$.
                The same holds for all block gadgets $\hat{B}^{G_{z'}}_{i,j'}$.
            \end{itemize}
            Thus, it follows that every path can induce at most $1$ inconsistency, and since there is a total of $n$ paths,
            there exists a copy $G_\pi$ which is consistent.
        \end{claimproof}

        Consider an assignment $f : X \rightarrow Y$ as follows.
        Let $a$ and $x$ denote vertices of the block gadget $\hat{B}^{G_\pi}_{i,j}$,
        where $c(a) = k \in [\chid]$.
        Then, set $f(x_i) = k$ if $c(x) = k$, and $f(x_i) = \chid + k$ otherwise.
        Notice that one of the above cases holds for every block gadget,
        thus all variables $x_i$ are assigned a value and $f$ is well defined.
    
        It remains to prove that this assignment satisfies all constraints.
        Consider the constraint gadget $\hat{C}^{G_\pi}_j$, where $j \in [m]$.
        We first prove that $c(v^j_\ell) = 1$, for some $\ell \in [s(j)]$.
        Assume that this is not the case.
        Then it follows that every vertex $k_w$ has two neighbors of color $2$,
        consequently $c(k_w) = 1$, for all $w \in [s(j)]$.
        However, due to $Q(p^2, k_{s(j)})$, it follows that $c(k_{s(j)}) = 2$, which is a contradiction.
        Let $v^j_\ell$ such that $c(v^j_\ell) = 1$.
        In that case, due to the implication/exclusion gadgets involving $v^j_\ell$,
        it follows that, if variable $x_i$ is involved in the constraint $c_j$ and
        $v^j_\ell$ corresponds to an assignment where $x_i$ has value $s \in Y$,
        then
        \begin{romanenumerate}
            \item if $s = k$, where $k \in [\chid]$, then $c(a) = k = c(x)$,
            \item if $s = \chid + k$, where $k \in [\chid]$, then $c(a) = k \neq c(x)$,
        \end{romanenumerate}
        where vertices $a$ and $x$ belong to $\hat{B}^{G_\pi}_{i,j}$.
        However, in that case, the assignment that corresponds to $v^j_\ell$ is a restriction of $f$,
        thus $f$ satisfies the constraint $c_j$.
        Since $j = 1, \ldots, m$ was arbitrary, this concludes the proof that $f$ is a satisfying assignment for $\phi$.
    \end{proof}
    
    \begin{lemma}\label{lem:dc_tw_lb_D1_pw}
        It holds that $\pw(G) \leq n + \bO(1)$.
    \end{lemma}
    
    \begin{proof}
        Due to~\cref{dc_gadgets_pw}, it holds that $\pw(G) = \pw(G' - P) + 3\chid$,
        where $G'$ is the graph we obtain from $G$ by removing all the equality/palette/difference/exclusion/implication gadgets
        and add all edges between their endpoints which are not already connected.
        It therefore suffices to show that $\pw(G' - P) = n + \bO(1)$.
    
        We will do so by providing a mixed search strategy to clean $G' - P$ using at most this many searchers simultaneously.
        Since for the mixed search number \ms{} it holds that $\pw(G' - P) \leq \ms(G' - P) \leq \pw(G' - P) + 1$,
        we will show that $\ms(G' - P) \leq n + 5 + B^q$ and the statement will follow.
    
        Start with graph $G_1$.
        Place $s(c_1) + 1$ searchers to the vertices $r_1$ and $v^1_w$ of $\hat{C}^{G_1}_1$, for $w \in [s(c_1)]$,
        as well as $n$ searchers on vertices $a$ of block gadgets $\hat{B}^{G_1}_{i,1}$, for $i \in [n]$.
        By moving the searcher placed on $r_1$ along the path formed by $k_1, r_2, k_2, \ldots$, all the edges of the constraint gadget can be cleaned.
        Next we will describe the procedure to clean $\hat{B}^{G_1}_{i,1}$.
        Move four extra searchers to all other vertices of $\hat{B}^{G_1}_{i,1}$, namely $x, y, b, a'$.
        Afterwards, remove the searchers from vertices $a,x,y,b$.
        Repeat the whole procedure for all $i \in [n]$.
    
        In order to clean the rest of the graph, we first move the searchers from $\hat{C}^{G_z}_j$ to $\hat{C}^{G_z}_{j+1}$ if $j < m$
        or to $\hat{C}^{G_{z+1}}_1$ alternatively (possibly introducing new searchers if required),
        clean the latter, and then proceed by cleaning the corresponding block gadgets.
        By repeating this procedure, in the end we clean all the edges of $G'-P$ by using at most $n + 5 + B^q = n + \bO (1)$ searchers.
    \end{proof}
    
    Therefore, in polynomial time, we can construct a graph $G$,
    of pathwidth $\pw(G) \leq n + \bO(1)$ due to~\cref{lem:dc_tw_lb_D1_pw},
    such that, due to~\cref{lem:dc_tw_lb_D1_cor1,lem:dc_tw_lb_D1_cor2},
    deciding whether $G$ admits a $(\chid,1)$-coloring is equivalent to deciding whether $\phi$ is satisfiable.
    In that case, assuming there exists a $\sO((2 \chid - \varepsilon)^{\pw(G)})$ algorithm for \DC{} for $\chid$ colors and $\Delta = 1$,
    then for $B = 2 \chid$,
    one could decide $q$-CSP-$B$ in time $\sO((2 \chid - \varepsilon)^{\pw(G)}) = \sO((B-\varepsilon)^{n + \bO(1)}) = \sO((B-\varepsilon)^n)$,
    which contradicts the SETH due to~\cref{thm:q_CSP_B_SETH}.
\end{proof}

\begin{theorem}
    For any constant $\varepsilon > 0$ and for any fixed $\chid \geq 2$, $\Delta \geq 2$,
    there is no $\sO((\chid \cdot (\Delta + 1) - \varepsilon)^{\pw})$ algorithm
    deciding whether $G$ admits a $(\chid,\Delta)$-coloring,
    where \pw{} denotes the graph's pathwidth,
    unless the SETH is false.
\end{theorem}

\begin{proof}
    Fix some positive $\varepsilon > 0$ for which we want to prove the theorem.
    Let $B = \chid \cdot (\Delta+1)$.
    We will reduce $q$-CSP-$B$, for some $q$ that is a constant that only depends on $\varepsilon$,
    to \DC{} for maximum degree $\Delta \geq 2$ and $\chid$ colors in a way that
    ensures that if the resulting \DC{} instance could be solved in time $\sO((\chid \cdot (\Delta+1) - \varepsilon)^{\pw})$,
    then we would obtain an algorithm for $q$-CSP-$B$ that would contradict the SETH due to~\cref{thm:q_CSP_B_SETH}.
    To this end, let $\phi$ be an instance of $q$-CSP-$B$ of $n$ variables $X = \setdef{x_i}{i \in [n]}$ taking values over the set $Y = [0, B-1]$
    and $m$ constraints $C = \setdef{c_j}{j \in [m]}$.
    For each constraint we are given a set of at most $q$ variables which are involved in this constraint and a list of satisfying assignments for these variables,
    the size of which is denoted by $s : C \rightarrow [B^q]$,
    i.e.~$s(c_j) \leq B^q = \bO(1)$ denotes the number of satisfying assignments for constraint $c_j$.
    We will construct in polynomial time an equivalent instance $G$ of \DC{} for $\Delta \geq 2$ and $\chid$ colors,
    where $\pw(G) \leq n + \bO (1)$.
    
    Since we will repeatedly use the equality, difference, and palette gadgets (see~\cref{sec:dc_coloring_gadgets}),
    we will use the following convention:
    whenever $v_1, v_2, v_3$ are vertices we have already introduced to $G$,
    when we say that we add an equality gadget $Q(v_1, v_2)$, a difference gadget $D(v_1, v_2, \delta)$ or a palette gadget $P(v_1, v_2, v_3)$,
    this means that we add to $G$ a copy of $Q(u_1, u_2, \chid, \Delta)$, of $D(u_1, u_2, \chid, \Delta, \delta)$ or of $P(u_1, u_2, u_3, \chid, \Delta)$ respectively
    and then identify $u_1, u_2(,u_3)$ with $v_1, v_2(,v_3)$ respectively.
    
    \proofsubparagraph*{Palette Vertices.}
    Construct a clique of $\chid$ vertices $P = \setdef{p^i}{i \in [\chid]}$.
    For $i \in [\chid]$, attach to vertex $p^i$ leaves $p^i_l$, for $l \in [\Delta]$,
    and add equality gadgets $Q(p^i, p^i_l)$.
    
    Whenever $v_1, v_2$ are vertices we have already introduced to $G$,
    when we say that we add an exclusion gadget $E(v_1, v_2, v'_1, v'_2)$ or an implication gadget $I(v_1, v_2,v'_1, v'_2)$,
    this means that we add to $G$ a copy of $E(u_1, u_2, c_1, c_2, C, \chid, \Delta)$ or of $I(u_1, u_2, c_1, c_2, C, \chid, \Delta)$ respectively
    and then identify $u_1, u_2, c_i$ with $v_1, v_2, p^i$ respectively, for all $i \in [\chid]$.
    
    \proofsubparagraph*{Block and Variable Gadgets.}
    For every variable $x_i$ and every constraint $c_j$, construct a \emph{block gadget} $\hat{B}_{i,j}$ as depicted in~\cref{fig:dc_tw_lb_t2_block}.
    In order to do so, we introduce vertices $a, a', b_1, b_2, \chi_i$ and $y_i$, for $i \in [\Delta]$.
    Then, we add gadgets $Q(a, b_2)$, $Q(b_2, a')$ and $D(b_1, b_2,\Delta+1)$.
    Finally, we add edges $\braces{a,\chi_i}$, $\braces{a',y_i}$, as well as between vertices $b_1, b_2$ and every vertex $\chi_i, y_i$.
    Moreover, if $\chid \geq 3$, we add palette gadgets $P(b_1, b_2, \chi_i)$ and $P(b_1, b_2, y_i)$, for all $i \in [\Delta]$.
    Next, for $j \in [m-1]$, we serially connect the block gadgets $\hat{B}_{i,j}$ and $\hat{B}_{i,j+1}$ so that the vertex $a'$ of $\hat{B}_{i,j}$ is the vertex $a$ of $\hat{B}_{i,j+1}$,
    thus resulting in $n$ ``paths'' $P_1, \ldots, P_n$ consisting of $m$ serially connected block gadgets, called \emph{variable gadgets}.
    Intuitively, the variable gadget is meant to represent a variable $x_i$ and hence needs to have
    $\chid (\Delta + 1)$ different viable configurations.
    These are made up by deciding on a color for $a$ ($\chid$ choices) and then deciding
    how many vertices of $\setdef{\chi_i}{i \in [\Delta]}$ will receive the same color as $a$ ($\Delta+1$ choices).

    \begin{figure}[ht]
    \centering 
      \begin{subfigure}[b]{0.4\linewidth}
      \centering
        \begin{tikzpicture}[scale=0.75, transform shape]
            
            
            \node[vertex] (a) at (0,5) {};
            \node[] () at (0,5.3) {$a$};
            
            \node[vertex] (x1) at (1.5,6) {};
            \node[] () at (1.5,6.3) {$\chi_1$};
            
            \node[vertex] (xD) at (1.5,4) {};
            \node[] () at (1.5,4.3) {$\chi_\Delta$};
            
            \node[] (dots1) at (1.5,5.1) {$\vdots$};

            \node[vertex] (b1) at (3,5) {};
            \node[] () at (3,5.3) {$b_1$};
            
            \node[vertex] (b2) at (3,3) {};
            \node[] () at (3,2.7) {$b_2$};

            \node[vertex] (y1) at (4.5,6) {};
            \node[] () at (4.5,6.3) {$y_1$};
            
            \node[vertex] (yD) at (4.5,4) {};
            \node[] () at (4.5,4.3) {$y_\Delta$};
            
            \node[] (dots2) at (4.5,5.1) {$\vdots$};
            
            \node[vertex] (a') at (6,5) {};
            \node[] () at (6,5.3) {$a'$};
            
            
            \draw[] (a)--(x1)--(b1)--(y1)--(a');
            \draw[] (a)--(xD)--(b2)--(yD)--(a');
            \draw[] (x1)--(b2)--(y1);
            \draw[] (xD)--(b1)--(yD);
            \draw[dashed] (a) edge [bend right] (b2);
            \draw[dashed] (a') edge [bend left] (b2);
            \draw[dashed] (b1)--(b2);
            
            \path (b1) -- (b2) node[midway] {$\neq$};
            \path (a) -- (b2) node[midway,below=0.7] {$=$};
            \path (a') -- (b2) node[midway,below=0.7] {$=$};

            \end{tikzpicture}
            \caption{Block gadget when $\Delta \geq 2$.}
            \label{fig:dc_tw_lb_t2_block}
        \end{subfigure}
    \begin{subfigure}[b]{0.4\linewidth}
    \centering
        \begin{tikzpicture}[scale=0.6, transform shape]
        
            
            \node[vertex] (a) at (0,5) {};
            
            \node[vertex] (x1) at (1.5,6) {};
            
            \node[vertex] (xD) at (1.5,4) {};
            
            \node[] (dots1) at (1.5,5.1) {$\vdots$};
            
            \node[vertex] (b1) at (3,5) {};
            
            \node[vertex] (b2) at (3,3) {};
                    
            \node[vertex] (y1) at (4.5,6) {};
            
            \node[vertex] (yD) at (4.5,4) {};
            
            \node[] (dots2) at (4.5,5.1) {$\vdots$};
            
            \node[vertex] (a') at (6,5) {};
    
            \node[vertex] (x1') at (7.5,6) {};
            
            \node[vertex] (xD') at (7.5,4) {};
            
            \node[] (dots1') at (7.5,5.1) {$\vdots$};
            
            \node[vertex] (b1') at (9,5) {};
            
            \node[vertex] (b2') at (9,3) {};
                    
            \node[vertex] (y1') at (10.5,6) {};
            
            \node[vertex] (yD') at (10.5,4) {};
            
            \node[] (dots2') at (10.5,5.1) {$\vdots$};
            
            \node[vertex] (a2) at (12,5) {};

            
            \draw[] (a)--(x1)--(b1)--(y1)--(a');
            \draw[] (a)--(xD)--(b2)--(yD)--(a');
            \draw[] (x1)--(b2)--(y1);
            \draw[] (xD)--(b1)--(yD);
            \draw[dashed] (a) edge [bend right] (b2);
            \draw[dashed] (a') edge [bend left] (b2);
            \draw[dashed] (b1)--(b2);
    
            \draw[] (a')--(x1')--(b1')--(y1')--(a2);
            \draw[] (a')--(xD')--(b2')--(yD')--(a2);
            \draw[] (x1')--(b2')--(y1');
            \draw[] (xD')--(b1')--(yD');
            \draw[dashed] (a') edge [bend right] (b2');
            \draw[dashed] (a2) edge [bend left] (b2');
            \draw[dashed] (b1')--(b2');
            
            \path (b1) -- (b2) node[midway] {$\neq$};
            \path (a) -- (b2) node[midway,below=0.7] {$=$};
            \path (a') -- (b2) node[midway,below=0.7] {$=$};
    
            \path (b1') -- (b2') node[midway] {$\neq$};
            \path (a') -- (b2') node[midway,below=0.7] {$=$};
            \path (a2) -- (b2') node[midway,below=0.7] {$=$};
        \end{tikzpicture}
        \caption{Serially connected block gadgets.}
        \label{fig:dc_tw_lb_t2_variable}
      \end{subfigure}
    \caption{Variable gadgets are comprised of serially connected block gadgets.}
    \label{fig:dc_tw_lb_t2_block_and_variable}
    \end{figure}
    
    \proofsubparagraph*{Constraint Gadget.}
    This gadget is responsible for determining constraint satisfaction,
    based on the choices made in the rest of the graph.
    For constraint $c_j$, construct the \emph{constraint gadget} $\hat{C}_j$ as depicted in~\cref{fig:dc_tw_lb_c_toul_3_constraint}:
    \begin{itemize}
        \item introduce vertices $r_w$, where $w \in [s(c_j)]$,
        and add $Q(p^2, r_1)$, as well as $P(p^1,p^2,r_w)$ when $\chid \geq 3$, for $w \in [2, s(j)]$
    
        \item for $w \in [s(j)]$, introduce vertices $v^j_w$,
        as well as palette gadgets $P(p^1,p^2,v^j_w)$ when $\chid \geq 3$,
        and fix an arbitrary one-to-one mapping between those vertices and the satisfying assignments of $c_j$,
    
        \item introduce vertices $k_w$,
        where $w \in [s(c_j)]$, and add gadgets $D(p^2,k_w,\Delta-1)$, $Q(p^2, k_{s(c_j)})$,
        as well as $P(p^1, p^2, k_w)$ when $\chid \geq 3$ for $w \in [s(j)-1]$
    
        \item add edges between vertex $k_w$ and vertices $r_w, v^j_w$, as well as $D(k_w, r_{w+1}, \Delta+1)$,
    
        \item if variable $x_i$ is involved in the constraint $c_j$ and $v^j_\ell$ corresponds to an assignment where $x_i$ has value $s \in Y$,
        where $s = (\Delta + 1) \cdot (k-1) + \delta$, for $k \in [\chid]$ and $\delta \in [0, \Delta]$,
        then add implication gadgets $I(v^j_\ell, a, p^1, p^k)$, $I(v^j_\ell, \chi_{i_1}, p^1, p^k)$ and
        exclusion gadgets $E(v^j_\ell, \chi_{i_2}, p^1, p^k)$,
        for $i_1 \in [\delta]$ and $i_2 \in [\delta+1, \Delta]$,
        where vertices $a$ and $\chi$ belong to $\hat{B}_{i,j}$.
    \end{itemize}
    Intuitively, the constraint gadget is set up in a way that forces, for some $\ell \in [s(c_j)]$,
    vertex $v^j_\ell$ to receive color $1$,
    which in turn ``activates'' the implication and exclusion gadgets we have added to this vertex.
    This ensures that the assignment encoded by the variable gadgets agrees with the satisfying assignment
    of $c_j$ represented by $v^j_\ell$.

    \begin{figure}[ht]
    \centering
    \begin{tikzpicture}[transform shape]
    
    
    \node[vertex] (r) at (0,1) {};
    \node[above =0.2 of r] {$r_1$};
    
    \node[black_vertex, right = of r] (k1) {};
    \node[below right =0.01 of k1] {$k_1$};
    
    \node[vertex, below = of k1] (c1) {};
    \node[right =0.01 of c1] {$v^j_1$};
    
    \node[vertex, right = of k1] (r1) {};
    \node[above =0.2 of r1] {$r_2$};
    \path (k1) -- (r1) node[midway,above] {$\neq$};
    
    \node[black_vertex, right = of r1] (k2) {};
    \node[below right =0.01 of k2] {$k_2$};
    
    \node[vertex, below = of k2] (c2) {};
    \node[right =0.01 of c2] {$v^j_2$};
    
    \node[vertex, right =2 of k2] (rm) {};
    \node[above =0.2 of rm] {$r_{s(c_j)}$};
    
    \node[black_vertex, right = of rm] (km) {};
    \node[below right =0.01 of km] {$k_{s(c_j)}$};
    
    \node[vertex, below = of km] (cm) {};
    \node[right =0.01 of cm] {$v^j_{s(c_j)}$};
    
    \draw[] (r)--(k1);
    \draw[] (c1)--(k1);
    
    \draw[dashed] (k1)--(r1);
    
    \draw[] (r1)--(k2);
    \draw[] (c2)--(k2);
    
    \draw[dotted] (k2)--(rm);
    
    \draw[] (rm)--(km);
    \draw[] (cm)--(km);
    
    \end{tikzpicture}
    \caption{Constraint gadget.
    For every black vertex $k_w$, there exists a gadget $D(p^2, k_w, \Delta-1)$.}
    \label{fig:dc_tw_lb_c_toul_3_constraint}
    \end{figure}
    
    Let graph $G_0$ correspond to the graph containing all variable gadgets $P_i$ as well as all the constraint gadgets $\hat{C}_j$,
    for $i \in [n]$ and $j \in [m]$.
    We refer to the block gadget $\hat{B}_{i,j}$, to the variable gadget $P_i$,
    and to the constraint gadget $\hat{C}_j$ of $G_z$
    as $\hat{B}^{G_z}_{i,j}$, $P^{G_z}_i$, and $\hat{C}^{G_z}_j$ respectively.
    To construct graph $G$, introduce $\kappa = n \cdot \Delta + 1$ copies $G_1, \ldots, G_\kappa$ of $G_0$,
    such that they are connected sequentially as follows:
    for $i \in [n]$ and $j \in [\kappa - 1]$, 
    the vertex $a'$ of $\hat{B}^{G_j}_{i,m}$ is the vertex $a$ of $\hat{B}^{G_{j+1}}_{i,1}$.
    Let $\mathcal{P}^i$ denote the ``path'' resulting from $P^{G_1}_i, \ldots, P^{G_\kappa}_i$.
        
    \begin{lemma}\label{lem:dc_tw_lb_D_toul_2_cor1}
        For any $\chid \geq 2$, if $\phi$ is satisfiable, then $G$ admits a $(\chid,\Delta)$-coloring.
    \end{lemma}
    
    \begin{proof}
        Let $f : X \rightarrow Y$ denote an assignment which satisfies all the constraints $c_1, \ldots, c_m$.
        We will describe a $(\chid, \Delta)$-coloring $c : V(G) \rightarrow [\chid]$ of $G$.
    
        Let $c(p^i) = c(p^i_l) = i$, for $i \in [\chid]$ and $l \in [\Delta]$.
        Next, for the vertices of block gadget $\hat{B}^{G_z}_{i,j}$, where $z \in [\kappa]$, $i \in [n]$ and $j \in [m]$,
        if $f(x_i) = (\Delta + 1) \cdot (k-1) + \delta$, for $k \in [\chid]$ and $\delta \in [0, \Delta]$,
        then let
        \begin{itemize}
            \item $c(a) = c(\chi_{i_1}) = c(y_{i_2}) = k$, where $i_1 \in [\delta]$ and $i_2 \in [\delta+1, \Delta]$,
            \item $c(b) = c(\chi_{i_1}) = c(y_{i_2}) = k'$, for some arbitrary $k' \in [\chid] \setminus \{ k \}$,
            where $i_1 \in [\delta+1, \Delta]$ and $i_2 \in [\delta]$.
        \end{itemize}
        Regarding the constraint gadgets, let $c_j$ be one of the constraints of $\phi$.
        Let the $\Delta-1$ leaves attached to vertex $k_w$ receive color $2$, for $w \in [s(c_j)]$. 
        Since $f$ is a satisfying assignment, there exists at least one vertex among $v^j_1, \ldots, v^j_{s(j)}$ in $\hat{C}^{G_z}_j$,
        for $z \in [\kappa]$, mapped to a restriction of $f$.
        Let $v^j_\ell$ be one such vertex of minimum index.
        Then, let $c(v^j_\ell) = 1$, while any other vertex $v^j_{\ell'}$, with $\ell' \neq \ell$, receives color $2$.
        Moreover, let $k_w$ receive color $1$ for $w < \ell$ and color $2$ for $\ell \leq w \leq s(c_j)$.
        On the other hand, let $r_w$ receive color $2$ for $w \leq \ell$ and color $1$ for $\ell < w \leq s(c_j)$.

        Lastly, properly color the internal vertices of the equality/palette/difference/exclusion/implication gadgets 
        using~\cref{lem:dc_equality_gadget,lem:dc_palette_gadget,lem:dc_difference_gadget,lem:dc_exclusion_gadget,lem:dc_implication_gadget}.
        To see that all gadgets are properly colored using these lemmas,
        observe that any vertex colored so far has at most $\Delta$ same-colored neighbors, while the following hold:
        \begin{itemize}
            \item $P = \setdef{p^i}{i \in [\chid]}$ consists of $\chid$ vertices,
            each receiving a distinct color,
            and for all $i \in [\chid]$, $l \in [\Delta]$,
            $c(p^i) = c(p^i_l)$,
            
            \item in all block gadgets, $c(a) = c(b_2)$,
            $c(b_2) = c(a')$,            
            $c(b_1) \neq c(b_2)$,
            and vertices $\chi_i,y_i$ for $i \in [\Delta]$ are colored either $c(b_1)$ or $c(b_2)$,
            
            \item for $z \in [\kappa]$ and $j \in [m]$, in constraint gadget $\hat{C}^{G_z}_j$ it holds that
            $c(r_1) = c(p^2)$,
            $c(k_{s(c_j)}) = c(p^2)$,
            $c(k_w) \neq c(r_{w+1})$ for $w \in [s(c_j)-1]$,
            vertices $k_w, r_w, v^j_w$ for $w \in [s(c_j)]$ are colored either $c(p^1)$ or $c(p^2)$
            while the leaves attached to $k_w$ all receive color $c(p^2)$,
            and lastly if $c(v^j_\ell) = c(p^1)$ and $x_i$ denotes a variable appearing in $c_j$ such that
            vertex $v^j_\ell$ corresponds to an assignment where $x_i$ receives value
            $s = (\Delta+1) \cdot (k-1) + \delta$ for $k \in [\chid]$ and $\delta \in [0, \Delta]$
            (which is a restriction of assignment $f$),
            then $c(a) = c(\chi_{i_1}) = k$ and $c(\chi_{i_2}) \neq k$, for $i_1 \in [\delta]$ and $i_2 \in [\delta+1, \Delta]$,
            where vertices $a,\chi_{i_1}, \chi_{i_2}$ belong to $\hat{B}^{G_z}_{i,j}$.
        \end{itemize}
        This concludes the proof.
    \end{proof}
    
    \begin{lemma}\label{lem:dc_tw_lb_D_toul_2_cor2}
        For any $\chid \geq 2$, if $G$ admits a $(\chid,\Delta)$-coloring,
        then $\phi$ is satisfiable.
    \end{lemma}
    
    \begin{proof}
        Let $c : V(G) \rightarrow [\chid]$ be a $(\chid,\Delta)$-coloring of $G$.
        Due to the properties of the equality gadgets, it holds that $c(p^i) = c(p^i_l)$,
        for all $i \in [\chid], l \in [\Delta]$.
        Since $c$ is a $(\chid,\Delta)$-coloring, it follows that $c(p^i) \neq c(p^j)$, for distinct $i,j \in [\chid]$.
        Assume without loss of generality that $c(p^i) = i$, for $i \in [\chid]$.
    
        For $G_z$, consider a mapping between the coloring of vertices of $\hat{B}^{G_z}_{i,j}$ and the value of $x_i$ for some assignment of the variables of $\phi$.
        In particular, the coloring of vertex $a$ as well as of $\delta \in [0,\Delta]$ vertices of $\setdef{\chi_i}{i \in [\Delta]}$ with color $k \in [\chid]$
        is mapped to an assignment where $x_i$ receives value $(\Delta + 1) \cdot (k-1) + \delta$.    
        We will say that an \emph{inconsistency} occurs in a variable gadget $P^{G_z}_i$ if there exist block gadgets $\hat{B}^{G_z}_{i,j}$ and $\hat{B}^{G_z}_{i,j+1}$,
        such that the coloring of the vertices of each block gadget maps to different assignments of $x_i$.
        We say that $G_z$ is \emph{consistent} if no inconsistency occurs in its variable gadgets $P^{G_z}_i$, for every $i \in [n]$.
    
        Notice that, for the vertices of the block gadget $\hat{B}^{G_z}_{i,j}$,
        it holds that vertices $\chi_{i_1}, y_{i_2}$, for $i_1,i_2 \in [\Delta]$,
        are colored either $c(a)$ or $c(b_1)$,
        and since there are $2 \Delta$ such vertices in total, exactly half of them are colored with each color,
        since otherwise either $b_1$ or $b_2$ have more than $\Delta$ same colored neighbors.
        We will now prove the following claim.
    
        \begin{claim}
            There exists $\pi \in [\kappa]$ such that $G_\pi$ is consistent.
        \end{claim}
    
        \begin{claimproof}
            We will prove that every path $\mathcal{P}^i$ may induce at most $\Delta$ inconsistencies.
            In that case, since there are $n$ such paths and $\kappa = n \cdot \Delta + 1$ copies of $G_0$,
            due to the pigeonhole principle there exists some $G_\pi$ without any inconsistencies.
    
            Consider a path $\mathcal{P}^i$ as well as a block gadget $\hat{B}^{G_z}_{i,j}$, for some $z \in [\kappa]$ and $j \in [m]$.
            Let $N(\hat{B}^{G_z}_{i,j})$ denote the block gadget right of $\hat{B}^{G_z}_{i,j}$,
            i.e.~vertex $a'$ of $\hat{B}^{G_z}_{i,j}$ coincides with vertex $a$ of $N(\hat{B}^{G_z}_{i,j})$.
            Moreover, let $\hat{B}^{G_{z'}}_{i,j'}$, where either a) $z' = z$ and $j' > j$ or b) $z' > z$ and $j' \in [m]$,
            denote some block gadget which appears to the right of $\hat{B}^{G_z}_{i,j}$.
            For every block gadget, due to the properties of the equality gadget,
            it holds that $c(a) = c(a')$,
            therefore the color of vertex $a$ is the same for all block gadgets belonging to the same path $\mathcal{P}^i$.
            For the vertices of $\hat{B}^{G_z}_{i,j}$, assume that exactly $\delta$ vertices $\chi_{i_1}$ are colored with color $c(a)$.
            We will prove that then, in every gadget $\hat{B}^{G_{z'}}_{i,j'}$, at most $\delta$ vertices $\chi_{i_1}$ are colored with color $c(a)$.
            For the base of the induction, notice that in $\hat{B}^{G_z}_{i,j}$, exactly $\Delta - \delta$ vertices $y_{i_2}$ are colored with color $c(a)$.
            Thus, in $N(\hat{B}^{G_z}_{i,j})$, at most $\delta$ vertices $\chi_{i_1}$ receive color $c(a)$,
            since vertex $a'$ of $\hat{B}^{G_z}_{i,j}$ coincides with vertex $a$ of $N(\hat{B}^{G_z}_{i,j})$.
            Assume that this is the case for some gadget $\hat{B}^{G_{z'}}_{i,j'}$ to the right of $\hat{B}^{G_z}_{i,j}$.
            Then, since there are at least $\Delta - \delta$ vertices $y_{i_2}$ of $\hat{B}^{G_z}_{i,j}$ receiving color $c(a)$,
            it follows that there are at most $\delta$ vertices $\chi_{i_1}$ of $N(\hat{B}^{G_{z'}}_{i,j'})$ receiving color $c(a)$.
            Consequently, it follows that every path can induce at most $\Delta$ inconsistencies, and since there is a total of $n$ paths,
            there exists a copy $G_\pi$ which is consistent.
        \end{claimproof}
    
        Let $f : X \rightarrow Y$ be an assignment such that $f(x_i) = (\Delta + 1) \cdot (k-1) + \delta$,
        where, for the vertices of the block gadget $\hat{B}^{G_\pi}_{i,j}$,
        it holds that $c(a) = k \in [\chid]$ and exactly $\delta \in [0, \Delta]$ vertices of $\setdef{\chi_i}{i \in [\Delta]}$ are of color $k$.
    
        It remains to prove that this is an assignment that satisfies all constraints.
        Consider the constraint gadget $\hat{C}^{G_\pi}_j$, where $j \in [m]$.
        We first prove that $c(v^j_\ell) = 1$, for some $\ell \in [s(j)]$.
        Assume that this is not the case.
        Then it follows that every vertex $k_w$ has $\Delta+1$ neighbors of color $2$
        (remember that due to $D(p^2, k_i, \Delta-1)$,
        $k_w$ has $\Delta - 1$ neighboring leaves of color $c(p^2)$),
        consequently $c(k_w) \neq 2$, for every $w \in [s(j)]$.
        However, due to $Q(p^2, k_{s(j)})$, it follows that $c(k_{s(j)}) = 2$, which is a contradiction.
        Let $v^j_\ell$ such that $c(v^j_\ell) = 1$.
        In that case, due to the implication and exclusion gadgets involving $v^j_\ell$,
        it follows that, if variable $x_i$ is involved in the constraint $c_j$ and
        $v^j_\ell$ corresponds to an assignment where $x_i$ has, for $k \in [\chid]$ and $\delta \in [0, \Delta]$,
        value $(\Delta + 1) \cdot (k-1) + \delta$,
        then in $\hat{B}^{G_\pi}_{i,j}$, vertex $a$ as well as exactly $\delta$ vertices of $\setdef{\chi_i}{i \in [\Delta]}$ have color $k$.
        In that case, the assignment corresponding to $v^j_\ell$ is a restriction of $f$,
        thus $f$ satisfies constraint $c_j$.
        Since $j = 1, \ldots, m$ was arbitrary, this concludes the proof that $f$ is a satisfying assignment for $\phi$.
    \end{proof}
    
    \begin{lemma}\label{lem:dc_tw_lb_D_toul_2_pw}
        It holds that $\pw(G) \leq n + \bO(1)$.
    \end{lemma}
    
    \begin{proof}
        Due to~\cref{dc_gadgets_pw}, it holds that $\pw(G) = \pw(G' - P) + 3\chid$,
        where $G'$ is the graph we obtain from $G$ by removing all the equality/palette/difference/exclusion/implication gadgets
        and add all edges between their endpoints which are not already connected.
        It therefore suffices to show that $\pw(G' - P) = n + \bO(1)$.
    
        We will do so by providing a mixed search strategy to clean $G' - P$ using at most this many searchers simultaneously.
        Since for the mixed search number \ms{} it holds that $\pw(G' - P) \leq \ms(G' - P) \leq \pw(G' - P) + 1$,
        we will show that $\ms(G' - P) \leq n + 5 + B^q$ and the statement will follow.
    
        Start with graph $G_1$.
        Place $s(c_1) + 1$ searchers to the vertices $r_1$ and $v^1_w$ of $\hat{C}^{G_1}_1$, for $w \in [s(c_1)]$,
        as well as $n$ searchers on vertices $a$ of block gadgets $\hat{B}^{G_1}_{i,1}$, for $i \in [n]$.
        By moving the searcher placed on $r_1$ along the path formed by $k_1, r_2, k_2, \ldots$, all the edges of the constraint gadget can be cleaned.
        Next we will describe the procedure to clean $\hat{B}^{G_1}_{i,1}$.
        Move four extra searchers to vertices $a', b_1, b_2, \chi_1$ of $\hat{B}^{G_1}_{i,1}$.
        Move the latter searcher to all other vertices $\chi_p$ and $y_p$, thus successfully cleaning $\hat{B}^{G_1}_{i,1}$.
        Lastly, remove all the searchers from $\hat{B}^{G_1}_{i,1}$ apart from the one present on vertex $a'$.
        Repeat the whole procedure for all $i \in [n]$.
    
        In order to clean the rest of the graph, we first move the searchers from $\hat{C}^{G_z}_j$ to $\hat{C}^{G_z}_{j+1}$ if $j < m$
        or to $\hat{C}^{G_{z+1}}_1$ alternatively (possibly introducing new searchers if required),
        clean the latter, and then proceed by cleaning the corresponding block gadgets.
        By repeating this procedure, in the end we clean all the edges of $G'-P$ by using at most $n + 5 + B^q = n + \bO (1)$ searchers.
    \end{proof}
    
    Therefore, in polynomial time, we can construct a graph $G$,
    of pathwidth $\pw(G) \leq n + \bO(1)$ due to~\cref{lem:dc_tw_lb_D_toul_2_pw},
    such that, due to~\cref{lem:dc_tw_lb_D_toul_2_cor1,lem:dc_tw_lb_D_toul_2_cor2},
    deciding whether $G$ admits a $(\chid,\Delta)$-coloring is equivalent to deciding whether $\phi$ is satisfiable.
    In that case, assuming there exists a $\sO((\chid \cdot (\Delta+1) - \varepsilon)^{\pw(G)})$ algorithm for \DC,
    then for $B = \chid \cdot (\Delta+1)$,
    one could decide $q$-CSP-$B$ in time $\sO((\chid \cdot (\Delta+1) -\varepsilon)^{\pw(G)}) = \sO((B-\varepsilon)^{n + \bO(1)}) = \sO((B-\varepsilon)^n)$,
    which contradicts the SETH due to~\cref{thm:q_CSP_B_SETH}.
\end{proof}

\subsubsection{Algorithm}\label{sec:dc_tw_algo}

Here we present an algorithm for \DC{} parameterized by the treewidth of the input graph plus the target degree $\Delta$.
The algorithm uses standard techniques,
and closely follows the approach previously sketched in the $(\chid \Delta)^{\bO (\tw)} n^{\bO (1)}$ algorithm of~\cite{siamdm/BelmonteLM20,dmtcs/BelmonteLM22}.
The novelty is the use of a convolution technique presented in~\cite{tcs/CyganP10} in order to speed up the computation in the case of the join nodes.

\begin{theorem}
    Given an instance $\mathcal{I} = (G, \chid, \Delta)$ of \DC,
    as well as a nice tree decomposition of $G$ of width \tw,
    there exists an algorithm that decides $\mathcal{I}$ in time $\sO ((\chid \cdot (\Delta+1))^{\tw})$.
\end{theorem}

\begin{proof}
To avoid confusion, we will refer to the vertices of the nice tree decomposition as \emph{nodes},
and to the vertices of $G$ as vertices.
Let $\mathcal{T} = (T, \braces{X_t}_{t \in V(T)})$ denote the nice tree decomposition of $G$,
where by $r$ we denote the root node.
For a node $t$ of $T$, let $X^\downarrow_t$ denote the union of all the bags present in the subtree rooted at $t$, including $X_t$.
Moreover, let $s : V(T) \to [\tw+1]$ such that $s(t) = |X_t|$,
i.e.~$s(t)$ denotes the size of the bag $X_t$, and assume that $X_t = \braces{v^t_1, \ldots, v^t_{s(t)}}$.

Assuming there exists a $(\chid, \Delta)$-coloring of $G$, there are $\chid \cdot (\Delta+1)$ different possibilities for every vertex $v \in V(G)$,
since it belongs to one of the $\chid$ color classes with some degree $\delta \in [0,\Delta]$ in the corresponding subgraph.
Therefore, for each node $t$ of $T$, we consider tuples $z^t_i = (w^{t,i}_1, \ldots, w^{t,i}_{|X_t|})$,
where each $w^{t,i}_j$ is a pair $w^{t,i}_j = (c^{t,i}_j, \delta^{t,i}_j)$ such that
$c^{t,i}_j \in [\chid]$ and $\delta^{t,i}_j \in [0, \Delta]$, used to encode this information for vertex $v^t_j$,
for a total of $(\chid \cdot (\Delta+1))^{s(t)}$ tuples per node $t$.

In that case, for node $t$, let $S[t, z^t_i]$, where $i \in [(\chid \cdot (\Delta+1))^{s(t)}]$,
denote the number of $(\chid, \Delta)$-colorings of graph $G[X^\downarrow_t]$,
where vertex $v^t_j$ receives color $c^{t,i}_j$ and has exactly $\delta^{t,i}_j$ same-colored neighbors in $X^\downarrow_t \setminus X_t$.
If $s(t) = 0$, then let $S[t,\emptyset]$ be equal to the $(\chid, \Delta)$-colorings of graph $G[X^\downarrow_t]$.
Then, in order to find the number of $(\chid,\Delta)$-colorings of $G$, it suffices to compute $S[r, \emptyset]$.

Notice that each such tuple $z^t_i$ induces $\chid$ sets of same colored vertices $V^{t,i}_c = \setdef{v^t_j \in X_t}{c^{t,i}_j = c}$, for $c \in [\chid]$.
For a tuple $z^t_i$ to be considered, it must hold that,
\begin{equation}\label{eq:prereq}
    \forall v^t_j \in X_t, v^t_j \in V^{t,i}_c \implies |N(v^t_j) \cap V^{t,i}_c| + \delta^{t,i}_j \leq \Delta,
\end{equation}
or in other words, that if some vertex has $\delta$ same colored neighboring vertices in the subgraph which do not belong to the bag,
it should have at most $\Delta - \delta$ same colored neighbors inside the bag.

\proofsubparagraph*{Leaf Node.} If node $t$ is a leaf, then $X_t = \{ v^t_1 \}$ and
\[
    S[t, (c^{t,i}_1, \delta^{t,i}_1)] =
    \begin{cases}
        1,  & \text{if } \delta^{t,i}_1 = 0,\\
        0,  & \text{otherwise}
    \end{cases}
\]
since no matter what color vertex $v^t_1$ is assigned, $X^\downarrow_t \setminus X_t$ is empty,
thus it cannot have any same colored neighbors.

\proofsubparagraph*{Introduce Node.} Suppose $t$ is an introduce node with child node $t_1$
such that $X_t = X_{t_1} \cup \{ v^{t}_{s(t)} \}$ for $v^{t}_{s(t)} \notin X_{t_1}$,
where for $j \in [s(t_1)]$, vertices $v^t_j$ and $v^{t_1}_j$ coincide.
Consider a tuple $z^{t}_{i} = (w^{t,i}_1, \ldots, w^{t,i}_{s(t)})$ of node $t$.
We will compute the value $S[t, z^t_i]$.
First, we verify that~\cref{eq:prereq} is satisfied
(if not we can put the value $0$ as answer).
Then,
\[
    S[t, z^t_i] =
    \begin{cases}
        S[t_1, z^{t_1}_{i_1}],  & \text{if } \delta^{t,i}_{s(t)} = 0,\\
        0,                      & \text{otherwise},
    \end{cases}
\]
where $z^{t_1}_{i_1} = (w^{t_1,i_1}_1, \ldots, w^{t_1,i_1}_{s(t_1)})$ is the tuple of node $t_1$
where $c^{t,i}_j = c^{t_1,i_1}_j$ and $\delta^{t,i}_j = \delta^{t_1,i_1}_j$, for all $j \in [s(t_1)]$.
Intuitively, vertex $v^{t}_{s(t)}$ cannot have any neighbors in $X^\downarrow_t \setminus X_t$,
therefore the only valid value for $\delta^{t,i}_{s(t)}$ is $0$.

\proofsubparagraph*{Forget Node.} Suppose $t$ is a forget node with child node $t_1$
such that $X_t = X_{t_1} \setminus \{ v^{t_1}_{s(t_1)} \}$ for $v^{t_1}_{s(t_1)} \in X_{t_1}$,
where for $j \in [s(t)]$, vertices $v^t_j$ and $v^{t_1}_j$ coincide.
Consider a tuple $z^{t}_{i} = (w^{t,i}_1, \ldots, w^{t,i}_{s(t)})$ of node $t$.
We will compute the value $S[t, z^t_i]$.
First, we verify that~\cref{eq:prereq} is satisfied
(if not we can put the value $0$ as answer).
Then,
\[
    S[t, z^t_i] = \sum_{i_1 \in \mathcal{R}_1(i)} S[t_1, z^{t_1}_{i_1}],
\]
where $i_1 \in \mathcal{R}_1(i)$ if, for all $j \in [s(t)]$,
\begin{itemize}
    \item $c^{t,i}_j = c^{t_1,i_1}_j$ and
    \item if $v^{t_1}_{s(t_1)} \in N(v^{t_1}_j)$ and $c^{t_1,i_1}_{s(t_1)} = c^{t_1,i_1}_j$,
    then $\delta^{t,i}_j = \delta^{t_1,i_1}_j + 1$, otherwise $\delta^{t,i}_j = \delta^{t_1,i_1}_j$.
\end{itemize}
In this case, we consider all possibilities regarding the forgotten vertex, taking into account how it affects the rest of the vertices;
if it was a same-colored neighbor of another vertex of the bag,
then the latter's same-colored neighbors in the subtree, excluding the bag,
increased by one.

\proofsubparagraph*{Join Node.} Suppose $t$ is a join node with children nodes $t_1,t_2$
such that $X_t = X_{t_1} = X_{t_2}$,
where for $j \in [s(t)]$, vertices $v^t_j, v^{t_1}_j$ and $v^{t_2}_j$ coincide.
Consider a tuple $z^{t}_{i} = (w^{t,i}_1, \ldots, w^{t,i}_{s(t)})$ of node $t$.
We will compute the value $S[t, z^t_i]$.
First, we verify that~\cref{eq:prereq} is satisfied
(if not we can put the value $0$ as answer).
Then,
\[
    S[t, z^t_i] = \sum_{(i_1, i_2) \in \mathcal{R}_2(i)} S[t_1, z^{t_1}_{i_1}] \cdot S[t_2, z^{t_2}_{i_2}],
\]
where $(i_1, i_2) \in \mathcal{R}_2(i)$ if, for all $j \in [s(t)]$,
\begin{itemize}
    \item $c^{t,i}_j = c^{t_1,i_1}_j = c^{t_2,i_2}_j$ and
    \item $\delta^{t_1,i_1}_j + \delta^{t_2,i_2}_j = \delta^{t,i}_j$,
    where $0 \leq \delta^{t_1,i_1}_j, \delta^{t_2,i_2}_j \leq \delta^{t,i}_j$.
\end{itemize}
Intuitively, in a join node, for every vertex in the bag,
we should take into account its same-colored neighbors in both
of its children subtrees, excluding the vertices of the bag,
as well as consider all possibilities regarding how these neighbors are partitioned in those subtrees.

Notice that table $S$ has at most $(\chid \cdot (\Delta+1))^{\tw+1}$ cells.
Moreover, by employing dynamic programming, we can fill all of its cells with a bottom-up approach.
In order to check~\cref{eq:prereq}, $n^{\bO(1)}$ time is required.
Then, each tuple of a leaf or introduce node can be computed in constant time,
while each tuple of forget nodes in time $\chid \cdot (\Delta+1)$.
However, in the case of join nodes, the time required per tuple is $\bO((\Delta+1)^{2\tw})$,
since we need to take into account all possible values the degree of each vertex may have in each child node.
In order to circumvent this, we employ a technique based on FFT introduced in~\cite{tcs/CyganP10}.
This allows us to compute, for a given join node,
the values of the table for all of its $\bO ((\chid \cdot (\Delta + 1))^\tw)$ tuples in time $\sO ((\chid \cdot (\Delta + 1))^\tw)$.

\proofsubparagraph*{Faster Join Computations.}
Let $t$ be a join node with children nodes $t_1$ and $t_2$.
First, we fix a coloring on the vertices of $X_t$, thus choosing among the $\chid^{s(t)} = \bO(\chid^{\tw})$ different choices.
We will describe how to compute $S[t,z^t_i]$ for any tuple $z^t_i$ respecting said coloring.
In the following, for every $z^t_i$ considered, we assume that it respects this coloring.

Let, for tuple $z^t_i$, $\Sigma(z^{t}_i) = \sum_{j=1}^{s(t)} \delta^{t,i}_j$ denote its \emph{sum}.
Since $\delta^{t,i}_j \in [0,\Delta]$ for all $j \in [s(t)]$,
it follows that $\Sigma (z^t_i) \in [0, s(t) \cdot \Delta]$.
For every tuple $z^{t_p}_i$ of $t_1$ and $t_2$, where $p \in \braces{1,2}$,
we will construct an identifier $i(z^{t_p}_{i})$ as follows:
\[
    i(z^{t_p}_{i}) = \sum_{j=1}^{s(t)} (\Delta+1)^{j-1} \cdot \delta^{t_p,i}_j \leq
    \Delta \cdot \sum_{j=1}^{s(t)} (\Delta+1)^{j-1} =
    \Delta \cdot \frac{1 - (\Delta+1)^{s(t)}}{1 - (\Delta+1)} =
    (\Delta+1)^{s(t)} - 1.
\]
    
Next, we introduce polynomials $P^{t_p}_{\Sigma_q} (x)$, where $p \in \braces{1,2}$ and $\Sigma_q \in [0, s(t) \cdot \Delta]$,
such that $P^{t_p}_{\Sigma_q} (x)$ is comprised of monomials $S[t_p, z^{t_p}_i] \cdot x^{i(z^{t_p}_i)}$, for every tuple $z^{t_p}_{i}$ such that
$\Sigma(z^{t_p}_{i}) = \Sigma_q$.
Subsequently, by using FFT,
we compute the polynomial $P^{t_1}_{\Sigma_{1}} \cdot P^{t_2}_{\Sigma_{2}}$, for all $\Sigma_1, \Sigma_2 \in [0, s(t) \cdot \Delta]$.
Since the multiplication of two polynomials of degree $n$ requires $\bO (n \log n)$ time~\cite{issac/Moenck76},
and we perform $(s(t) \cdot \Delta + 1)^2 = \bO(n^2)$ such multiplications on polynomials of degree at most $(\Delta+1)^{s(t)} - 1$,
it follows that in total $\bO(n^2 \cdot (\Delta+1)^{s(t)} \cdot s(t) \log(\Delta + 1)) = \sO((\Delta+1)^{s(t)})$ time is required.

\begin{claim}\label{claim:dc_tw_algo_claim}
    For $z^{t_1}_{i_1}, z^{t_2}_{i_2}$,
    let $i(z^{t_1}_{i_1}) + i(z^{t_2}_{i_2}) = \sum_{j=1}^{s(t)+1} a_j \cdot (\Delta + 1)^{j-1}$, where $0 \leq a_j \leq \Delta$.
    Then, the following are equivalent:
    \begin{itemize}
        \item $\Sigma(z^{t_1}_{i_1}) + \Sigma (z^{t_2}_{i_2}) = \sum_{j=1}^{s(t)} a_j$,
        \item for all $j \in [s(t)]$, $\delta^{t_1,i_1}_j + \delta^{t_2,i_2}_j \leq \Delta$.
    \end{itemize}
\end{claim}

\begin{claimproof}
    Express $i(z^{t_1}_{i_1})$, $i(z^{t_2}_{i_2})$, as well as $\Sigma(z^{t_1}_{i_1}) + \Sigma (z^{t_2}_{i_2})$
    as numbers $d^1,d^2$ and $a$ respectively in the $(\Delta+1)$-ary system,
    where each of their digits $d^1_j, d^2_j, a_j$ is a number between $0$ and $\Delta$.
    Then, $\delta^{t_1,i_1}_j$, $\delta^{t_2,i_2}_j$, and $a_j$ correspond to the $j$-th digit
    of $d^1$, $d^2$, and $a$ respectively.

    Now, assume that we add the numbers $d^1$ and $d^2$,
    and let $s_j = d^1_j + d^2_j$ if $d^1_j + d^2_j \leq \Delta$, while $s_j = 1 + d^1_j + d^2_j - (\Delta + 1) < d^1_j + d^2_j$ otherwise, where we assumed that $\Delta > 0$.
    Notice that it holds that $\sum_{j=1}^{s(t)} s_j \geq \sum_{j=1}^{s(t)} a_j$.
    In that case, assuming $\Sigma(z^{t_1}_{i_1}) + \Sigma (z^{t_2}_{i_2}) = \sum_{j=1}^{s(t)} a_j$ implies that
    $\sum_{j=1}^{s(t)} d^1_j + \sum_{j=1}^{s(t)} d^2_j \leq \sum_{j=1}^{s(t)} s_j$,
    which in turn implies that $s_j = d^1_j + d^2_j \implies d^1_j + d^2_j \leq \Delta$ for all $j$.

    On the other hand, if $d^1_j + d^2_j \leq \Delta$ for every $j$,
    it follows that $a_j = d^1_j + d^2_j$,
    and thus $\sum_{j=1}^{s(t)} d^1_j + \sum_{j=1}^{s(t)} d^2_j = \sum_{j=1}^{s(t)} a_j \implies \Sigma(z^{t_1}_{i_1}) + \Sigma (z^{t_2}_{i_2}) = \sum_{j=1}^{s(t)} a_j$.
\end{claimproof}

As a consequence of~\cref{claim:dc_tw_algo_claim}, we can easily distinguish whether
a monomial of $P^{t_1}_{\Sigma_{q_1}} \cdot P^{t_2}_{\Sigma_{q_2}}$ corresponds to a tuple of $t$
occurring from the addition of tuples of $t_1$ and $t_2$ of sum $\Sigma_{q_1}$ and $\Sigma_{q_2}$ respectively.
Moreover, any such tuple of $t$ is encoded by some monomial, whose coefficient dictates
the number of different ways this tuple can occur from pairs of tuples of such sums.
Therefore, in order to identify the number of ways a tuple of $t$ may occur,
it suffices to add all the coefficients of the corresponding monomial in all $\bO(n^2)$ polynomial multiplications performed.

Lastly, in order to compute the value of $S$ for the rest of the tuples of $t$,
it suffices to repeat the whole procedure for all different colorings of $X_t$,
thus resulting in $\chid^{s(t)}$ iterations.

In the end, in order to compute the value of $S$ for all $\bO ((\chid (\Delta+1))^\tw)$ tuples of $t$,
we need $\sO (\chid^\tw \cdot (\Delta + 1)^\tw)$ time.

\proofsubparagraph*{Complexity.}
For the final complexity of the algorithm, notice that, for a node $t$,
it holds that in order to compute the value of $S$ for all of its $\bO((\chid \cdot (\Delta+1)^\tw)$ tuples,
we require time:
\begin{itemize}
    \item $\bO (1)$ per tuple, if $t$ is a leaf or an introduction node,
    
    \item $\bO (\chid (\Delta + 1))$ per tuple, if $t$ is a forget node,
    
    \item $\sO ((\chid (\Delta + 1))^\tw)$ for all tuples if $t$ is a join node.
\end{itemize}
Therefore, the total running time of the algorithm is upper bounded by $\sO ((\chid (\Delta + 1))^\tw)$.
\end{proof}

\section{Tree-depth Lower Bounds}\label{sec:td_lb}
In this section we present tight lower bounds on the complexity of solving \BDD{} and \DC,
when parameterized by the tree-depth of the input graph.
As in previous reductions, we start from a \kMC\ instance $G = (V,E)$,
where the vertices are partitioned into $k$ sets $V_i$, for $i \in [k]$.
Our main technical contribution is a recursive construction which allows us to
keep the tree-depth of the constructed graph linear with respect to $k$, thereby tightening previously
known lower bounds.
In the following we provide a high level sketch of the new ingredients of our construction.
For an illustration we refer to~\cref{fig:td_lb_construction}.

For every set $V_i$ we design a simple choice gadget $\hat{C}_i$ which encodes the choice of a vertex in $V_i$.
We also design a simple ``copy'' gadget, which, using a constant number of extra vertices per copy,
allows us to produce multiple choice gadgets which encode the same value.
In previous reductions (\cite{siamdm/BelmonteLM20,algorithmica/GanianKO21}),
we would now construct for each of the $k$ main choice gadgets at most $k-1$ copies,
and then for each $i_1, i_2 \in [k]$ we would select a distinct pair of copies of $\hat{C}_{i_1}, \hat{C}_{i_2}$
and add some machinery on these copies to verify that the choices for these groups encode the endpoints of an edge.
This approach naturally leads to a graph with tree-depth $\bO(k^2)$,
and as a matter of fact it establishes hardness even for more restrictive
parameters: as~\cite{algorithmica/GanianKO21} points out,
if we remove the $\bO(k^2)$ vertices that ensure that the copies of the choice gadgets
encode the same values,
we obtain a collection of graphs of constant tree-depth,
i.e.~the parameter is in fact modulator size to constant tree-depth,
rather than tree-depth itself.

The new ingredient in our approach is to observe that if we allow our reduction
to use the full power of tree-depth as a parameter,
we can avoid the quadratic blow-up in this construction.
Consider the slightly more general problem, where we have two intervals $I_1, I_2 \subseteq [k]$
and we want to construct some machinery that checks if for each $i_1 \in I_1$ and $i_2 \in I_2$
our choices for $V_{i_1}, V_{i_2}$ are valid, that is, encode the endpoints of an edge.
On a high level, this is the problem we want to solve for $I_1 = I_2 = [k]$,
and suppose we have some base gadget for the case $|I_1| = |I_2| = 1$.
We now observe that one way to solve the general case is recursive:
we cut the two intervals in half, say $I_1 = I_1^L \cup I_1^H$ and $I_2 = I_2^L \cup I_2^H$,
and then check the same condition for each pair in $(I_1^L, I_2^L), (I_1^L, I_2^H), (I_1^H, I_2^L), (I_1^H, I_2^H)$.
To this end, we make two copies of each choice gadget, thus constructing $\bO(k)$ new separator vertices,
but reducing to four instances of the same problem where all interval sizes have been cut in half.
As a result, to calculate the tree-depth of such a construction we get a recurrence of the form
$T(k) \leq \bO(k) + T(k/2)$, which in the end gives tree-depth $\bO(k)$.
Observe that this technique manages to produce better results than previous reductions
exactly because we are exploiting the full power of tree-depth:
we construct an instance that has $ck$ vertices whose removal,
rather than breaking the graph down into trivial components,
gives components which (recursively) have $ck/2$ vertices whose removal produces even simpler components,
and so on, through a recursion depth of height $\log k$.
In other words, unlike previous reductions, we crucially rely on the recursive definition of tree-depth.

\begin{figure}[ht]
\centering
\begin{tikzpicture}[scale=0.85, transform shape]
    \node[rectangle,draw,minimum width=1.5cm,minimum height = 1.1cm] (ci1) at (3,1) {$\hat{C}_{i_1}$};
    \node () at (3,2.1) {$\vdots$};
    \node[rectangle,draw,minimum width=1.5cm,minimum height = 1.1cm] (ci2) at (3,3) {$\hat{C}_{\floor*{\frac{i_1+i_2}{2}}}$};
    \node[rectangle,draw,minimum width=1.5cm,minimum height = 1.1cm] (ci3) at (3,7) {$\hat{C}_{\ceil*{\frac{i_1+i_2}{2}}}$};
    \node () at (3,8.1) {$\vdots$};
    \node[rectangle,draw,minimum width=1.5cm,minimum height = 1.1cm] (ci4) at (3,9) {$\hat{C}_{i_2}$};

    \node[rectangle,draw,minimum width=1.5cm,minimum height = 1.1cm] (ci'1) at (9,12) {$\hat{C}_{i'_1}$};
    \node () at (10.5,12) {$\cdots$};
    \node[rectangle,draw,minimum width=1.5cm,minimum height = 1.1cm] (ci'2) at (12,12) {$\hat{C}_{\floor*{\frac{i'_1+i'_2}{2}}}$};
    \node[rectangle,draw,minimum width=1.5cm,minimum height = 1.1cm] (ci'3) at (17,12) {$\hat{C}_{\ceil*{\frac{i'_1+i'_2}{2}}}$};
    \node () at (18.5,12) {$\cdots$};
    \node[rectangle,draw,minimum width=1.5cm,minimum height = 1.1cm] (ci'4) at (20,12) {$\hat{C}_{i'_2}$};

    \node[rectangle,draw,minimum width=1.5cm,minimum height = 1.1cm] (ci11) at (6,1) {$\hat{C}_{i_1}$};
    \node () at (6,2.1) {$\vdots$};
    \node[rectangle,draw,minimum width=1.5cm,minimum height = 1.1cm] (ci21) at (6,3) {$\hat{C}_{\floor*{\frac{i_1+i_2}{2}}}$};
    \node[rectangle,draw,minimum width=1.5cm,minimum height = 1.1cm] (ci31) at (6,7) {$\hat{C}_{\ceil*{\frac{i_1+i_2}{2}}}$};
    \node () at (6,8.1) {$\vdots$};
    \node[rectangle,draw,minimum width=1.5cm,minimum height = 1.1cm] (ci41) at (6,9) {$\hat{C}_{i_2}$};

    \node[rectangle,draw,minimum width=1.5cm,minimum height = 1.1cm] (ci'11) at (8,1) {$\hat{C}_{i'_1}$};
    \node () at (8,2.1) {$\vdots$};
    \node[rectangle,draw,minimum width=1.5cm,minimum height = 1.1cm] (ci'21) at (8,3) {$\hat{C}_{\floor*{\frac{i'_1+i'_2}{2}}}$};
    \node[rectangle,draw,minimum width=1.5cm,minimum height = 1.1cm] (ci'12) at (8,7) {$\hat{C}_{i'_1}$};
    \node () at (8,8.1) {$\vdots$};
    \node[rectangle,draw,minimum width=1.5cm,minimum height = 1.1cm] (ci'22) at (8,9) {$\hat{C}_{\floor*{\frac{i'_1+i'_2}{2}}}$};

    \node[rectangle,draw,minimum width=1.5cm,minimum height = 1.1cm] (ci12) at (14,1) {$\hat{C}_{i_1}$};
    \node () at (14,2.1) {$\vdots$};
    \node[rectangle,draw,minimum width=1.5cm,minimum height = 1.1cm] (ci22) at (14,3) {$\hat{C}_{\floor*{\frac{i_1+i_2}{2}}}$};
    \node[rectangle,draw,minimum width=1.5cm,minimum height = 1.1cm] (ci32) at (14,7) {$\hat{C}_{\ceil*{\frac{i_1+i_2}{2}}}$};
    \node () at (14,8.1) {$\vdots$};
    \node[rectangle,draw,minimum width=1.5cm,minimum height = 1.1cm] (ci42) at (14,9) {$\hat{C}_{i_2}$};

    \node[rectangle,draw,minimum width=1.5cm,minimum height = 1.1cm] (ci'31) at (16,1) {$\hat{C}_{\ceil*{\frac{i'_1+i'_2}{2}}}$};
    \node () at (16,2.1) {$\vdots$};
    \node[rectangle,draw,minimum width=1.5cm,minimum height = 1.1cm] (ci'41) at (16,3) {$\hat{C}_{i'_2}$};
    \node[rectangle,draw,minimum width=1.5cm,minimum height = 1.1cm] (ci'32) at (16,7) {$\hat{C}_{\ceil*{\frac{i'_1+i'_2}{2}}}$};
    \node () at (16,8.1) {$\vdots$};
    \node[rectangle,draw,minimum width=1.5cm,minimum height = 1.1cm] (ci'42) at (16,9) {$\hat{C}_{i'_2}$};
    
    \draw[dashed] (ci1)--(ci11);
    \draw[dashed] (ci2)--(ci21);
    \draw[dashed] (ci3)--(ci31);
    \draw[dashed] (ci4)--(ci41);
    \draw[dashed] (ci1) edge [bend right] (ci12);
    \draw[dashed] (ci2) edge [bend left] (ci22);
    \draw[dashed] (ci3) edge [bend right] (ci32);
    \draw[dashed] (ci4) edge [bend left] (ci42);

    \draw[dashed] (ci'1) edge [bend left] (ci'11);
    \draw[dashed] (ci'1) edge [bend left] (ci'12);
    \draw[dashed] (ci'2) edge [bend left] (ci'21);
    \draw[dashed] (ci'2) edge [bend left] (ci'22);
    \draw[dashed] (ci'3) edge [bend left] (ci'31);
    \draw[dashed] (ci'3) edge [bend left] (ci'32);
    \draw[dashed] (ci'4) edge [bend left] (ci'41);
    \draw[dashed] (ci'4) edge [bend left] (ci'42);


    \draw[thick] (5,0.25) rectangle (9,3.75);
    \draw[thick] (5,6.25) rectangle (9,9.75);

    \draw[thick] (13,0.25) rectangle (17,3.75);
    \draw[thick] (13,6.25) rectangle (17,9.75);
    
    \end{tikzpicture}
    \caption{
    Illustration where $I_1 = [i_1, i_2]$ and $I_2 = [i'_1, i'_2]$.
    Dashed lines denote copies,
    while the rectangles denote the reduced instances.
    }
    \label{fig:td_lb_construction}
\end{figure}

\subsection{Bounded Degree Vertex Deletion}\label{subsec:bdd_td_lb}

\begin{theorem}
    For any computable function $f$,
    if there exists an algorithm that solves \BDD\ in time $f(\td)n^{o(\td)}$,
    where \td{} denotes the tree-depth of the input graph,
    then the ETH is false.
\end{theorem}

\begin{proof}
    Let $(G,k)$ be an instance of \kMC,
    such that every vertex of $G$ has a self loop,
    i.e.~$\braces{v,v} \in E(G)$ for all $v \in V(G)$.
    Recall that we assume that $G$ is given to us partitioned into $k$ independent sets $V_1, \ldots, V_k$,
    where $V_i = \braces{v^i_1, \ldots, v^i_n}$.
    Assume without loss of generality that $k = 2^z$ for some $z \in \N$
    (one can do so by adding dummy independent sets connected to all the other vertices of the graph).
    Moreover, let $E^{i_1,i_2} \subseteq E(G)$ denote the edges of $G$ with one endpoint in $V_{i_1}$ and the other in $V_{i_2}$.
    Set $\Delta = n^3$.
    We will construct in polynomial time a graph $H$ of tree-depth $\td(H) = \bO(k)$ and size $|V(H)| = k^{\bO(1)} \cdot n^{\bO(1)}$,
    such that there exists $S \subseteq V(H)$, $|S| \leq k'$, and $H - S$ has maximum degree at most $\Delta$, for some $k'$,
    if and only if $G$ has a $k$-clique.
    
    \proofsubparagraph*{Choice Gadget.}
    For an independent set $V_i$, we construct the \emph{choice gadget} $\hat{C}_i$ as depicted in~\cref{fig:bdd_td_lb_choice_gadget}.
    We first construct independent sets $\hat{C}^p_i = \braces{v^{i,p}_1, \ldots, v^{i,p}_n}$, where $p \in \braces{h,l}$.
    Afterwards, we connect $v^{i,h}_j$ and $v^{i,l}_j$ with a vertex $q^i_j$, and add to the latter $\Delta - 1$ leaves.
    Intuitively, we will consider a one-to-one mapping between the vertex $v^i_j$ of $V_i$ belonging to a supposed $k$-clique of $G$ and
    the deletion of exactly $j$ vertices of $\hat{C}^l_i$ and $n-j$ vertices of $\hat{C}^h_i$.
    
    \proofsubparagraph*{Copy Gadget.}
    Given two instances $\mathcal{I}_1$, $\mathcal{I}_2$ of a choice gadget $\hat{C}_i$,
    when we say that we connect them with a \emph{copy gadget},
    we introduce two vertices $g_1$ and $g_2$, attach to each of those $\Delta - n$ leaves,
    and lastly add an edge between $g_1$ (respectively, $g_2$) and the vertices of $\hat{C}^l_{i}$ of instance $\mathcal{I}_1$ (respectively, $\mathcal{I}_2$),
    as well as the vertices of $\hat{C}^h_{i}$ of instance $\mathcal{I}_2$ (respectively, $\mathcal{I}_1$),
    as depicted in~\cref{fig:bdd_td_lb_copy_choice_gadget}.
    
    \begin{figure}[ht]
    \centering 
      \begin{subfigure}[b]{0.4\linewidth}
      \centering
        \begin{tikzpicture}[scale=0.75, transform shape]
        \node[vertex] (v1l) at (1,1) {};
        \node[] () at (1,0.6) {$v^{i,l}_1$};
        \node[] () at (2,1) {$\cdots$};
        \node[vertex] (vnl) at (3,1) {};
        \node[] () at (3,0.6) {$v^{i,l}_n$};
        
        \node[vertex] (v1h) at (1,5) {};
        \node[] () at (1,5.4) {$v^{i,h}_1$};
        \node[] () at (2,5) {$\cdots$};
        \node[vertex] (vnh) at (3,5) {};
        \node[] () at (3,5.4) {$v^{i,h}_n$};
        
        \node[black_vertex] (q1) at (1,3) {};
        \node[] () at (1.2,3.4) {$q^i_1$};
        \node[] () at (2,3) {$\cdots$};
        \node[black_vertex] (qn) at (3,3) {};
        \node[] () at (2.8,3.4) {$q^i_n$};
        
        
        \node[] () at (2,0.3) {$\hat{C}^l_{i}$};
        \node[] () at (2,5.7) {$\hat{C}^h_{i}$};
        
        \draw[] (v1l)--(q1)--(v1h);
        \draw[] (vnl)--(qn)--(vnh);

        \draw[dashed] (0.5,0) rectangle (3.5,2);
        \draw[dashed] (0.5,4) rectangle (3.5,6);
        
        \end{tikzpicture}
        \caption{Choice gadget $\hat{C}_i$.}
        \label{fig:bdd_td_lb_choice_gadget}
      \end{subfigure}
    \begin{subfigure}[b]{0.4\linewidth}
    \centering
        \begin{tikzpicture}[scale=0.7, transform shape]
        
        \node[vertex] (v1l) at (1,1) {};
        \node[] () at (2,1) {$\cdots$};
        \node[vertex] (vnl) at (3,1) {};
        
        \node[vertex] (v1h) at (1,5) {};
        \node[] () at (2,5) {$\cdots$};
        \node[vertex] (vnh) at (3,5) {};
        
        \node[] () at (2,0.3) {$\hat{C}^l_{i}$};
        \node[] () at (2,5.7) {$\hat{C}^h_{i}$};
    
        \node[] () at (2,7) {$\mathcal{I}_1$};

        \node[vertex] (v1l') at (7,1) {};
        \node[] () at (8,1) {$\cdots$};
        \node[vertex] (vnl') at (9,1) {};
        
        \node[vertex] (v1h') at (7,5) {};
        \node[] () at (8,5) {$\cdots$};
        \node[vertex] (vnh') at (9,5) {};
        
        \node[] () at (8,0.3) {$\hat{C}^l_{i}$};
        \node[] () at (8,5.7) {$\hat{C}^h_{i}$};
    
        \node[] () at (8,7) {$\mathcal{I}_2$};

        \node[gray_vertex] (g1) at (5,2.5) {};
        \node[] () at (5,2.2) {$g_1$};
        \node[gray_vertex] (g2) at (5,3.5) {};
        \node[] () at (5,3.8) {$g_2$};
        
        \draw[] (v1l) edge [bend left] (g1);
        \draw[] (vnl) edge [bend left] (g1);
        \draw[] (v1h) edge [bend right] (g2);
        \draw[] (vnh) edge [bend right] (g2);
    
        \draw[] (v1l') edge [bend right] (g2);
        \draw[] (vnl') edge [bend right] (g2);
        \draw[] (v1h') edge [bend left] (g1);
        \draw[] (vnh') edge [bend left] (g1);

        \draw[dashed] (0.5,0) rectangle (3.5,2);
        \draw[dashed] (0.5,4) rectangle (3.5,6);
        \draw[dashed] (6.5,0) rectangle (9.5,2);
        \draw[dashed] (6.5,4) rectangle (9.5,6);

        \end{tikzpicture}
        \caption{Making a copy of a choice gadget $\hat{C}_{i}$.}
        \label{fig:bdd_td_lb_copy_choice_gadget}
      \end{subfigure}
    \caption{Black vertices have $\Delta-1$ and gray vertices $\Delta - n$ leaves attached.}
    \end{figure}  
    
    \proofsubparagraph*{Edge Gadget.}
    Let $e = \braces{v^{i_1}_{j_1}, v^{i_2}_{j_2}} \in E^{i_1,i_2}$ be an edge of $G$.
    Construct the \emph{edge gadget} $\hat{E}_e$ as depicted in~\cref{fig:bdd_td_lb_edge_gadget},
    where every vertex $c^i_j$ has $\Delta$ leaves attached.
    
    \begin{figure}[ht]
    \centering
    \begin{tikzpicture}[scale=1.0, transform shape]
    
    \node[vertex] (r) at (5,5) {};
    \node[] () at (5,5.3) {$r$};
    
    \node[black_vertex] (v1) at (4,3) {};
    \node[] () at (4.4,3) {$c^{i_1}_1$};
    \node[vertex] (s1) at (3,3) {};
    \node[] () at (2.7,3) {$s^{i_1}_1$};
    \node[] () at (3,3.6) {$\vdots$};
    \node[black_vertex] (v2) at (4,4) {};
    \node[] () at (4,4.4) {$c^{i_1}_{j_1}$};
    \node[vertex] (s2) at (3,4) {};
    \node[] () at (2.7,4) {$s^{i_1}_{j_1}$};
    
    \node[black_vertex] (v3) at (4,6) {};
    \node[] () at (4,5.6) {$c^{i_1}_{j_1+1}$};
    \node[vertex] (s3) at (3,6) {};
    \node[] () at (2.7,5.8) {$s^{i_1}_{j_1+1}$};
    \node[] () at (3,6.6) {$\vdots$};
    \node[black_vertex] (v4) at (4,7) {};
    \node[] () at (4.4,7) {$c^{i_1}_n$};
    \node[vertex] (s4) at (3,7) {};
    \node[] () at (2.7,7) {$s^{i_1}_n$};
    
    \node[black_vertex] (v5) at (6,3) {};
    \node[] () at (5.6,3) {$c^{i_2}_1$};
    \node[vertex] (s5) at (7,3) {};
    \node[] () at (7.4,3) {$s^{i_2}_1$};
    \node[] () at (7,3.6) {$\vdots$};
    \node[black_vertex] (v6) at (6,4) {};
    \node[] () at (6,4.4) {$c^{i_2}_{j_2}$};
    \node[vertex] (s6) at (7,4) {};
    \node[] () at (7.4,4) {$s^{i_2}_{j_2}$};
    
    \node[black_vertex] (v7) at (6,6) {};
    \node[] () at (6.2,5.6) {$c^{i_2}_{j_2+1}$};
    \node[vertex] (s7) at (7,6) {};
    \node[] () at (7.6,6) {$s^{i_2}_{j_2+1}$};
    \node[] () at (7,6.6) {$\vdots$};
    \node[black_vertex] (v8) at (6,7) {};
    \node[] () at (5.6,7) {$c^{i_2}_n$};
    \node[vertex] (s8) at (7,7) {};
    \node[] () at (7.4,7) {$s^{i_2}_n$};
    
    \draw[] (r)--(v1)--(s1);
    \draw[] (r)--(v2)--(s2);
    \draw[] (r)--(v3)--(s3);
    \draw[] (r)--(v4)--(s4);
    \draw[] (r)--(v5)--(s5);
    \draw[] (r)--(v6)--(s6);
    \draw[] (r)--(v7)--(s7);
    \draw[] (r)--(v8)--(s8);
    
    \end{tikzpicture}
    \caption{Edge gadget $\hat{E}_e$ for $e = \braces{v^{i_1}_{j_1}, v^{i_2}_{j_2}}$.
    Black vertices have $\Delta$ leaves attached.}
    \label{fig:bdd_td_lb_edge_gadget}
    \end{figure}
    
    \proofsubparagraph*{Adjacency Gadget.}
    For $i_1 \leq i_2$ and $i'_1 \leq i'_2$, we define the \emph{adjacency gadget} $\hat{A}(i_1, i_2, i'_1, i'_2)$ as follows:
    \begin{itemize}
        \item Consider first the case when $i_1 = i_2$ and $i'_1 = i'_2$.
        Let the adjacency gadget contain instances of the edge gadgets $\hat{E}_e$, for $e \in E^{i_1,i'_1}$,
        the choice gadgets $\hat{C}_{i_1}$ and $\hat{C}_{i'_1}$,
        as well as vertices $\ell^l_{i_1,i'_1}$, $\ell^h_{i_1,i'_1}$, $r^l_{i_1,i'_1}$ and $r^h_{i_1,i'_1}$.
        Add edges between
        \begin{multicols}{2}
        \begin{itemize}
            \item $\ell^l_{i_1,i'_1}$ and $\hat{C}^l_{i_1}$,
            \item $\ell^h_{i_1,i'_1}$ and $\hat{C}^h_{i_1}$,
            \item $r^l_{i_1,i'_1}$ and $\hat{C}^l_{i'_1}$,
            \item $r^h_{i_1,i'_1}$ and $\hat{C}^h_{i'_1}$.
        \end{itemize}
        \end{multicols}
        If $e = \braces{v^{i_1}_{j_1}, v^{i'_1}_{j_2}} \in E^{i_1,i'_1}$, then add the following edges adjacent to $\hat{E}_e$:
        \begin{multicols}{2}
        \begin{itemize}
            \item $\ell^l_{i_1,i'_1}$ with $s^{i_1}_\kappa$, for $\kappa \in [j_1]$,
            \item $\ell^h_{i_1,i'_1}$ with $s^{i_1}_\kappa$, for $\kappa \in [j_1+1, n]$,
            \item $r^l_{i_1,i'_1}$ with $s^{i'_1}_\kappa$, for $\kappa \in [j_2]$,
            \item $r^h_{i_1,i'_1}$ with $s^{i'_1}_\kappa$, for $\kappa \in [j_2+1, n]$.
        \end{itemize}
        \end{multicols}
        Let $\tau(x)$, where $x \in \braces{\ell^l_{i_1,i'_1}, \ell^h_{i_1,i'_1}, r^l_{i_1,i'_1}, r^h_{i_1,i'_1}}$,
        denote the number of neighbors of $x$ belonging to some edge gadget.
        Attach $\Delta - \tau(x)$ leaves to vertex $x$.
        For an illustration see~\cref{fig:bdd_td_lb_adj_edge_case}.
    
        \item Now consider the case when $i_1 < i_2$ and $i'_1 < i'_2$.
        Then, let $\hat{A}(i_1,i_2,i'_1,i'_2)$ contain choice gadgets $\hat{C}_{i}$ and $\hat{C}_{i'}$, where $i \in [i_1,i_2]$ and $i' \in [i'_1, i'_2]$,
        which we will refer to as the \emph{original choice gadgets} of $\hat{A}(i_1,i_2,i'_1,i'_2)$,
        as well as the adjacency gadgets
        \begin{multicols}{2}
        \begin{itemize}
            \item $\hat{A}\parens*{i_1, \floor*{\frac{i_1+i_2}{2}}, i'_1, \floor*{\frac{i'_1+i'_2}{2}}}$,
            \item $\hat{A}\parens*{i_1, \floor*{\frac{i_1+i_2}{2}}, \ceil*{\frac{i'_1+i'_2}{2}}, i'_2}$,
            \item $\hat{A}\parens*{\ceil*{\frac{i_1+i_2}{2}}, i_2, i'_1, \floor*{\frac{i'_1+i'_2}{2}}}$,
            \item $\hat{A}\parens*{\ceil*{\frac{i_1+i_2}{2}}, i_2, \ceil*{\frac{i'_1+i'_2}{2}}, i'_2}$.
        \end{itemize}
        \end{multicols}
        Lastly, we connect with a copy gadget any choice gadgets $\hat{C}_i$ and $\hat{C}_{i'}$ appearing in said adjacency gadgets,
        with the corresponding original choice gadget $\hat{C}_i$ and $\hat{C}_{i'}$.
        Notice that then, every original choice gadget is taking part in two copy gadgets.
        For a high level illustration see~\cref{fig:td_lb_construction}.
    \end{itemize}
    
    \begin{figure}[ht]
    \centering
    \begin{tikzpicture}[scale=1.0, transform shape]
    
    
    \node[] () at (0.6,3.4) {$\hat{C}^l_{i_1}$};
    \node[vertex] (v1) at (1,3) {};
    \node[] () at (2,3) {$\cdots$};
    \node[vertex] (v2) at (3,3) {};
    
    \node[] () at (0.75,6.6) {$\hat{C}^h_{i_1}$};
    \node[vertex] (v3) at (1,7) {};
    \node[] () at (2,7) {$\cdots$};
    \node[vertex] (v4) at (3,7) {};
    
    \node[black_vertex] (ll) at (5,4) {};
    \node[] at (5,3.6) {$\ell^l_{i_1,i'_1}$};
    
    \node[black_vertex] (lh) at (5,6) {};
    \node[] at (5,6.5) {$\ell^h_{i_1,i'_1}$};
    
    \node[circle,draw] (e1) at (6.5,4) {$\hat{E}_{e_1}$};
    \node[] at (6.5,5) {$\vdots$};
    \node[circle,draw] (em) at (6.5,6) {$\hat{E}_{e_\lambda}$};
    
    \node[black_vertex] (rl) at (8,4) {};
    \node[] at (8,3.6) {$r^l_{i_1,i'_1}$};
    
    \node[black_vertex] (rh) at (8,6) {};
    \node[] at (8,6.5) {$r^h_{i_1,i'_1}$};
    
    \node[] () at (12.25,3.4) {$\hat{C}^l_{i'_1}$};
    \node[vertex] (v'1) at (10,3) {};
    \node[] () at (11,3) {$\cdots$};
    \node[vertex] (v'2) at (12,3) {};
    
    \node[] () at (12.25,6.6) {$\hat{C}^h_{i'_1}$};
    \node[vertex] (v'3) at (10,7) {};
    \node[] () at (11,7) {$\cdots$};
    \node[vertex] (v'4) at (12,7) {};
    
    \draw[dashed] (0,2) rectangle (4,4);
    \draw[dashed] (0,6) rectangle (4,8);
    \draw[dashed] (9,2) rectangle (13,4);
    \draw[dashed] (9,6) rectangle (13,8);
    
    
    \draw[] (v1) edge [bend left] (ll);
    \draw[] (v2)--(ll);
    \draw[] (v3) edge [bend right] (lh);
    \draw[] (v4)--(lh);
    
    \draw[] (lh)--(em.west);
    \draw[] (lh) edge [bend right] (e1.north west);
    
    \draw[] (ll)--(e1.west);
    \draw[] (ll) edge [bend left] (em.south west);
    
    \draw[] (rh)--(em.east);
    \draw[] (rh) edge [bend left] (e1.north east);
    
    \draw[] (rl)--(e1.east);
    \draw[] (rl) edge [bend right] (em.south east);
    
    \draw[] (v'1)--(rl);
    \draw[] (v'2) edge [bend right] (rl);
    \draw[] (v'3)--(rh);
    \draw[] (v'4) edge [bend left] (rh);

    \end{tikzpicture}
    \caption{Adjacency gadget $\hat{A}(i_1,i_1,i'_1,i'_1)$,
    where $E^{i_1,i'_1} = \setdef{e_i}{i \in [\lambda]}$.
    Black vertices have leaves attached.}
    \label{fig:bdd_td_lb_adj_edge_case}
    \end{figure}
    
    Let graph $H$ be the adjacency gadget $\hat{A}(1,k,1,k)$.
    Notice that it holds that $|V(H)| = (n \cdot k)^{\bO(1)}$.
    Let $\beta = 2k (2k - 1)$,
    and set $k' = 2(|E(G)| - kn) \cdot 2n + kn \cdot 2n + 2 \binom{k}{2} + k + n \cdot \beta$.
    
    \begin{lemma}\label{lem:bdd_td_lb_helper}
        $H$ has the following properties:
        \begin{itemize}
            \item The number of instances of choice gadgets present in $H$ is $\beta$,
            \item The number of instances of edge gadget $\hat{E}_e$ present in $H$, where $e = \braces{v^{i_1}_{j_1}, v^{i_2}_{j_2}} \in E(G)$,
            is one if $i_1 = i_2$, and two otherwise.
        \end{itemize}
    \end{lemma}
    
    \begin{proof}
        For the first statement,
        notice that the number of instances of choice gadgets is given by the recursive formula $T(k) = 2k + 4T(k/2)$, where $T(1) = 2$.
        In that case, it follows that
        \[
            T(k) = \sum_{i=0}^{\log k} \parens*{4^i \cdot 2 \cdot \frac{k}{2^i}} =
            2k \sum_{i=0}^{\log k} 2^i =
            2k (2k - 1) =
            \beta.
        \]
    
        For the second statement, first we will prove that for every adjacency gadget $\hat{A}(i_1,i_2,i'_1,i'_2)$ appearing in $H$,
        it holds that $i_2 - i_1 = i'_2 - i'_1 = 2^c - 1$, for some $c \in \N$.
        The statement holds for $\hat{A}(1,k,1,k)$, as well as when $i_2 - i_1 = i'_2 - i'_1 = 0$.
        Suppose that it holds for some $\hat{A}(i_1,i_2,i'_1,i'_2)$, i.e.~$i_2 - i_1 = i'_2 - i'_1 = 2^c - 1 > 0$, for some $c \in \N$.
        Then, it follows that $\floor*{\frac{i_1+i_2}{2}} - i_1 = \floor*{i_2 - 2^{c-1} + 0.5} - i_1 = i_2 - i_1 - 2^{c-1} = 2^{c-1} - 1$.
        Moreover, it follows that $i_2 - \ceil*{\frac{i_1+i_2}{2}} = i_2 - \ceil*{i_1 + 2^{c-1} - 0.5} = i_2 - (i_1 + 2^{c-1}) = 2^{c-1} - 1$.
        Therefore, the stated property holds.
    
        In that case, for some $\hat{A}(i_1,i_2,i'_1,i'_2)$, in every step of the recursion,
        intervals $[i_1,i_2]$ and $[i'_1,i'_2]$ are partitioned in the middle,
        and an adjacency gadget is considered for each of the four combinations.
        In that case, starting from $\hat{A}(1,k,1,k)$, there is a single way to produce
        every adjacency gadget $\hat{A}(i_1,i_1,i_2,i_2)$, where $i_1,i_2 \in [k]$.
        Consider an edge gadget $\hat{E}_e$, where $e = \braces{v^{i_1}_{j_1}, v^{i_2}_{j_2}} \in E(G)$.
        There are two cases:
        \begin{itemize}
            \item if $i_1 = i_2$, then this gadget appears only in the adjacency gadget $\hat{A}(i_1,i_1,i_1,i_1)$,
            \item alternatively, it appears in both adjacency gadgets $\hat{A}(i_1,i_1,i_2,i_2)$ and $\hat{A}(i_2,i_2,i_1,i_1)$.
        \end{itemize}
        This concludes the proof.
    \end{proof}

    \begin{lemma}\label{lem:bdd_td_lb_td}
        It holds that $\td(H) = \bO(k)$.
    \end{lemma}
    
    \begin{proof}
        Let $T(\kappa)$ denote the tree-depth of $\hat{A}(i_1,i_2,i'_1,i'_2)$ in the case when $i_2 - i_1 = i'_2 - i'_1 = \kappa$.
        
        First, notice that, for $i_1,i_2 \in [k]$, the tree-depth of $\hat{A}(i_1,i_1,i_2,i_2)$ is less than $8$:
        remove vertices $\ell^l_{i_1,i_2}$, $\ell^h_{i_1,i_2}$, $r^l_{i_1,i_2}$ and $r^h_{i_1,i_2}$,
        and all remaining connected components are trees of height at most $3$.
        Consequently, $T(1) \leq 8$.
    
        Now, consider the adjacency gadget $\hat{A}(i_1,i_2,i'_1,i'_2)$, where $i_2 - i_1 = i'_2 - i'_1 = \kappa$.
        This is comprised of adjacency gadgets
        \begin{multicols}{2}
        \begin{itemize}
            \item $\hat{A}\parens*{i_1, \floor*{\frac{i_1+i_2}{2}}, i'_1, \floor*{\frac{i'_1+i'_2}{2}}}$,
            \item $\hat{A}\parens*{i_1, \floor*{\frac{i_1+i_2}{2}}, \ceil*{\frac{i'_1+i'_2}{2}}, i'_2}$,
            \item $\hat{A}\parens*{\ceil*{\frac{i_1+i_2}{2}}, i_2, i'_1, \floor*{\frac{i'_1+i'_2}{2}}}$,
            \item $\hat{A}\parens*{\ceil*{\frac{i_1+i_2}{2}}, i_2, \ceil*{\frac{i'_1+i'_2}{2}}, i'_2}$.
        \end{itemize}
        \end{multicols}
        as well as of exactly $2\kappa$ original choice gadgets,
        each of which is connected with two copy gadgets to other instances of choice gadgets present in the adjacency gadgets.
        By removing all vertices $g_1$ and $g_2$ of the copy gadgets (see~\cref{fig:bdd_td_lb_copy_choice_gadget}),
        all the original choice gadgets as well as the adjacency gadgets are disconnected.
        Therefore, it holds that $T(\kappa) \leq 8 \kappa + T(\kappa / 2)$,
        thus, it follows that $T(k) \leq 8 \sum_{i=0}^{\log k} \frac{k}{2^i} = \bO(k)$.
    \end{proof}
    
    \begin{lemma}\label{lem:bdd_td_lb_cor1}
        If $G$ contains a $k$-clique,
        then there exists $S \subseteq V(H)$, with $|S| \leq k'$,
        such that $H - S$ has maximum degree at most $\Delta$.
    \end{lemma}
    
    \begin{proof}
        Let $\mathcal{V} \subseteq V(G)$ be a $k$-clique of $G$, consisting of vertices $v^i_{s(i)} \in V_i$, for $i \in [k]$.
        We will construct a deletion set $S \subseteq V(H)$ as follows:
        \begin{itemize}
            \item Let $S$ contain vertices $v^{i,h}_{j_1}$ and $v^{i,l}_{j_2}$, for $j_1 \in [s(i)]$ and $j_2 \in [s(i)+1, n]$,
            from every instance of the choice gadget $\hat{C}_{i}$.
            \item Let $\hat{E}_e$ be the edge gadget of edge $e$,
            where $e = \braces{v^{i_1}_{j_1}, v^{i_2}_{j_2}}$.
            If $v^{i_1}_{j_1}, v^{i_2}_{j_2} \in \mathcal{V}$,
            then let $S$ include from $\hat{E}_e$ the vertices $r$, $s^{i_1}_{j}$ and $s^{i_2}_{j}$,
            where $j \in [n]$.
            Alternatively, let $S$ include from $\hat{E}_e$ all the vertices $c^{i_1}_{j}$ and $c^{i_2}_{j}$, where $j \in [n]$.
        \end{itemize}
        The edges for which the first case holds are exactly $\binom{k}{2} + k$.
        Due to~\cref{lem:bdd_td_lb_helper}, it holds that $\hat{E}_e$ appears once in $H$ if $e$ is a self loop and twice if not,
        while exactly $\beta$ instances of choice gadgets are present in $H$.
        In the end, $S$ contains $2(|E(G)| - kn) \cdot 2n + kn \cdot 2n + 2 \binom{k}{2} + k$ vertices due to the edge gadgets,
        plus $\beta \cdot n$ vertices due to the choice gadgets,
        thus $|S| = k'$ follows.
    
        It remains to prove that $H - S$ has maximum degree at most $\Delta$.
        One can easily verify that every vertex $g_1$ and $g_2$ of a copy gadget has degree exactly $\Delta$.
        Moreover, every vertex $q^i_j$ in the instances of choice gadgets has exactly one neighbor in $S$.
        For every vertex $c^i_j$ in an edge gadget, either itself or its neighboring vertices $r$ and $s^i_j$ belong to $S$.
        Lastly, for $i_1, i_2 \in [k]$, let $v^{i_1}_{s(i_1)}, v^{i_2}_{s(i_2)} \in \mathcal{V}$, where $e$ denotes the edge connecting them.
        Then, for $\ell^l_{i_1,i_2}$, it holds that it has exactly $\Delta - \tau(\ell^l_{i_1,i_2}) + v^{i_1}_{s(i_1)} + \tau(\ell^l_{i_1,i_2}) - v^{i_1}_{s(i_1)} = \Delta$ neighbors in $H - S$,
        due to its leaves, its neighbors in $\hat{C}^l_{i_1}$ and its neighbors in all the edge gadgets.
        In an analogous way, one can show that $\ell^h_{i_1,i_2}$, $r^l_{i_1,i_2}$ and $r^h_{i_1,i_2}$ all have degree $\Delta$ in $H - S$.
    \end{proof}
    
    \begin{lemma}\label{lem:bdd_td_lb_cor2}
        If there exists $S \subseteq V(H)$, with $|S| \leq k'$,
        such that $H - S$ has maximum degree at most $\Delta$,
        then $G$ contains a $k$-clique.
    \end{lemma}
    
    \begin{proof}
        Let $S \subseteq V(H)$, with $|S| \leq k'$,
        such that $H - S$ has maximum degree at most $\Delta$.
        Due to~\cref{lem:bdd_td_lb_helper}, it holds that $H$ contains exactly $\beta$ instances of choice gadgets,
        while the edge gadget $\hat{E}_e$, where $e = \braces{v^{i_1}_{j_1}, v^{i_2}_{j_2}}$, appears once if $i_1 = i_2$ and twice otherwise.
        Notice that $S$ must contain at least $n$ vertices per choice gadget,
        since every vertex $q^i_j$ has degree $\Delta+1$, while no two share any neighbors,
        for a total of at least $n \cdot \beta$ vertices.
        Moreover, $S$ contains at least $2n$ vertices per edge gadget, since vertices $c^i_j$ are of degree $\Delta+2$ and share only a single neighbor.
        Notice that there are $2(|E(G)| - kn) + kn$ instances of edge gadgets, for a total of $2(|E(G)| - kn) \cdot 2n + kn \cdot 2n$ vertices.
    
        \begin{claim}
            There exists a single edge gadget in $\hat{A}(i_1,i_1,i_2,i_2)$ from which $S$ contains exactly $2n+1$ vertices,
            for every $i_1,i_2 \in [k]$.
        \end{claim}
    
        \begin{claimproof}
            Since $S$ contains at least $n$ and $2n$ vertices per choice and edge gadget respectively,
            it follows that we are left with an additional budget of at most $2 \binom{k}{2} + k$.
            The number of adjacency gadgets $\hat{A}(i_1,i_1,i_2,i_2)$ is $2 \binom{k}{2} + k$, due to~\cref{lem:bdd_td_lb_helper}.
            Consequently, we will prove for every such adjacency gadget,
            $S$ contains $2n + 2n \cdot |E^{i_1,i_2}| + 1$ (remember that each such gadget has $2$ choice gadgets as well as $|E^{i_1,i_2}|$ edge gadgets).
            In fact, we will prove that for each such adjacency gadget,
            there exists a single edge gadget from which $S$ contains $2n+1$ vertices.
        
            First, we will prove that $S$ contains exactly $2n + 2n \cdot |E^{i_1,i_2}| + 1$ vertices per adjacency gadget $\hat{A}(i_1,i_1,i_2,i_2)$.
            If it contains less, then it necessarily holds that it contains $n$ vertices from each of the two choice gadgets,
            as well as $2n|E^{i_1,i_2}|$ vertices $c^i_j$ from all the edge gadgets of the adjacency gadget.
            In that case, both vertices $\ell^l_{i_1,i_2}$ and $\ell^h_{i_1,i_2}$ must have no neighbors from $\hat{C}^l_{i_1}$ and $\hat{C}^h_{i_1}$ respectively,
            which cannot be the case since only $n$ vertices of $\hat{C}_{i_1}$ belong to $S$.
            On the other hand, if $S$ contains more than $2n + 2n \cdot |E^{i_1,i_2}| + 1$ vertices from some adjacency gadget $\hat{A}(i_1,i_1,i_2,i_2)$,
            then it follows that there exists another adjacency gadget $\hat{A}(i'_1,i'_1,i'_2,i'_2)$ from which it contains less than that many vertices, contradiction.
        
            Now, suppose that for some adjacency gadget $\hat{A}(i_1,i_1,i_2,i_2)$, there is no edge gadget from which $S$ contains $2n+1$ vertices.
            In that case, from every edge gadget, all the vertices $c^i_j$ belong to $S$, and consequently,
            for every vertex $\ell^l_{i_1,i_2}$, $\ell^h_{i_1,i_2}$,
            $r^l_{i_1,i_2}$ and $r^h_{i_1,i_2}$, it holds that either all of its neighbors in the choice gadgets belong to $S$ or that it itself does.
            Since at most one of those vertices may belong to $S$, this leads to a contradiction.
            As a consequence of the above, it follows that $S$ contains exactly $n$ vertices per choice gadget of $\hat{A}(i_1,i_1,i_2,i_2)$,
            as well as all vertices $c^i_j$ from all but one edge gadgets present.
            For the extra edge gadget, we can assume that $S$ contains vertex $r$ as well as all the vertices $s^i_j$ (if that is not the case,
            there exists some deletion set $S'$ of same cardinality for which this holds, since only vertices $s^i_j$ have an edge with vertices outside of the edge gadget).
        \end{claimproof}
        
        Consequently, $S$ contains exactly $n$ vertices per choice gadget,
        as well as $2n \cdot |E^{i_1,i_2}| + 1$ vertices from the edge gadgets of the adjacency gadget $\hat{A}(i_1,i_1,i_2,i_2)$.
        Notice that no vertex $g_1$ and $g_2$ of a copy gadget belongs to $S$, and both have at most $n$ neighbors in $H-S$
        from the corresponding parts of the choice gadgets (see~\cref{fig:bdd_td_lb_copy_choice_gadget}).
        In that case, it follows that only vertices $v^{i,l}_j$ and $v^{i,h}_j$ belong to $S$ from the choice gadgets.
        Additionally, it follows that, in~\cref{fig:bdd_td_lb_copy_choice_gadget},
        the number of vertices of $\hat{C}^j_{i}$ belonging to $S$ of instance $\mathcal{I}_1$ is the same as the one of instance $\mathcal{I}_2$,
        for $j \in \braces{l,h}$.
        Therefore, we conclude that the number of vertices belonging to $S$ from $\hat{C}^l_{i}$ (respectively, $\hat{C}^h_{i}$)
        is the same in all the instances of the choice gadget $\hat{C}_i$.
    
        Let $\mathcal{V} \subseteq V(G)$ be a set of cardinality $k$, containing vertex $v^i_{s(i)} \in V_i$ if,
        for choice gadget $\hat{C}_{i}$,
        it holds that $|S \cap \hat{C}^h_{i}| = s(i)$ and $|S \cap \hat{C}^l_{i}| = n - s(i)$.
        Notice that $\mathcal{V} \cap V_i \neq \emptyset$, for all $i \in [k]$.
        We will prove that $\mathcal{V}$ is a clique.
        
        Let $v^{i_1}_{s(i_1)}, v^{i_2}_{s(i_2)}$ belong to $\mathcal{V}$.
        Consider the adjacency gadget $\hat{A}(i_1,i_1,i_2,i_2)$.
        We will prove that it contains an edge gadget $\hat{E}_e$, where $e = \braces{v^{i_1}_{s(i_1)}, v^{i_2}_{s(i_2)}}$.
        Consider the vertices $\ell^h_{i_1,i_2}$, $\ell^l_{i_1,i_2}$, $r^h_{i_1,i_2}$ and $r^l_{i_1,i_2}$.
        It holds that $x$ has $\Delta - \tau(x)$ leaves, as well as $\tau(x)$ neighbors in the edge gadgets,
        where $x \in \braces{ \ell^h_{i_1,i_2}, \ell^l_{i_1,i_2}, r^h_{i_1,i_2}, r^l_{i_1,i_2}  }$.
        Moreover,
        \begin{itemize}
            \item $\ell^h_{i_1,i_2}$ has $n - s(i_1)$ neighbors due to $\hat{C}^h_{i_1}$ in $H - S$,
            \item $\ell^l_{i_1,i_2}$ has $s(i_1)$ neighbors due to $\hat{C}^l_{i_1}$ in $H - S$,
            \item $r^h_{i_1,i_2}$ has $n - s(i_2)$ neighbors due to $\hat{C}^h_{i_2}$ in $H - S$,
            \item $r^l_{i_1,i_2}$ has $s(i_1)$ neighbors due to $\hat{C}^l_{i_2}$ in $H - S$.
        \end{itemize}
        Notice that from all but one edge gadgets, $S$ contains all the $c^i_j$ vertices.
        Since all $x$ have degree at most $\Delta$ in $H-S$,
        it follows that there exists an edge gadget $\hat{E}_{e'}$, where
        \begin{itemize}
            \item $\ell^h_{i_1,i_2}$ has at least $n - s(i_1)$ neighbors in,
            \item $\ell^l_{i_1,i_2}$ has at least $s(i_1)$ neighbors in,
            \item $r^h_{i_1,i_2}$ has at least $n - s(i_2)$ neighbors in,
            \item $r^l_{i_1,i_2}$ has at least $s(i_1)$ neighbors in,
        \end{itemize}
        and from which all the vertices $s^i_j$ belong to $S$.
        The only case this may happen is when $e' = e$, thus there exists such an edge in $G$.
    
        Since this holds for any two vertices belonging to $\mathcal{V}$, it follows that $G$ has a $k$-clique.
    \end{proof}

    Therefore, in polynomial time, we can construct a graph $H$,
    of tree-depth $\td = \bO(k)$ due to~\cref{lem:bdd_td_lb_td},
    such that, due to~\cref{lem:bdd_td_lb_cor1,lem:bdd_td_lb_cor2},
    deciding whether there exists $S \subseteq V(H)$ of size $|S| \leq k'$ and $H - S$ has maximum degree at most $\Delta = n^3$
    is equivalent to deciding whether $G$ has a $k$-clique.
    In that case, assuming there exists a $f(\td) |V(H)|^{o(\td)}$ algorithm for \BDD,
    where $f$ is any computable function,
    one could decide \kMC{} in time $f(\td) |V(H)|^{o(\td)} = g(k) \cdot n^{o(k)}$,
    for some computable function $g$, which contradicts the ETH.
\end{proof}

\subsection{Defective Coloring}\label{sec:dc_td_lb}

The proof will closely follow the hardness proof presented in~\cite{siamdm/BelmonteLM20}.
Since we will repeatedly use the equality and palette gadgets (see~\cref{sec:dc_coloring_gadgets}),
we will use the following convention:
whenever $v_1, v_2$ are two vertices we have already introduced to the constructed graph $H$,
when we say that we add an equality gadget $Q(v_1, v_2)$,
this means that we add to $H$ a copy of $Q(u_1, u_2, \chid, \Delta)$
and then identify $u_1, u_2$ with $v_1, v_2$ respectively (similarly for palette gadgets).

\begin{theorem}
    For any fixed $\chid \geq 2$,
    if there exists an algorithm that solves \DC\ in time $f(\td)n^{o(\td)}$,
    where $f$ is any computable function and \td{} denotes the tree-depth of the input graph,
    then the ETH is false.
\end{theorem}

\begin{proof}
    Let $(G,k)$ be an instance of \kMC, such that every vertex of $G$ has a self loop,
    i.e.~$\braces{v,v} \in E(G)$ for all $v \in V(G)$.
    Recall that we assume that $G$ is given to us partitioned into $k$ independent sets $V_1, \ldots, V_k$,
    where $V_i = \braces{v^i_1, \ldots, v^i_n}$, and each independent set has size exactly $n$.
    Assume without loss of generality that $k = 2^z$ for some $z \in \N$
    (one can do so by adding dummy independent sets connected to all the other vertices of the graph).
    Moreover, let $E^{i_1,i_2} \subseteq E(G)$ denote the edges of $G$ with one endpoint in $V_{i_1}$ and the other in $V_{i_2}$.
    Set $\Delta = 2(|E(G)| - kn) + kn - (2\binom{k}{2} + k)$.
    We will construct in polynomial time a graph $H$ of tree-depth $\td(H) = \bO(k)$ such that
    $H$ admits a $(\chid, \Delta)$-coloring
    if and only if $G$ has a $k$-clique.
    
    \proofsubparagraph*{Palette Vertices.}
    Construct two vertices $p^A$ and $p^B$, which we call main palette vertices, and add an edge connecting them.
    Next, construct vertices $p^A_i$ and $p^B_i$, for $i \in [\Delta]$,
    add equality gadgets $Q(p^j, p^j_i)$ as well as edges between $p^j$ and $p^j_i$,
    where $j \in \braces{A, B}$ and $i \in [\Delta]$.
    
    \proofsubparagraph*{Choice Gadget.}
    For an independent set $V_i$, we construct the \emph{choice gadget} $\hat{C}_{i}$ as depicted in~\cref{fig:dc_td_lb_choice_gadget}.
    We first construct independent sets $\hat{C}^p_{i} = \setdef{v^{i,p}_j}{j \in [n]}$, where $p \in \braces{h,l}$.
    We will refer to these vertices as \emph{choice vertices}.
    Next, we introduce vertices $f^A_i$ and $f^B_i$, connected with all choice vertices,
    and add equality gadgets $Q(p^A, f^A_i)$ and $Q(p^B, f^B_i)$.
    Finally, we attach to $f^A_i$ (respectively, $f^B_i$) $\Delta - n$ leaves $l^{f^A_i}_j$ (respectively, $l^{f^B_i}_j$),
    for $j \in [\Delta-n]$,
    and add equality gadgets $Q(p^A,l^{f^A_i}_j)$ (respectively, $Q(p^B,l^{f^B_i}_j)$).
    If $\chid \geq 3$, then we add a palette gadget $P(p^A, p^B, v^{i,q}_j)$ for all choice vertices $v^{i,q}_j$.
    Intuitively, we consider a one-to-one mapping between the vertex $v^i_j$ of $V_i$ belonging to a supposed $k$-clique of $G$ and the coloring of exactly $j$ vertices
    of $\hat{C}^l_{i}$ and $n-j$ of $\hat{C}^h_{i}$ with the same color as the one used to color $p^A$.
    
    \proofsubparagraph*{Copy Gadget.}
    Given two instances $\mathcal{I}_1$, $\mathcal{I}_2$ of a choice gadget $\hat{C}_i$,
    when we say that we connect them with a \emph{copy gadget},
    we introduce two vertices $g_1$ and $g_2$ and add equality gadgets $Q(p^A,g_1)$ and $Q(p^A,g_2)$.
    Moreover, we add an edge between $g_1$ (respectively, $g_2$) and the vertices of $\hat{C}^l_{i}$  of instance $\mathcal{I}_1$ (respectively, $\mathcal{I}_2$),
    as well as the vertices of $\hat{C}^h_{i}$ of instance $\mathcal{I}_2$ (respectively, $\mathcal{I}_1$),
    as depicted in~\cref{fig:dc_td_lb_copy_choice_gadget}.
    Lastly, we attach to each of $g_1$ and $g_2$ $\Delta - n$ leaves $l^{g_1}_j$ and $l^{g_b}_j$ respectively,
    where $j \in [\Delta - n]$, and add equality gadgets $Q(p^A,l^{g_1}_j)$ and $Q(p^A,l^{g_2}_j)$.
    
    \begin{figure}[ht]
    \centering 
      \begin{subfigure}[b]{0.4\linewidth}
      \centering
        \begin{tikzpicture}[scale=0.75, transform shape]
        \node[vertex] (v1l) at (1,1) {};
        \node[] () at (1,0.6) {$v^{i,l}_1$};
        \node[] () at (2,1) {$\cdots$};
        \node[vertex] (vnl) at (3,1) {};
        \node[] () at (3,0.6) {$v^{i,l}_n$};
        
        \node[vertex] (v1h) at (1,5) {};
        \node[] () at (1,5.4) {$v^{i,h}_1$};
        \node[] () at (2,5) {$\cdots$};
        \node[vertex] (vnh) at (3,5) {};
        \node[] () at (3,5.4) {$v^{i,h}_n$};
    
        \node[gray_vertex] (g1) at (1,3) {};
        \node[] () at (0.7,3.4) {$f^A_i$};
        \node[gray_vertex] (g2) at (3,3) {};
        \node[] () at (3.3,3.4) {$f^B_i$};
    
    
        \node[] () at (2,0.3) {$\hat{C}^l_{i}$};
        \node[] () at (2,5.7) {$\hat{C}^h_{i}$};
        
        \draw[] (v1l)--(g1)--(v1h);
        \draw[] (vnl)--(g2)--(vnh);
        \draw[] (vnh) edge [bend right] (g1);
        \draw[] (vnl) edge [bend left] (g1);
        \draw[] (v1h) edge [bend left] (g2);
        \draw[] (v1l) edge [bend right] (g2);
        
        \draw[dashed] (0.5,0) rectangle (3.5,2);
        \draw[dashed] (0.5,4) rectangle (3.5,6);
        
        \end{tikzpicture}
        \caption{Choice gadget $\hat{C}_{i}$.}
        \label{fig:dc_td_lb_choice_gadget}
      \end{subfigure}
    \begin{subfigure}[b]{0.4\linewidth}
    \centering
        \begin{tikzpicture}[scale=0.7, transform shape]
        
        \node[vertex] (v1l) at (1,1) {};
        \node[] () at (2,1) {$\cdots$};
        \node[vertex] (vnl) at (3,1) {};
        
        \node[vertex] (v1h) at (1,5) {};
        \node[] () at (2,5) {$\cdots$};
        \node[vertex] (vnh) at (3,5) {};
        
        \node[] () at (2,0.3) {$\hat{C}^l_{i}$};
        \node[] () at (2,5.7) {$\hat{C}^h_{i}$};
    
        \node[] () at (2,7) {$\mathcal{I}_1$};
    
        \node[vertex] (v1l') at (7,1) {};
        \node[] () at (8,1) {$\cdots$};
        \node[vertex] (vnl') at (9,1) {};
        
        \node[vertex] (v1h') at (7,5) {};
        \node[] () at (8,5) {$\cdots$};
        \node[vertex] (vnh') at (9,5) {};
        
        \node[] () at (8,0.3) {$\hat{C}^l_{i}$};
        \node[] () at (8,5.7) {$\hat{C}^h_{i}$};
    
        \node[] () at (8,7) {$\mathcal{I}_2$};
    
        \node[gray_vertex] (g1) at (5,2.5) {};
        \node[] () at (5,2.2) {$g_1$};
        \node[gray_vertex] (g2) at (5,3.5) {};
        \node[] () at (5,3.8) {$g_2$};
        
        \draw[] (v1l) edge [bend left] (g1);
        \draw[] (vnl) edge [bend left] (g1);
        \draw[] (v1h) edge [bend right] (g2);
        \draw[] (vnh) edge [bend right] (g2);
    
        \draw[] (v1l') edge [bend right] (g2);
        \draw[] (vnl') edge [bend right] (g2);
        \draw[] (v1h') edge [bend left] (g1);
        \draw[] (vnh') edge [bend left] (g1);

        \draw[dashed] (0.5,0) rectangle (3.5,2);
        \draw[dashed] (0.5,4) rectangle (3.5,6);
        \draw[dashed] (6.5,0) rectangle (9.5,2);
        \draw[dashed] (6.5,4) rectangle (9.5,6);

        \end{tikzpicture}
        \caption{Making a copy of a choice gadget $\hat{C}_i$.}
        \label{fig:dc_td_lb_copy_choice_gadget}
      \end{subfigure}
    \caption{Gray vertices have $\Delta-n$ leaves attached.}
    \end{figure}  
    
    \proofsubparagraph*{Edge Gadget.}
    Let $e = \braces{v^{i_1}_{j_1}, v^{i_2}_{j_2}} \in E^{i_1,i_2}$ be an edge of $G$.
    Construct the \emph{edge gadget} $\hat{E}_e$ as depicted in~\cref{fig:bdd_td_lb_edge_gadget},
    where vertex $c_e$ has $\Delta$ leaves $l^e_j$, $j \in [\Delta]$ attached,
    and add equality gadgets $Q(p^B,l^e_j)$ for each such leaf.
    Moreover, if $\chid \geq 3$, then add a palette gadget $P(p^A, p^B, v)$, for every vertex $v \in \braces{c_e, s^{i_1}_j,s^{i_2}_j}$.
    
    \begin{figure}[ht]
    \centering
    \begin{tikzpicture}[scale=1, transform shape]
    
    \node[black_vertex] (r) at (5,5) {};
    \node[] () at (5,5.3) {$c_e$};
    
    \node[vertex] (s1) at (3,3) {};
    \node[] () at (2.7,3) {$s^{i_1}_1$};
    \node[] () at (3,3.6) {$\vdots$};
    \node[vertex] (s2) at (3,4) {};
    \node[] () at (2.7,4) {$s^{i_1}_{j_1}$};
    
    \node[vertex] (s3) at (3,6) {};
    \node[] () at (2.7,5.8) {$s^{i_1}_{j_1+1}$};
    \node[] () at (3,6.6) {$\vdots$};
    \node[vertex] (s4) at (3,7) {};
    \node[] () at (2.7,7) {$s^{i_1}_n$};
    
    \node[vertex] (s5) at (7,3) {};
    \node[] () at (7.4,3) {$s^{i_2}_1$};
    \node[] () at (7,3.6) {$\vdots$};
    \node[vertex] (s6) at (7,4) {};
    \node[] () at (7.4,4) {$s^{i_2}_{j_2}$};
    
    \node[vertex] (s7) at (7,6) {};
    \node[] () at (7.6,6) {$s^{i_2}_{j_2+1}$};
    \node[] () at (7,6.6) {$\vdots$};
    \node[vertex] (s8) at (7,7) {};
    \node[] () at (7.4,7) {$s^{i_2}_n$};
    
    \draw[] (r)--(s1);
    \draw[] (r)--(s2);
    \draw[] (r)--(s3);
    \draw[] (r)--(s4);
    \draw[] (r)--(s5);
    \draw[] (r)--(s6);
    \draw[] (r)--(s7);
    \draw[] (r)--(s8);
    
    \end{tikzpicture}
    \caption{Edge gadget $\hat{E}_e$ for $e = \braces{v^{i_1}_{j_1}, v^{i_2}_{j_2}}$.
    The black vertex has $\Delta$ leaves attached.}
    \label{fig:dc_td_lb_edge_gadget}
    \end{figure}
    
    \proofsubparagraph*{Adjacency Gadget.}
    For $i_1 \leq i_2$ and $j_1 \leq j_2$, we define the \emph{adjacency gadget} $\hat{A}(i_1, i_2, j_1, j_2)$ as follows:
    \begin{itemize}
        \item Consider first the case when $i_1 = i_2$ and $i'_1 = i'_2$.
        Let the adjacency gadget contain instances of the edge gadgets $\hat{E}_e$, for $e \in E^{i_1,i'_1}$,
        the choice gadgets $\hat{C}_{i_1}$ and $\hat{C}_{i'_1}$,
        as well as vertices $\ell^l_{i_1,i'_1}$, $\ell^h_{i_1,i'_1}$, $r^l_{i_1,i'_1}$ and $r^h_{i_1,i'_1}$.
        Add edges between
        \begin{multicols}{2}
        \begin{itemize}
            \item $\ell^l_{i_1,i'_1}$ and $\hat{C}^l_{i_1}$,
            \item $\ell^h_{i_1,i'_1}$ and $\hat{C}^h_{i_1}$,
            \item $r^l_{i_1,i'_1}$ and $\hat{C}^l_{i'_1}$,
            \item $r^h_{i_1,i'_1}$ and $\hat{C}^h_{i'_1}$.
        \end{itemize}
        \end{multicols}
        If $e = \braces{v^{i_1}_{j_1}, v^{i'_1}_{j_2}} \in E^{i_1,i'_1}$, then add the following edges adjacent to $\hat{E}_e$:
        \begin{multicols}{2}
        \begin{itemize}
            \item $\ell^l_{i_1,i'_1}$ with $s^{i_1}_\kappa$, for $\kappa \in [j_1]$,
            \item $\ell^h_{i_1,i'_1}$ with $s^{i_1}_\kappa$, for $\kappa \in [j_1+1, n]$,
            \item $r^l_{i_1,i'_1}$ with $s^{i'_1}_\kappa$, for $\kappa \in [j_2]$,
            \item $r^h_{i_1,i'_1}$ with $s^{i'_1}_\kappa$, for $\kappa \in [j_2+1, n]$.
        \end{itemize}
        \end{multicols}
        For every vertex $x \in \braces{\ell^l_{i_1,i'_1}, \ell^h_{i_1,i'_1}, r^l_{i_1,i'_1}, r^h_{i_1,i'_1}}$,
        add an equality gadget $Q(p^A,x)$ and attach $\Delta - n$ leaves.
        Lastly, for each leaf $l$, add an equality gadget $Q(p^A,l)$.
        For an illustration see~\cref{fig:dc_td_lb_adj_edge_case}.
    
        \item Now consider the case when $i_1 < i_2$ and $i'_1 < i'_2$.
        Then, let $\hat{A}(i_1,i_2,i'_1,i'_2)$ contain choice gadgets $\hat{C}_{i}$ and $\hat{C}_{i'}$, where $i \in [i_1,i_2]$ and $i' \in [i'_1, i'_2]$,
        which we will refer to as the \emph{original choice gadgets} of $\hat{A}(i_1,i_2,i'_1,i'_2)$,
        as well as the adjacency gadgets
        \begin{multicols}{2}
        \begin{itemize}
            \item $\hat{A}\parens*{i_1, \floor*{\frac{i_1+i_2}{2}}, i'_1, \floor*{\frac{i'_1+i'_2}{2}}}$,
            \item $\hat{A}\parens*{i_1, \floor*{\frac{i_1+i_2}{2}}, \ceil*{\frac{i'_1+i'_2}{2}}, i'_2}$,
            \item $\hat{A}\parens*{\ceil*{\frac{i_1+i_2}{2}}, i_2, i'_1, \floor*{\frac{i'_1+i'_2}{2}}}$,
            \item $\hat{A}\parens*{\ceil*{\frac{i_1+i_2}{2}}, i_2, \ceil*{\frac{i'_1+i'_2}{2}}, i'_2}$.
        \end{itemize}
        \end{multicols}
        Lastly, we connect with a copy gadget any choice gadgets $\hat{C}_i$ and $\hat{C}_{i'}$ appearing in said adjacency gadgets,
        with the corresponding original choice gadget $\hat{C}_i$ and $\hat{C}_{i'}$.
        Notice that then, every original choice gadget is taking part in two copy gadgets.
        For a high level illustration see~\cref{fig:td_lb_construction}.
    \end{itemize}
    
    \begin{figure}[ht]
    \centering
    \begin{tikzpicture}[scale=1, transform shape]
    
    
    \node[] () at (0.6,3.3) {$\hat{C}^l_{i_1}$};
    \node[vertex] (v1) at (1,3) {};
    \node[] () at (2,3) {$\cdots$};
    \node[vertex] (v2) at (3,3) {};
    
    \node[] () at (0.75,6.7) {$\hat{C}^h_{i_1}$};
    \node[vertex] (v3) at (1,7) {};
    \node[] () at (2,7) {$\cdots$};
    \node[vertex] (v4) at (3,7) {};
    
    \node[gray_vertex] (ll) at (5,4) {};
    \node[] at (5,3.6) {$\ell^l_{i_1,i'_1}$};
    
    \node[gray_vertex] (lh) at (5,6) {};
    \node[] at (5,6.4) {$\ell^h_{i_1,i'_1}$};
    
    \node[circle,draw] (e1) at (6.5,4) {$\hat{E}_{e_1}$};
    \node[] at (6.5,5) {$\vdots$};
    \node[circle,draw] (em) at (6.5,6) {$\hat{E}_{e_\lambda}$};
    
    \node[gray_vertex] (rl) at (8,4) {};
    \node[] at (8,3.6) {$r^l_{i_1,i'_1}$};
    
    \node[gray_vertex] (rh) at (8,6) {};
    \node[] at (8,6.4) {$r^h_{i_1,i'_1}$};
    
    \node[] () at (12.25,3.3) {$\hat{C}^l_{i'_1}$};
    \node[vertex] (v'1) at (10,3) {};
    \node[] () at (11,3) {$\cdots$};
    \node[vertex] (v'2) at (12,3) {};
    
    \node[] () at (12.25,6.7) {$\hat{C}^h_{i'_1}$};
    \node[vertex] (v'3) at (10,7) {};
    \node[] () at (11,7) {$\cdots$};
    \node[vertex] (v'4) at (12,7) {};
    
    \draw[dashed] (0,2) rectangle (4,4);
    \draw[dashed] (0,6) rectangle (4,8);
    \draw[dashed] (9,2) rectangle (13,4);
    \draw[dashed] (9,6) rectangle (13,8);
    
    
    \draw[] (v1) edge [bend left] (ll);
    \draw[] (v2)--(ll);
    \draw[] (v3) edge [bend right] (lh);
    \draw[] (v4)--(lh);
    
    \draw[] (lh)--(em.west);
    \draw[] (lh) edge [bend right] (e1.north west);
    
    \draw[] (ll)--(e1.west);
    \draw[] (ll) edge [bend left] (em.south west);
    
    \draw[] (rh)--(em.east);
    \draw[] (rh) edge [bend left] (e1.north east);
    
    \draw[] (rl)--(e1.east);
    \draw[] (rl) edge [bend right] (em.south east);
    
    \draw[] (v'1)--(rl);
    \draw[] (v'2) edge [bend right] (rl);
    \draw[] (v'3)--(rh);
    \draw[] (v'4) edge [bend left] (rh);

    \end{tikzpicture}
    \caption{Adjacency gadget $\hat{A}(i_1,i_1,i'_1,i'_1)$,
    where $E^{i_1,i'_1} = \setdef{e_i}{i \in [\lambda]}$.
    Gray vertices have $\Delta - n$ leaves attached.}
    \label{fig:dc_td_lb_adj_edge_case}
    \end{figure}
    
    To construct graph $H$, first construct the adjacency gadget $\hat{A}(1,k,1,k)$.
    Then, introduce a vertex $u$, which has an edge with the vertex $c_e$ of $\hat{E}_e$, for all $e \in E(G)$.
    Lastly, add an equality gadget $Q(p^A,u)$.
    
    \begin{lemma}\label{lem:dc_td_lb_helper}
        The number of instances of edge gadget $\hat{E}_e$ present in $H$, where $e = \braces{v^{i_1}_{j_1}, v^{i_2}_{j_2}} \in E(G)$,
        is one if $i_1 = i_2$, and two otherwise.
    \end{lemma}
    
    \begin{proof}
        First we will prove that for every adjacency gadget $\hat{A}(i_1,i_2,i'_1,i'_2)$ appearing in $H$,
        it holds that $i_2 - i_1 = i'_2 - i'_1 = 2^c - 1$, for some $c \in \N$.
        The statement holds for $\hat{A}(1,k,1,k)$, as well as when $i_2 - i_1 = i'_2 - i'_1 = 0$.
        Suppose that it holds for some $\hat{A}(i_1,i_2,i'_1,i'_2)$, i.e. $i_2 - i_1 = i'_2 - i'_1 = 2^c - 1 > 0$, for some $c \in \N$.
        Then, it follows that $\floor*{\frac{i_1+i_2}{2}} - i_1 = \floor*{i_2 - 2^{c-1} + 0.5} - i_1 = i_2 - i_1 - 2^{c-1} = 2^{c-1} - 1$.
        Moreover, it follows that $i_2 - \ceil*{\frac{i_1+i_2}{2}} = i_2 - \ceil*{i_1 + 2^{c-1} - 0.5} = i_2 - (i_1 + 2^{c-1}) = 2^{c-1} - 1$.
        Therefore, the stated property holds.
    
        In that case, for some $\hat{A}(i_1,i_2,i'_1,i'_2)$, in every step of the recursion,
        intervals $[i_1,i_2]$ and $[i'_1,i'_2]$ are partitioned in the middle,
        and an adjacency gadget is considered for each of the four combinations.
        In that case, starting from $\hat{A}(1,k,1,k)$, there is a single way to produce
        every adjacency gadget $\hat{A}(i_1,i_1,i_2,i_2)$, where $i_1,i_2 \in [k]$.
        Consider an edge gadget $\hat{E}_e$, where $e = \braces{v^{i_1}_{j_1}, v^{i_2}_{j_2}} \in E(G)$.
        There are two cases:
        \begin{itemize}
            \item if $i_1 = i_2$, then this gadget appears only in the adjacency gadget $\hat{A}(i_1,i_1,i_1,i_1)$,
            \item alternatively, it appears in both adjacency gadgets $\hat{A}(i_1,i_1,i_2,i_2)$ and $\hat{A}(i_2,i_2,i_1,i_1)$.
        \end{itemize}
        This concludes the proof.
    \end{proof}
    
    \begin{lemma}\label{lem:dc_td_lb_copy_gadget}
        Let $\mathcal{I}_1, \mathcal{I}_2$ be two instances of $\hat{C}_i$ connected by a copy gadget.
        Then, for any $(\chid,\Delta)$-coloring $c : V(H) \rightarrow [\chid]$ of $H$,
        it holds that
        \begin{itemize}
            \item the number of vertices of color $c(p^A)$ of $\hat{C}^l_i$ (respectively, $\hat{C}^h_i$) of $\mathcal{I}_1$ is equal
            to the number of vertices of color $c(p^A)$ of $\hat{C}^l_i$ (respectively, $\hat{C}^h_i$) of $\mathcal{I}_2$,
            \item the number of vertices of color $c(p^B)$ of $\hat{C}^l_i$ (respectively, $\hat{C}^h_i$) of $\mathcal{I}_1$ is equal
            to the number of vertices of color $c(p^B)$ of $\hat{C}^l_i$ (respectively, $\hat{C}^h_i$) of $\mathcal{I}_2$.
        \end{itemize}
    \end{lemma}
    
    \begin{proof}
        Due to the palette gadgets, any choice vertex of $\hat{C}_i$ is colored either $c(p^A)$ or $c(p^B)$.
        Due to the properties of the equality gadgets, it follows that vertex $f^A_i$ (respectively, $f^B_i$) is of color $c(p^A)$ (respectively, $c(p^B)$),
        and has $\Delta - n$ same colored neighboring leaves.
        Consequently, exactly $n$ choice vertices of $\hat{C}_i$ receive color $c(p^A)$ and exactly $n$ color $c(p^B)$.
    
        Let $g_1,g_2$ be the vertices present in the copy gadget connecting instances $\mathcal{I}_1, \mathcal{I}_2$.
        It follows that each of them has at most $n$ neighbors of color $c(p^A)$ apart from its leaves.
        In that case, for $j \in \braces{l,h}$, if the number of vertices of color $c(p^A)$ of $\hat{C}^j_i$ of instance $\mathcal{I}_1$ differs from the respective number of instance $\mathcal{I}_2$,
        either $g_1$ or $g_2$ has more than $\Delta$ same colored neighbors, contradiction.
    \end{proof}
    
    \begin{lemma}\label{lem:dc_td_lb_td}
        It holds that $\td(H) = \bO(k)$.
    \end{lemma}
    
    \begin{proof}
        We first observe that all equality and palette gadgets added to the graph have at most one endpoint outside of $\braces{p^A, p^B}$.
        Hence, by~\cite[Lemmata~3.6 and 3.9]{siamdm/BelmonteLM20},
        we can conclude that $\td(H) = \td(H' \setminus  \{p^A, p^B\}) + \chid + 1$,
        where $H'$ is the graph we obtain from $H$ by removing all the equality and palette gadgets.
        It therefore suffices to show that $\td(H') = \bO(k)$.
        We first remove vertex $u$.
        
        Now, let $T(\kappa)$ denote the tree-depth of $\hat{A}(i_1,i_2,i'_1,i'_2)$ in the case when $i_2 - i_1 = i'_2 - i'_1 = \kappa$.
        First, notice that, for $i_1,i_2 \in [k]$, the tree-depth of $\hat{A}(i_1,i_1,i_2,i_2)$ is at most $8$:
        remove vertices $\ell^l_{i_1,i_2}$, $\ell^h_{i_1,i_2}$, $r^l_{i_1,i_2}$ and $r^h_{i_1,i_2}$,
        resulting in the choice gadgets becoming disconnected with the edge gadgets.
        Then, it suffices to remove the vertices $f^A_i, f^B_i, f^A_j$ and $f^B_j$ from the choice gadgets,
        while the edge gadgets are trees of height $2$.
        Consequently, $T(1) \leq 8$.
    
        Now, consider the adjacency gadget $\hat{A}(i_1,i_2,i'_1,i'_2)$, where $i_2 - i_1 = i'_2 - i'_1 = \kappa$.
        This is comprised of adjacency gadgets
        \begin{multicols}{2}
        \begin{itemize}
            \item $\hat{A}\parens*{i_1, \floor*{\frac{i_1+i_2}{2}}, i'_1, \floor*{\frac{i'_1+i'_2}{2}}}$,
            \item $\hat{A}\parens*{i_1, \floor*{\frac{i_1+i_2}{2}}, \ceil*{\frac{i'_1+i'_2}{2}}, i'_2}$,
            \item $\hat{A}\parens*{\ceil*{\frac{i_1+i_2}{2}}, i_2, i'_1, \floor*{\frac{i'_1+i'_2}{2}}}$,
            \item $\hat{A}\parens*{\ceil*{\frac{i_1+i_2}{2}}, i_2, \ceil*{\frac{i'_1+i'_2}{2}}, i'_2}$.
        \end{itemize}
        \end{multicols}
        as well as of exactly $2\kappa$ original choice gadgets,
        each of which is connected with two copy gadgets to other instances of choice gadgets present in the adjacency gadgets.
        By removing all vertices $g_1$ and $g_2$ of the copy gadgets (see~\cref{fig:dc_td_lb_copy_choice_gadget}),
        all the original choice gadgets as well as the adjacency gadgets are disconnected.
        Therefore, it holds that $T(\kappa) \leq 8 \kappa + T(\kappa / 2)$,
        thus, it follows that $T(k) \leq 8 \sum_{i=0}^{\log k} \frac{k}{2^i} = \bO(k)$,
        i.e.~$\td(H') = \bO(k)$.
    \end{proof}
    
    \begin{lemma}\label{lem:dc_td_lb_cor1}
        For any $\chid \geq 2$, if $G$ contains a $k$-clique,
        then $H$ admits a $(\chid, \Delta)$-coloring.
    \end{lemma}
    
    \begin{proof}
        Let $\mathcal{V} \subseteq V(G)$ be a $k$-clique of $G$,
        consisting of vertices $v^i_{s(i)} \in V_i$, for $i \in [k]$.
        We will produce a $(\chid, \Delta)$-coloring of $H$ as follows:
        \begin{itemize}
            \item Vertex $p^A$ receives color 1 and vertex $p^B$ color 2.
            \item All vertices for which we have added an equality gadget with one endpoint identified with $p^A$ (respectively, $p^B$)
            take color 1 (respectively, 2).
            \item We use~\cite[Lemma~3.5]{siamdm/BelmonteLM20} to properly color the internal vertices of the equality gadgets.
            \item For choice gadget $\hat{C}_{i}$, we color the vertices of $\setdef{v^{i,h}_j \in \hat{C}^h_{i}}{j \in [s(i)]} \cup \setdef{v^{i,l}_j \in \hat{C}^l_{i}}{j \in [s(i)+1,n]}$
            with color 1, while we color the vertices of $\setdef{v^{i,l}_j \in \hat{C}^l_{i}}{j \in [s(i)]} \cup \setdef{v^{i,h}_j \in \hat{C}^h_{i}}{j \in [s(i)+1,n]}$
            with color 2.
            \item For every $e = \braces{v^{i_1}_{s(i_1)}, v^{i_2}_{s(i_2)}}$ that is contained in the clique,
            we color all vertices $s^{i_1}_j, s^{i_2}_j$ of $\hat{E}_e$, for $j \in [n]$,
            with color $1$, and the vertex $c_e$ with color $2$.
            For all other edges $e'$ not belonging to the clique,
            we use the opposite coloring,
            coloring vertices $s^{i_1}_j, s^{i_2}_j$ with color $2$ and vertex $c_{e'}$ with color $1$.
            \item We use~\cite[Lemma~3.8]{siamdm/BelmonteLM20} to properly color the internal vertices of palette gadgets,
            since all palette gadgets that we add use either color $1$ or color $2$ twice in their endpoints.
        \end{itemize}
        This completes the coloring.
    
        It remains to prove that this is indeed a $(\chid,\Delta)$-coloring.
        We first note that by Lemmata~3.5 and~3.8 of~\cite{siamdm/BelmonteLM20},
        internal vertices of equality and palette gadgets are properly colored.
        Vertices $p^A$, $p^B$ have exactly $\Delta$ neighbors with the same color.
        Vertices $f^A_i$ and $f^B_i$ have exactly $n$ neighbors of the same color among vertices of $\hat{C}^l_{i} \cup \hat{C}^h_{i}$,
        thus exactly $\Delta$ neighbors of the same color overall.
        The same holds for every vertex $g_1$ and $g_2$ of a copy gadget.
        Choice vertices have a constant number of neighbors of the same color (due to vertices $f^A_i, f^B_i, g_1, g_2, \ell^l_{i_1,i_2}, \ell^h_{i_1,i_2}, r^l_{i_1,i_2}, r^h_{i_1,i_2}$).
        The vertex $u$ has exactly $\deg_H (u) - (2 \binom{k}{2} + k) = \Delta$ neighbors with color $1$,
        since from~\cref{lem:dc_td_lb_helper},
        it follows that for adjacency gadget $\hat{A}(i,i,j,j)$, if $i \neq j$, then it appears twice in $H$ and once otherwise,
        while the clique contains $\binom{k}{2}$ edges $e = \braces{v^{i_1}_{s(i_1)}, v^{i_2}_{s(i_2)}}$ where $i_1 \neq i_2$ and $k$ where $i_1 = i_2$.
        Finally, for the vertices $\ell^l_{i_1,i_2}, \ell^h_{i_1,i_2}, r^l_{i_1,i_2}, r^h_{i_1,i_2}$,
        we note that $\ell^h_{i_1,i_2}$ (respectively, $\ell^l_{i_1,i_2}$) has exactly $s(i_1)$ (respectively, $n - s(i_1)$)
        neighbors with color $1$ among the choice vertices.
        Moreover, $\ell^h_{i_1,i_2}$ (respectively, $\ell^l_{i_1,i_2}$) has exactly $n - s(i_1)$ (respectively, $s(i_1)$)
        neighbors with color $1$ in the edge gadgets, those belonging to $\hat{E}_e$, where $e = \braces{v^{i_1}_{s(i_1)},v^{i_2}_{s(i_2)}}$.
        In a similar way, one can show that $r^l_{i_1,i_2}$ and $r^h_{i_1,i_2}$ have exactly $\Delta$ neighbors of the same color.
    \end{proof}
    
    \begin{lemma}\label{lem:dc_td_lb_cor2}
        For any $\chid \geq 2$, if $H$ admits a $(\chid, \Delta)$-coloring,
        then $G$ contains a $k$-clique.
    \end{lemma}
    
    \begin{proof}
        Assume $c : V(H) \rightarrow [\chid]$ is a $(\chid, \Delta)$-coloring of $H$.
        It holds that $c(p^A) \neq c(p^B)$,
        since each vertex has $\Delta$ neighboring leaves of the same color due to the properties of the equality gadget.
        Assume without loss of generality that $c(p^A) = 1$ and $c(p^B) = 2$.
        For every instance of a choice gadget $\hat{C}_{i}$,
        it holds that half of its choice vertices are colored with color $1$ and the other half with color $2$:
        due to the palette gadgets, every vertex is colored with one of these two colors, and if more than $n$ were colored with the same color,
        one of $f^A_i, f^B_i$ would have more than $\Delta$ same colored neighbors.
        Moreover, due to~\cref{lem:dc_td_lb_copy_gadget}, it follows that on every instance of the choice gadget $\hat{C}_{i}$ present in $H$,
        the number of vertices colored with color $1$ in $\hat{C}^l_{i}$ (respectively, $\hat{C}^h_{i}$) is the same.
        We will first prove the following claim.
    
        \begin{claim}\label{claim:dc_td_claim}
            For every $i_1,i_2 \in [k]$,
            there exists a single edge gadget $\hat{E}_e$ present in the adjacency gadget $\hat{A}(i_1,i_1,i_2,i_2)$ where $c(c_e) = 2$.
        \end{claim}
    
        \begin{claimproof}
            First, notice that due to~\cref{lem:dc_td_lb_helper}, $u$ has a total of $2(E(G) - kn) + kn$ neighbors.
            Consequently, at least $2(E(G) - kn) + kn - \Delta = 2 \binom{k}{2} + k$ of its neighbors are colored with color $2$,
            since $u$ is of color $1$ due to the equality gadget $Q(p^A, u)$.
            Notice that this is also the number of adjacency gadgets $\hat{A}(i_1,i_1,i_2,i_2)$, since every such gadget appears twice if $i_1 \neq i_2$ and once otherwise.
            Assume that for some adjacency gadget $\hat{A}(i_1,i_1,i_2,i_2)$ there is an edge gadget $\hat{E}_e$ such that $c(c_e) = 2$.        
            In that case, $c_e$ has $\Delta$ same colored neighboring leaves, therefore all vertices $s^{i_1}_j, s^{i_2}_j$ of $\hat{E}_e$ are colored with color $1$
            (if $\chid \geq 3$, this is a consequence of the palette gadget attached to those vertices).
            If there existed a second edge gadget $\hat{E}_{e'}$ present in $\hat{A}(i_1,i_1,i_2,i_2)$ with $c(c_{e'}) = 2$,
            then it follows that $\ell^l_{i_1,i_2}$ and $\ell^h_{i_1,i_2}$ have in total $2n$ neighbors colored with color $1$ and belonging to the edge gadgets.
            Consequently, since both already have $\Delta - n$ neighboring same colored leaves, it follows that each has exactly $n$ neighbors from the edge gadgets
            which are colored with color $1$, and no such neighbor from the choice gadget $\hat{C}_{i_1}$, which results in a contradiction,
            since exactly $n$ choice vertices of $\hat{C}_{i_1}$ are colored with color $1$.
            Consequently, for $u$ to have at least $2 \binom{k}{2} + k$ neighbors with color $2$,
            it follows that each adjacency gadget $\hat{A}(i_1,i_1,i_2,i_2)$ has a single edge gadget $\hat{E}_e$ such that $c(c_e) = 2$.
        \end{claimproof}
    
        Let $\mathcal{V} \subseteq V(G)$ be a set of cardinality $k$, containing vertex $v^i_{s(i)} \in V_i$ if, for choice gadget $\hat{C}_{i}$,
        it holds that exactly $s(i)$ vertices of $\hat{C}^h_{i}$ and $n - s(i)$ vertices of $\hat{C}^l_{i}$ are colored with color $1$.
        Notice that $\mathcal{V} \cap V_i \neq \emptyset$, for every $i$.
        We will prove that $\mathcal{V}$ is a clique.
    
        Let $v^{i_1}_{s(i_1)}, v^{i_2}_{s(i_2)}$ belong to $\mathcal{V}$.
        Consider the adjacency gadget $\hat{A}(i_1,i_1,i_2,i_2)$.
        We will prove that it contains an edge gadget $\hat{E}_e$, where $e = \braces{v^{i_1}_{s(i_1)}, v^{i_2}_{s(i_2)}}$.
        Consider the vertices $\ell^h_{i_1,i_2}$, $\ell^l_{i_1,i_2}$, $r^h_{i_1,i_2}$ and $r^l_{i_1,i_2}$.
        Each one of them has attached $\Delta - n$ neighboring leaves of color $1$.
        Moreover, it holds that
        \begin{itemize}
            \item $\ell^h_{i_1,i_2}$ has $s(i_1)$ same colored neighbors from $\hat{C}^h_{i_1}$,
            \item $\ell^l_{i_1,i_2}$ has $n - s(i_1)$ same colored neighbors from $\hat{C}^l_{i_1}$,
            \item $r^h_{i_1,i_2}$ has $s(i_2)$ same colored neighbors from $\hat{C}^h_{i_2}$,
            \item $r^l_{i_1,i_2}$ has $n - s(i_2)$ same colored neighbors from $\hat{C}^l_{i_2}$.
        \end{itemize}
        Notice that due to~\cref{claim:dc_td_claim}, there exists a single edge gadget $\hat{E}_{e'}$, where $e' = \braces{v^{i_1}_{j_1}, v^{i_2}_{j_2}}$,
        for which $c(c_{e'}) = 2$ holds, thus all vertices $s^{i_1}_j, s^{i_2}_j$ have color $1$.
        For $c$ to be a $(\chid,\Delta)$-coloring,
        it holds that $\ell^l_{i_1,i_2}$ (respectively, $\ell^h_{i_1,i_2}$) has at most $s(i_1)$ (respectively, $n - s(i_1)$) neighbors from the set $\setdef{s^{i_1}_j}{j \in [n]}$.
        Consequently, it follows that $\ell^l_{i_1,i_2}$ (respectively, $\ell^h_{i_1,i_2}$) has exactly $s(i_1)$ (respectively, $n - s(i_1)$) neighbors from said set.
        Similarly, it follows that $r^l_{i_1,i_2}$ (respectively, $r^h_{i_1,i_2}$) has exactly $s(i_2)$ (respectively, $n - s(i_2)$) neighbors from the set $\setdef{s^{i_2}_j}{j \in [n]}$.
        Consequently, $e' = \braces{v^{i_1}_{j_1}, v^{i_2}_{j_2}}$, thus there exists such an edge in $G$.
    
        Since this holds for any two vertices belonging to $\mathcal{V}$, it follows that $G$ has a $k$-clique.
    \end{proof}
    
    Therefore, in polynomial time, we can construct a graph $H$,
    of tree-depth $\td = \bO(k)$ due to~\cref{lem:dc_td_lb_td},
    such that, due to~\cref{lem:dc_td_lb_cor1,lem:dc_td_lb_cor2},
    deciding whether $H$ admits a $(\chid, \Delta)$-coloring
    is equivalent to deciding whether $G$ has a $k$-clique.
    In that case, assuming there exists a $f(\td) |V(H)|^{o(\td)}$ algorithm for \DC,
    where $f$ is any computable function,
    one could decide \kMC{} in time $f(\td) |V(H)|^{o(\td)} = g(k) n^{o(k)}$,
    for some computable function $g$,
    which contradicts the ETH.
\end{proof}

\section{Vertex Cover Lower Bounds}\label{sec:vc_lb}
In this section we present lower bounds on the complexity of solving \BDD{} and \DC{} when parameterized by the vertex cover number of the input graph.
In both cases we start from a 3-SAT instance of $n$ variables, and produce an equivalent instance
where the input graph has vertex cover $\bO (n / \log n)$, hence any algorithm
solving the latter problem in time $\vc^{o(\vc)}n^{\bO(1)}$ would refute the ETH.
As a consequence of the above, already known algorithms for both of these problems are essentially optimal.
We start by presenting some necessary tools used in both reductions,
and then prove the stated results.

\subsection{Preliminary Tools}

We first define a constrained version of 3-SAT, called \SD.
This variant is closely related to the (3,4)-SAT problem~\cite{dam/Tovey84} which asks whether a given formula $\phi$ is satisfiable,
where $\phi$ is a 3-SAT formula each clause of which contains exactly $3$ different variables and each variable occurs in at most $4$ clauses.
As observed by Bonamy et al.~\cite{toct/BonamyKPSW19},
a corollary of Tovey's work~\cite{dam/Tovey84} is that there is no $2^{o(n)}$ algorithm for (3,4)-SAT unless the ETH is false,
where $n$ denotes the number of variables of the formula.
Here we prove an analogous lower bound for \SD.
Subsequently, by closely following Lemma~3.2 from~\cite{toct/BonamyKPSW19},
we present a way to partition the formula's variables and clauses into groups such that
variables appearing in clauses of the same clause group belong to different variable groups.

\problemdef{\SD}
{A 3-SAT formula $\phi$ every clause of which contains exactly 3 distinct variables and each variable appears in at most 4 clauses.}
{Determine whether there exists an assignment to the variables of $\phi$ such that each clause has exactly one True literal.}

\begin{theorem}\label{thm:3_4_xsat_hardness}
    \SD{} cannot be decided in time $2^{o(n)}$,
    where $n$ denotes the number of variables of the input formula,
    unless the ETH fails.
\end{theorem}

\begin{proof}
    We will closely follow a known reduction from 3SAT to 1-\textsc{in}-3-SAT, which is a simplification of Schaefer's work~\cite{stoc/Schaefer78}.
    Let $\phi$ be a (3,4)-SAT formula of $n$ variables and $m \leq \frac{4}{3} \cdot n$ clauses.
    Let $V$ denote the set of its variables, and $C$ the set of its clauses.
    We will construct an equivalent instance $\phi'$ of \SD,
    where $V'$ denotes the set of its variables and $C'$ the set of its clauses,
    as follows:
    \begin{itemize}
        \item for every variable $x \in V$, introduce a variable $x \in V'$,

        \item for every clause $c_i \in C$, introduce variables $\alpha_i, \beta_i, \gamma_i, \delta_i \in V'$,
        
        \item for every clause $c_i = x \lor y \lor z$ of $\phi$,
        introduce clauses $c^1_i = \neg x \lor \alpha_i \lor \beta_i$,
        $c^2_i = y \lor \beta_i \lor \gamma_i$ and
        $c^3_i = \neg z \lor \gamma_i \lor \delta_i$ in $\phi'$.
    \end{itemize}
    Notice that $\phi'$ is a valid \SD{} instance, since every one of its $3m$ clauses contains exactly 3 different variables,
    while all of its $n + 4m$ variables appear in at most 4 clauses.
    In the following we prove that $\phi$ is satisfiable if and only if there exists an assignment to the variables of $\phi'$ such that each of its clauses has exactly one True literal.

    For the forward direction, let $f : V \to \braces{T, F}$ be a satisfying assignment for $\phi$.
    Consider an assignment $f' : V' \to \braces{T, F}$ such that $f'(x) = f(x)$, for all $x \in V$.
    It remains to determine the value of $f'$ for any variable belonging to $V' \setminus V$.
    Let $c_i = x \lor y \lor z$ be a clause of $\phi$.
    Since $f$ is a satisfying assignment, it holds that at least one of $x, y, z$ has a truthful assignment.
    If $f(y) = T$ then let $f'(\beta_i) = f'(\gamma_i) = F$, while $f'(\alpha_i) = f(x)$ and $f'(\delta_i) = f(z)$.
    On the other hand, if $f(y) = F$ then one of the following cases holds:
    \begin{romanenumerate}
        \item $f(x) = f(z) = T$. Then, set $f'(\alpha_i) = T$, $f'(\beta_i) = F$, $f'(\gamma_i) = T$, and $f'(\delta_i) = F$.
        \item $f(x) = T$ and $f(z) = F$. Then, set $f'(\alpha_i) = F$, $f'(\beta_i) = T$, $f'(\gamma_i) = F$, and $f'(\delta_i) = F$.
        \item $f(x) = F$ and $f(z) = T$. Then, set $f'(\alpha_i) = F$, $f'(\beta_i) = F$, $f'(\gamma_i) = T$, and $f'(\delta_i) = F$.
    \end{romanenumerate}        
    Notice that $f'$ is an assignment on $V'$ such that all clauses of $\phi'$ have exactly one True literal.

    For the converse direction, let $f' : V' \to \braces{T, F}$ be an assignment such that every clause of $\phi'$ has exactly one True literal,
    and let $f : V \to \braces{T, F}$ be the restriction of $f'$ to $V$, i.e.~$f(x) = f'(x)$, for all $x \in V$.
    We will prove that $f$ is a satisfying assignment for $\phi$.
    Indeed, assume there exists a clause $c_i = x \lor y \lor z$ of $\phi$ which is not satisfied by $f$.
    In that case, $f(x) = f(y) = f(z) = F$, and since $f$ is a restriction of $f'$ it follows that
    \begin{itemize}
        \item $f'(\alpha_i) = f'(\beta_i) = F$ due to $c^1_i$,
        \item $f'(\beta_i) \neq f'(\gamma_i)$ due to $c^2_i$ and
        \item $f'(\gamma_i) = f'(\delta_i) = F$ due to $c^3_i$,
    \end{itemize}
    which is a contradiction.
    
    Lastly, assume there exists a $2^{o(|V'|)}$ algorithm for \SD.
    Then, since $|V'| = n + 4m$, (3,4)-SAT could be decided in $2^{o(n)}$, thus the ETH fails.
\end{proof}

We proceed by proving that, given a \SD{} instance, we can partition the variables and clauses of the formula into groups such that
variables appearing in clauses of the same clause group belong to different variable groups.

\begin{lemma}\label{lem:x_3_4_sat_partition}
    Let $\phi$ be an instance of \SD, where $V$ denotes the set of its $n$ variables and $C$ the set of its clauses.
    Moreover, let $b \leq \sqrt{n}$.
    One can produce in time $n^{\bO(1)}$ a partition of $\phi$'s variables into $n_V$ disjoint sets $V_1, \ldots, V_{n_V}$ of size at most $b$
    as well as a partition of its clauses into $n_C$ disjoint sets $C_1, \ldots, C_{n_C}$ of size at most $\sqrt{n}$,
    for some integers $n_V = \bO(n / b)$ and $n_C = \bO(\sqrt{n})$, such that,
    for any $i \in [n_C]$, any two variables appearing in clauses of $C_i$ belong to different variable subsets,
    while any variable appears in at most $1$ clause of $C_i$.
\end{lemma}

\begin{proof}
    We will first partition the variable set $V$ into disjoint subsets $V_1, \ldots, V_{n_V}$,
    where $n_V = \bO(n / b)$,
    such that $|V_i| \leq b$, and variables appearing in the same clause belong to different $V_i$'s.
    Consider the graph $G_1$, which has a vertex per variable of $\phi$,
    while two vertices have an edge if there exists a clause in $\phi$ that both corresponding variables appear in,
    i.e.~$G_1$ is the primal graph of $\phi$.
    $G_1$ does not have any loops, since no clause contains repeated variables.
    Since every variable of $\phi$ occurs in at most four clauses and since those clauses contain two other variables,
    the maximum degree of $G_1$ is at most $8$.
    Hence, $G_1$ can be greedily colored with $9$ colors, thus inducing a partition into different colored groups of size $n_1, \ldots, n_9$ respectively,
    where $n_1 + \ldots + n_9 = n$.
    Subsequently, we refine said partition so that every group has size at most $b$,
    resulting in at most
    \[
        n_V =
        \sum_{i=1}^9 \ceil*{\frac{n_i}{b}} \leq
        9 + \sum_{i=1}^9 \frac{n_i}{b} =
        9 + \frac{n}{b}
    \]
    groups $V_1, \ldots, V_{n_V}$.
    Notice that it holds that, variables appearing in the same clause belong to different $V_i$'s,
    since any two such variables are adjacent in $G_1$ and thus get different colors.

    Next, we partition the clause set $C$ into disjoint subsets $C_1, \ldots, C_{n_C}$, where $n_C = \bO(\sqrt{n})$,
    such that $|C_i| \leq \sqrt{n}$,
    and any two variables appearing in clauses of a group $C_i$ belong to different variable groups.
    For this, consider the graph $G_2$ with clauses as vertices and with an edge between clauses if they contain variables from the same variable group.
    $G_2$ has no loops, since any variables occurring in the same clause belong to different $V_i$'s.
    Since every clause contains exactly $3$ variables,
    each variable group has size at most $b$,
    and every variable occurs in at most $4$ clauses,
    the maximum degree of $G_2$ is at most $12 b$.
    We can therefore color $G_2$ greedily with $12 b + 1$ colors.
    Similarly as before, we refine said partition into $n_C \leq 12 b + 1 + |C| / \sqrt{n}$ subsets
    $C_1, \ldots, C_{n_C}$ of size at most $\sqrt{n}$ each.
    By the construction of the coloring, it follows that if a variable $v \in V_i$ appears in some clause $c \in C_j$,
    then no variable belonging to $V_i$ appears in any clause in $C_j \setminus \braces{c}$,
    while no other variable of $V_i$ appears in $c$.
    Moreover, since $|C| \leq \frac{4}{3} \cdot n$ and $b \leq \sqrt{n}$, it follows that $n_C = \bO(\sqrt{n})$.
\end{proof}

\begin{definition}
    A \emph{$d$-detecting family} for a finite set $U$ is a family $\mathcal{F} \subseteq 2^U$ of subsets of $U$
    such that, for every two functions $f, g : U \to \braces{0, \ldots, d-1}$ where $f \neq g$,
    there exists $S \in \mathcal{F}$ such that
    $\sum_{x \in S} f(x) \neq \sum_{x \in S} g(x)$.
\end{definition}

Lindstr\"{o}m~\cite{lindstrom_1965} has provided a deterministic construction of sublinear $d$-detecting families,
while Bonamy et al.~\cite{toct/BonamyKPSW19} were the first to use them in the context of computational complexity,
proving tight lower bounds for the \textsc{Multicoloring} problem under the ETH.
The following theorem will be crucial towards proving the stated lower bounds.

\begin{theorem}[\cite{lindstrom_1965}]\label{thm:d_detecting}
    For every constant $d \in \N$ and finite set $U$,
    there is a $d$-detecting family $\mathcal{F}$ on $U$ of size $\frac{2 |U|}{\log_d |U|} \cdot (1 + o(1))$.
    Moreover, $\mathcal{F}$ can be constructed in time polynomial in $|U|$.
\end{theorem}

\subsection{Bounded Degree Vertex Deletion}\label{subsec:bdd_vc_lb}

\begin{theorem}\label{thm:bdd_vc_lb}
    There is no $\vc^{o (\vc)}n^{\bO(1)}$ time algorithm for \BDD,
    where \vc{} denotes the size of the minimum vertex cover of the input graph,
    unless the ETH fails.
\end{theorem}

\begin{proof}
    Let $\phi$ be an instance of \SD{} of $n$ variables.
    Assume without loss of generality that $n$ is a power of $4$
    (this can be achieved by adding dummy variables to the instance if needed).
    Making use of~\cref{lem:x_3_4_sat_partition} one can obtain in time $n^{\bO(1)}$ the following:
    \begin{itemize}
        \item a partition of $\phi$'s variables into subsets $V_1, \ldots, V_{n_V}$,
        where $|V_i| \leq \log n$ and $n_V = \bO(n / \log n)$,
        \item a partition of $\phi$'s clauses into subsets $C_1, \ldots, C_{n_C}$,
        where $|C_i| \leq \sqrt{n}$ and $n_C = \bO(\sqrt{n})$,
    \end{itemize}
    where any two variables occurring in clauses of the same clause subset belong to different variable subsets.
    For $i \in [n_C]$, let $\braces{C_{i, 1}, \ldots, C_{i, n^i_\mathcal{F}}}$ be a 4-detecting family of subsets of $C_i$ for some $n^i_\mathcal{F} = \bO(\sqrt{n} / \log n)$,
    produced in time $n^{\bO(1)}$ due to~\cref{thm:d_detecting}.
    Moreover, let $n_\mathcal{F} = \max_{i \in [n_C]} n^i_\mathcal{F}$.
    Define $\Delta = n^3$ and $k = n_V$.
    We will construct a graph $G = (V, E)$ such that there exists $S \subseteq V(G)$
    of size $|S| \leq k$ and $G - S$ has maximum degree at most $\Delta$
    if and only if there exists an assignment such that every clause of $\phi$ has exactly one True literal.
    
    \proofsubparagraph*{Choice Gadget.}
    For each variable subset $V_i$ we define the choice gadget graph $G_i$ as follows:
    \begin{itemize}
        \item introduce vertices $\kappa_i$, $\lambda_i$, and $v^i_j$, where $j \in [n]$, 
        \item add edges $\braces{\kappa_i, v^i_j}$ and $\braces{\lambda_i, v^i_j}$ for all $j \in [n]$,
        \item attach sufficiently many leaves to $\kappa_i$ and $\lambda_i$ such that their degree is $\Delta + 1$.
    \end{itemize}
    Let $\mathcal{V}_i = \setdef{v^i_j}{j \in [n]}$, for $i = 1, \ldots, n_V$.
    We fix an arbitrary one-to-one mapping so that every vertex of $\mathcal{V}_i$ corresponds to a different assignment for the variables of $V_i$.
    Since $2^{|V_i|} \leq n$, there are sufficiently many vertices to uniquely encode all the different assignments of $V_i$.
    Let $\mathcal{V} = \mathcal{V}_1 \cup \ldots \cup \mathcal{V}_{n_V}$ denote the set of all such vertices.

    \proofsubparagraph*{Clause Gadget.}
    For $i \in [n_C]$, let $C_i$ be a clause subset and $\braces{C_{i, 1}, \ldots, C_{i, n^i_\mathcal{F}}}$ its 4-detecting family.
    For every subset $C_{i,j}$ of the 4-detecting family introduce vertices $c_{i,j}$ and $c'_{i,j}$.
    Add an edge between $c_{i,j}$ and $v^p_q$ if there exists variable $x \in V_p$ such that $x$ appears in some clause $c \in C_{i,j}$,
    and $v^p_q$ corresponds to an assignment of $V_p$ that satisfies $c$.
    Due to~\cref{lem:x_3_4_sat_partition}, $c_{i,j}$ has exactly $|C_{i,j}| \cdot \frac{3n}{2}$ such edges:
    there are exactly $3 |C_{i,j}|$ different variables appearing in clauses of $C_{i,j}$, each belonging to a different variable subset,
    and for each such variable, half the assignments of the corresponding variable subset result in the satisfaction of the corresponding clause of $C_{i,j}$.
    Attach to $c_{i,j}$ a sufficient number of leaves such that its total degree is $\Delta + |C_{i,j}|$.
    Moreover, for $v \in \mathcal{V}$, let $v \in N(c'_{i,j})$ if $v \notin N(c_{i,j})$.
    Notice that then, it holds that $N(c_{i,j}) \cup N(c'_{i,j}) \supseteq \mathcal{V}$,
    while $N(c_{i,j}) \cap N(c'_{i,j}) = \emptyset$.
    Lastly, attach to $c'_{i,j}$ a sufficient number of leaves such that its total degree is $\Delta + (k - |C_{i,j}|)$.
    
    Let $\mathcal{I} = (G, \Delta, k)$ be an instance of \BDD.
    
    \begin{lemma}\label{lem:bdd_vc_bound}
        It holds that $\vc(G) = \bO(n / \log n)$.
    \end{lemma}
    
    \begin{proof}
        Notice that the deletion of all vertices $\kappa_i$, $\lambda_i$, $c_{i,j}$ and $c'_{i,j}$ induces an independent set.
        Therefore,
        \[
            \vc(G) \leq 2 n_V + 2 \sum_{i = 1}^{n_C} n^i_\mathcal{F} \leq
            2 n_V + 2 n_C \cdot n_\mathcal{F} =
            \bO \parens*{\frac{n}{\log n} + \sqrt{n} \cdot \frac{\sqrt{n}}{\log n}} =
            \bO \parens*{\frac{n}{\log n}},
        \]
        and the statement follows.
    \end{proof}
    
    \begin{lemma}\label{lem:bdd_vc_correctness_1}
        If $\phi$ is a Yes instance of \SD, then $\mathcal{I}$ is a Yes instance of \BDD.
    \end{lemma}
    
    \begin{proof}
        Let $f : V_1 \cup \ldots \cup V_{n_V} \to \braces{T,F}$ be an assignment such that every clause of $\phi$ has exactly 1 True literal.
        Let $S$ contain from each $\mathcal{V}_i$ the vertex corresponding to this assignment restricted to $V_i$.
        It holds that $|S| = n_V = k$.
        We will prove that $G - S$ has maximum degree at most $\Delta$.
        First, notice that in $G-S$, any vertex $v^q_p$ has at most $2 + n_C \cdot n_\mathcal{F} \leq \Delta$ neighbors,
        while any vertex $\kappa_i$ and $\lambda_i$ has degree $\Delta$,
        since $|S \cap \mathcal{V}_i| = 1$ for all $i \in [n_V]$.
        Lastly, for each vertex $c_{i,j}$ corresponding to a clause set $C_{i,j}$,
        it holds that, since $f$ is an assignment where every clause of $\phi$ has exactly 1 True literal,
        and due to~\cref{lem:x_3_4_sat_partition}, $S$ contains exactly $|C_{i,j}|$ neighbors of $c_{i,j}$.
        In that case, the remaining $k - |C_{i,j}|$ vertices belonging to $S$ are neighbors of $c'_{i,j}$.
    \end{proof}
    
    \begin{lemma}\label{lem:bdd_vc_correctness_2}
        If $\mathcal{I}$ is a Yes instance of \BDD, then $\phi$ is a Yes instance of \SD.
    \end{lemma}
    
    \begin{proof}
        Let $S \subseteq V(G)$, where $|S| \leq k$ and $G - S$ has maximum degree at most $\Delta$.
        We will first prove that $S$ contains a single vertex from every set $\mathcal{V}_i$.
    
        \begin{claim}
            $|S \cap \mathcal{V}_i| = 1$, for all $i \in [n_V]$.
        \end{claim}
    
        \begin{claimproof}
            We will first prove that $|S \cap V(G_i)| = 1$, for all $i \in [n_V]$.
            Suppose that this is not the case.
            Then, since $|S| \leq n_V$, it holds that there exists some $i$ such that $|S \cap V(G_i)| = 0$.
            In that case, $\deg_{G - S} (\kappa_i) = \Delta + 1$, contradiction.
            Lastly, suppose there exists $i$ such that, $v \in S$ and $v \in V(G_i) \setminus \mathcal{V}_i$.
            Since $\deg_G (\kappa_i) = \deg_G (\lambda_i) = \Delta + 1$, and $v \notin N(\kappa_i) \cap N(\lambda_i)$,
            it follows that one of $\kappa_i, \lambda_i$ does not belong in $S$ and has degree $\Delta + 1$ in $G - S$, contradiction.
        \end{claimproof}
        Now consider the assignment $h : V_1 \cup \ldots \cup V_{n_V} \to \braces{T,F}$ of the variables of $\phi$ depending on which vertex of $\mathcal{V}_i$ belongs to $S$.
    
        Let $C_i$, where $i \in [n_C]$, be a clause subset resulting from the
        partition due to~\cref{lem:x_3_4_sat_partition}.
        Let $f : C_i \to \braces{0,1,2,3}$ be the function assigning to each clause $c \in C_i$ the number of vertices $v \in S$ such that
        $v \in \mathcal{V}_j$ corresponds to an assignment $V_j \to \braces{T,F}$ that satisfies $c$.
        Notice that since every clause contains 3 literals and $|\mathcal{V}_i \cap S| = 1$,
        $f(c) \leq 3$ follows.
        It holds that $\sum_{c \in C_{i,j}} f(c) = |C_{i,j}|$, for any $j \in [n^i_\mathcal{F}]$:
        $|S \cap N(c_{i,j})| \geq |C_{i,j}|$, $|S \cap N(c'_{i,j})| \geq k - |C_{i,j}|$, while $N(c_{i,j}) \cap N(c'_{i,j}) = \emptyset$
        and $|S| = k$.
    
        Now consider $g : C_i \to \braces{0,1,2,3}$ to be the constant function $g \equiv 1$.
        Notice that $\sum_{c \in C_{i,j}} f(c) = \sum_{c \in C_{i,j}} g(c)$, for all $j \in [n^i_\mathcal{F}]$.
        Since $\braces{C_{i,1}, \ldots, C_{i, n^i_\mathcal{F}}}$ is a 4-detecting family, this implies that $f = g$.
        Thus, for every clause $c \in C_i$, it holds that $f(c) = 1$,
        meaning that there exists a single vertex in $S$
        which corresponds to a partial assignment that satisfies $c$.
        Since $i = 1, \ldots, n_C$ was arbitrary, this concludes the proof that $h$ is a satisfying assignment for $\phi$.
    \end{proof}
        
    Therefore, in polynomial time, we can construct a graph $G$ with vertex cover number $\vc  = \bO(n / \log n)$ due to~\cref{lem:bdd_vc_bound},
    such that, due to~\cref{lem:bdd_vc_correctness_1,lem:bdd_vc_correctness_2},
    deciding whether there exists $S \subseteq V(G)$ of size $|S| \leq k$ and $G - S$ has maximum degree at most $\Delta$
    is equivalent to deciding if there exists an assignment such that every clause of $\phi$ has exactly one True literal.
    In that case, assuming there exists a $\vc^{o(\vc)} n^{\bO(1)}$ algorithm for \BDD,
    one could decide \SD{} in time
    \[
        \vc^{o(\vc)} n^{\bO(1)} = \parens*{\frac{n}{\log n}}^{o(n / \log n)} n^{\bO(1)} =
        2^{(\log n - \log \log n) o(n / \log n)} =
        2^{o(n)},
    \]
    which contradicts the ETH due to~\cref{thm:3_4_xsat_hardness}.
\end{proof}

\subsection{Defective Coloring}\label{sec:dc_vc_lb}

\begin{theorem}\label{thm:dc_vc_lb}
    For any fixed $\chid \geq 2$,
    if there exists an algorithm that solves \DC{} in time $\vc^{o (\vc)}n^{\bO(1)}$,
    where \vc{} denotes the size of the minimum vertex cover of the input graph,
    then the ETH is false.
\end{theorem}

\begin{proof}
    Let $\phi$ be an instance of \SD{} of $n$ variables.
    Assume without loss of generality that $n$ is a power of $16$
    (this can be achieved by adding dummy variables to the instance if needed).
    Making use of~\cref{lem:x_3_4_sat_partition}, one can obtain in time $n^{\bO(1)}$ the following:
    \begin{itemize}
        \item a partition of $\phi$'s variables into subsets $V_1, \ldots, V_{n_V}$,
        where $|V_i| \leq \log \sqrt[4]{n}$ and $n_V = \bO(n / \log n)$,
        \item a partition of $\phi$'s clauses into subsets $C_1, \ldots, C_{n_C}$,
        where $|C_i| \leq \sqrt{n}$ and $n_C = \bO(\sqrt{n})$,
    \end{itemize}
    where any two variables occurring in clauses of the same clause subset belong to different variable subsets.
    For $i \in [n_C]$, let $\braces{C_{i, 1}, \ldots, C_{i, n^i_\mathcal{F}}}$ be a 4-detecting family of subsets of $C_i$ for some $n^i_\mathcal{F} = \bO(\sqrt{n} / \log n)$,
    produced in time $n^{\bO(1)}$ due to~\cref{thm:d_detecting}.
    Moreover, let $n_\mathcal{F} = \max_{i \in [n_C]} n^i_\mathcal{F}$.
    Define $\Delta = n_V + n_C \cdot n_\mathcal{F}$.
    We will construct a graph $G = (V, E)$ such that $G$ admits a $(\chid, \Delta)$-coloring
    if and only if there exists an assignment such that every clause of $\phi$ has exactly one True literal.
    
    \proofsubparagraph*{Main Palette Vertices.}
    Introduce vertices $b_i$ and $r_j$, where $i \in [\Delta + 1]$ and $j \in [2 \Delta + 1]$,
    such that they comprise a complete bipartite graph
    (i.e.~all edges $\braces{b_i, r_j}$ are present).
    We will refer to those as the \emph{main palette vertices}.
    
    \proofsubparagraph*{Choice Gadget.}
    For each variable subset $V_i$ we define the choice gadget graph $G_i$ as follows:
    \begin{itemize}
        \item introduce vertices $\kappa_i$, $\lambda_i$, and $v^i_j$, where $j \in [\sqrt[4]{n}]$,
        \item add edges $\braces{\kappa_i, v^i_j}$ and $\braces{\lambda_i, v^i_j}$, for all $j \in [\sqrt[4]{n}]$,
        \item add edges $\braces{\kappa_i, r_p}$ and $\braces{\kappa_i, b_q}$, for $p \in [\Delta + 1]$ and $q \in [\Delta - 1]$,
        \item add edges $\braces{\lambda_i, b_p}$ and $\braces{\lambda_i, r_q}$, for $p \in [\Delta + 1]$ and $q \in [\Delta - (\sqrt[4]{n} - 1)]$.
    \end{itemize}
    Let $\mathcal{V}_i = \setdef{v^i_j}{j \in [\sqrt[4]{n}]}$, for $i = 1, \ldots, n_V$.
    We fix an arbitrary one-to-one mapping so that every vertex of $\mathcal{V}_i$ corresponds to a different assignment for the variables of $V_i$.
    Since $2^{|V_i|} \leq \sqrt[4]{n}$, there are sufficiently many vertices to uniquely encode all the different assignments of $V_i$.
    Let $\mathcal{V} = \mathcal{V}_1 \cup \ldots \cup \mathcal{V}_{n_V}$ denote the set of all such vertices.

    \proofsubparagraph*{Clause Gadget.}
    For $i \in [n_C]$, let $C_i$ be a clause subset and $\braces{C_{i, 1}, \ldots, C_{i, n^i_\mathcal{F}}}$ its 4-detecting family.
    For every subset $C_{i,j}$ of the 4-detecting family, introduce vertices $c_{i,j}$ and $c'_{i,j}$.
    Add an edge between $c_{i,j}$ and $v^p_q$ if there exists variable $x \in V_p$ such that $x$ occurs in some clause $c \in C_{i,j}$,
    and $v^p_q$ corresponds to an assignment of $V_p$ that satisfies $c$.
    Due to~\cref{lem:x_3_4_sat_partition}, $c_{i,j}$ has exactly $|C_{i,j}| \cdot \frac{3 \sqrt[4]{n}}{2}$ such edges:
    there are exactly $3 |C_{i,j}|$ different variables appearing in clauses of $C_{i,j}$, each belonging to a different variable subset,
    and for each such variable, half the assignments of the corresponding variable subset result in the satisfaction of the corresponding clause of $C_{i,j}$.
    Additionally, add edges $\braces{c_{i,j}, r_p}$ and $\braces{c_{i,j}, b_q}$, where $p \in [\Delta + 1]$ and $q \in [\Delta - |C_{i,j}|]$.
    Regarding $c'_{i,j}$, add an edge $\braces{c'_{i,j}, v}$ for all $v \in \mathcal{V}$
    such that $v \in N(c_{i,j})$.
    Finally, add edges $\braces{c'_{i,j}, b_p}$ and $\braces{c'_{i,j}, r_q}$, where $p \in [\Delta + 1]$ and $q \in [\Delta - (|C_{i,j}| \cdot \frac{3 \sqrt[4]{n}}{2} - |C_{i,j}|)]$.

    \proofsubparagraph*{Secondary Palette Vertices.}
    If $\chid \geq 3$, then introduce independent sets $P_i = \setdef{p^i_j}{j \in [i \cdot \Delta + 1]}$, for $i \in [3, \chid]$.
    We will refer to the vertices of $P = P_3 \cup \ldots \cup P_{\chid}$ as \emph{secondary palette vertices}.
    Add an edge from every secondary palette vertex $p^i_j \in P_i$
    to all vertices introduced except those belonging to $P_i$.    
    
    \begin{lemma}\label{lem:dc_vc_size}
        It holds that $\vc(G) = \bO(n / \log n)$.
    \end{lemma}
    
    \begin{proof}
        Notice that the deletion of all main and secondary palette vertices as well as of all vertices
        $\kappa_i$, $\lambda_i$, $c_{i,j}$ and $c'_{i,j}$ induces an independent set.
        Therefore,
        \[
            \vc(G) \leq
            \sum_{i=1}^{\chid} (i \cdot \Delta + 1) + 2 n_V + 2 \sum_{i = 1}^{n_C} n^i_\mathcal{F} \leq
            \chid^2 \cdot \Delta + \chid + 2 n_V + 2 n_C \cdot n_\mathcal{F} =
            \bO \parens*{\frac{n}{\log n}},
        \]
        and the statement follows.
    \end{proof}
    
    \begin{lemma}\label{lem:dc_vc_correctness_1}
        For any $\chid \geq 2$,
        if $\phi$ is a Yes instance of \SD, then $G$ has a $(\chid, \Delta)$-coloring.
    \end{lemma}
    
    \begin{proof}
        Let $f : V_1 \cup \ldots \cup V_{n_V} \to \braces{T,F}$ be an assignment such that every clause of $\phi$ has exactly 1 True literal.
        If $\chid \geq 3$, color every secondary palette vertex of $P_i$ with color $i$, for $i \in [3,\chid]$,
        thus resulting in that many independent sets of distinct colors.
        Refer to the (remaining) two colors as blue and red and consider the following coloring:
        \begin{itemize}
            \item vertices $b_i$, $\kappa_i$, and $c_{i,j}$ are colored blue,
            \item vertices $r_i$, $\lambda_i$, and $c'_{i,j}$ are colored red,
            \item for every $\mathcal{V}_i$, a single vertex corresponding to the assignment $f$ restricted to $V_i$ is colored blue,
            while the rest are colored red.
        \end{itemize}
        Let $B$ and $R$ be the subsets of $V(G)$ colored blue and red respectively.
        We will prove that $G[B]$ and $G[R]$ have maximum degree at most $\Delta$.
        Regarding vertices $b_i$ and $r_i$, notice that each has at most $n_V + \sum_{i=1}^{n_C} n^i_\mathcal{F} \leq \Delta$ same colored neighbors.
        Regarding vertices $\kappa_i$ and $\lambda_i$, it holds that $|\mathcal{V}_i \cap B| = 1$ and $|\mathcal{V}_i \cap R| = \sqrt[4]{n} - 1$,
        thus each has degree equal to $\Delta$ in $G[B]$ and $G[R]$ respectively.
        For every $v \in \mathcal{V}_i$, it holds that it has at most $1 + \sum_{i=1}^{n_C} n^i_\mathcal{F} \leq \Delta$ same colored neighbors,
        due to either $\kappa_i$ plus the vertices $c_{i,j}$ in case it is colored blue or $\lambda_i$ plus the vertices $c'_{i,j}$ otherwise.
        Lastly, for each vertex $c_{i,j}$ corresponding to a clause set $C_{i,j}$,
        it holds that, since $f$ is an assignment where every clause of $\phi$ has exactly 1 True literal,
        and due to~\cref{lem:x_3_4_sat_partition}, exactly $|C_{i,j}|$ neighbors of $c_{i,j}$ in $\mathcal{V}$ are colored blue.
        It that case, the rest of its neighbors in $\mathcal{V}$ are colored red.
        Consequently, it follows that each $c_{i,j}$ and $c'_{i,j}$ has exactly $\Delta$ same colored neighbors.
    \end{proof}

    \begin{lemma}\label{lem:dc_vc_correctness_2}
        For any $\chid \geq 2$,
        if $G$ has a $(\chid, \Delta)$-coloring, then $\phi$ is a Yes instance of \SD.
    \end{lemma}
    
    \begin{proof}
        Let $q^* : V(G) \to [\chid]$ be a $(\chid,\Delta)$-coloring of $G$.
        We start with the following claim.

        \begin{claim}
            The restriction of $q^*$ to $V(G - P)$ is a $(2,\Delta)$-coloring of $G - P$.
        \end{claim}

        \begin{claimproof}
            If $\chid = 2$ this is true since $P = \emptyset$.
            Assume that $\chid \geq 3$, and let for $i \in [3,\chid]$, $G_i = G - \bigcup_{k=i+1}^{\chid} P_{k}$, while $G_{\chid} = G$.
            We will prove that if there exists a $(i,\Delta)$-coloring for $G_i$,
            then the restriction of the same coloring to $G_{i-1}$ is a $(i-1,\Delta)$-coloring of $G_{i-1}$,
            for $i \in [3,\chid]$.
            The statement holds for $\chid$:
            notice that, since $P_{\chid}$ contains $\chid \cdot \Delta + 1$ vertices,
            there exist at least $\Delta + 1$ vertices of the same color.
            In that case, since all those vertices are connected to any other vertex of $G$ not belonging to $P_{\chid}$,
            it follows that at most $\chid - 1$ colors are used in the rest of the graph,
            thus the restriction of $q^*$ to $V(G - P_{\chid})$ is a $(i - 1,\Delta)$-coloring of $G - P_{\chid}$.
            Now assume that there exists a $(i,\Delta)$-coloring for $G_i$, where $i \in [3,\chid-1]$.
            Then, notice that, since $P_{i}$ contains $i \cdot \Delta + 1$ vertices,
            there exist at least $\Delta + 1$ vertices of the same color.
            In that case, since all those vertices are connected to any other vertex of $G_i$ not belonging to $P_{i}$,
            it follows that at most $i - 1$ colors are used in the rest of the graph,
            thus the restriction of $q^*$ to $V(G_{i-1})$ is a $(i - 1,\Delta)$-coloring of $G_{i-1}$.
        \end{claimproof}
        Let $q : V(G-P) \to \braces{b,r}$ be the restriction of $q^*$ to $G-P$, where $q(v)$ denotes the (blue or red) color of every vertex $v$ of $G-P$.
        Assume without loss of generality that $q(b_1) = b$.
        \begin{itemize}
            \item We first prove that $q(b_i) = b$, for all $i \in [\Delta + 1]$.
            Assume that this is not the case, i.e. there exists $i$ such that $q(b_i) = r$.
            It holds that $r_j \in N(b_1) \cap N(b_i)$, for all $j \in [2\Delta + 1]$,
            and since every vertex $r_j$ is colored either blue or red,
            it follows that there exist at least $\Delta+1$ such vertices of the same color.
            In that case, either $b_1$ or $b_i$ has $\Delta+1$ same colored neighbors, contradiction.
            
            \item Since every vertex $r_i$, $\lambda_i$, and $c'_{i,j}$ has $\Delta+1$ blue colored neighbors, it follows that $q(r_i) = q(\lambda_i) = q(c'_{i,j}) = r$.
    
            \item Since every vertex $\kappa_i$ and $c_{i,j}$ has $\Delta+1$ red colored neighbors, it follows that $q(\kappa_i) = q(c_{i,j}) = b$.
            
            \item In that case, any $r_i$ has at most $n_V + \sum_{i=1}^{n_C} n^i_\mathcal{F} < \Delta$ same colored neighbors due to the vertices $\lambda_i$ and $c'_{i,j}$.
            Symmetrically, any $b_i$ has at most $n_V + \sum_{i=1}^{n_C} n^i_\mathcal{F} < \Delta$ same colored neighbors.
        \end{itemize}
    
        For every subset $\mathcal{V}_i$, it holds that exactly one vertex is colored blue:
        if all vertices were colored red, then $\lambda_i$ would have in total $\Delta - (\sqrt[4]{n} - 1) + \sqrt[4]{n} > \Delta$ red neighbors,
        while if more than one vertices, say $p > 1$, were blue then $\kappa_i$ would have $\Delta - 1 + p > \Delta$ blue neighbors.
    
        Now consider the assignment $h : V_1 \cup \ldots \cup V_{n_V} \to \braces{T,F}$ of the variables of $\phi$ depending on which vertex of $\mathcal{V}_i$ is colored blue.
        Let $C_i$, where $i \in [n_C]$ be a clause subset resulting from the partition due to~\cref{lem:x_3_4_sat_partition}.
        Let $f : C_i \to \braces{0,1,2,3}$ be the function assigning to each clause $c \in C_i$ the number of blue vertices $v \in \mathcal{V}$ such that
        $v \in \mathcal{V}_j$ corresponds to an assignment $V_j \to \braces{T,F}$ that satisfies $c$.
        Notice that since every clause contains 3 literals, while from every subset $\mathcal{V}_i$ exactly one vertex is colored blue,
        $f(c) \leq 3$ follows.
        It holds that $\sum_{c \in C_{i,j}} f(c) = |C_{i,j}|$, for any $j \in [n^i_\mathcal{F}]$:
        $c_{i,j}$ and $c'_{i,j}$ have the same neighborhood in $\mathcal{V}$ of size $|C_{i,j}| \cdot \frac{3 \sqrt[4]{n}}{2}$,
        and out of those vertices, at most $|C_{i,j}|$ may be blue while at most $|C_{i,j}| \cdot \frac{3 \sqrt[4]{n}}{2} - |C_{i,j}|$ may be red.
    
        Now consider $g : C_i \to \braces{0,1,2,3}$ to be the constant function $g \equiv 1$.
        Notice that $\sum_{c \in C_{i,j}} f(c) = \sum_{c \in C_{i,j}} g(c)$, for all $j \in [n^i_\mathcal{F}]$.
        Since $\braces{C_{i,1}, \ldots, C_{i, n^i_\mathcal{F}}}$ is a 4-detecting family, this implies that $f = g$.
        Thus, for every clause $c \in C_i$, it holds that $f(c) = 1$, meaning that there exists a single blue vertex in $\mathcal{V}$
        which corresponds to a partial assignment that satisfies $c$.
        Since $i = 1, \ldots, n_C$ was arbitrary, this concludes the proof that $h$ is a satisfying assignment for $\phi$.
    \end{proof}

    Therefore, in polynomial time, we can construct a graph $G$, with vertex cover number $\vc = \bO(n / \log n)$ due to~\cref{lem:dc_vc_size},
    such that, due to~\cref{lem:dc_vc_correctness_1,lem:dc_vc_correctness_2},
    deciding whether $G$ admits a $(\chid,\Delta)$-coloring
    is equivalent to deciding if there exists an assignment such that every clause of $\phi$ has exactly one True literal.
    In that case, assuming there exists a $\vc^{o(\vc)} n^{\bO(1)}$ algorithm for \DC,
    one could decide \SD{} in time
    \[
        \vc^{o(\vc)} n^{\bO(1)} = \parens*{\frac{n}{\log n}}^{o(n / \log n)} n^{\bO(1)} =
        2^{(\log n - \log \log n) o(n / \log n)} =
        2^{o(n)},
    \]
    which contradicts the ETH due to~\cref{thm:3_4_xsat_hardness}.
\end{proof}

\section{Conclusion}\label{sec:conclusion}
In this work we have examined in depth the complexity of \BDD{} and \DC{} under the perspective of parameterized complexity.
In particular, we have precisely determined the complexity of both problems parameterized by some of the most commonly used structural parameters.
As a direction for future research, we consider the question of whether we could obtain a $n^{o(\fvs)}$ lower bound for \BDD{} as well as
for \DC{} when $\chid = 2$, where $\fvs$ denotes the size of the minimum feedback vertex set of the input graph.

\bibliography{bibliography}

\end{document}